\documentclass[aps,superscriptaddress,11pt,twoside]{revtex4}

\usepackage{amsmath,latexsym,amssymb,verbatim,enumerate,graphicx,color,amsthm}

\usepackage{enumitem}
\usepackage{amsfonts}
\usepackage{mathtools}
\usepackage{hyperref}
\usepackage{bm}
\usepackage{standalone}

\usepackage{orcidlink}

\newtheorem{theorem}{Theorem}
\numberwithin{theorem}{section}
\newtheorem{corollary}[theorem]{Corollary}
\newtheorem{lemma}[theorem]{Lemma}
\newtheorem{proposition}[theorem]{Proposition}
\newtheorem{notation}[theorem]{Notation} 

\theoremstyle{definition}
\newtheorem{definition}[theorem]{Definition}
\newtheorem{remark}[theorem]{Remark}

\def\squareforqed{\hbox{\rlap{$\sqcap$}$\sqcup$}}
\def\qed{\ifmmode\squareforqed\else{\unskip\nobreak\hfil
\penalty50\hskip1em\null\nobreak\hfil\squareforqed
\parfillskip=0pt\finalhyphendemerits=0\endgraf}\fi}
\def\endenv{\ifmmode\;\else{\unskip\nobreak\hfil
\penalty50\hskip1em\null\nobreak\hfil\;
\parfillskip=0pt\finalhyphendemerits=0\endgraf}\fi}

\mathchardef\ordinarycolon\mathcode`\:
\mathcode`\:=\string"8000
\def\vcentcolon{\mathrel{\mathop\ordinarycolon}}
\begingroup \catcode`\:=\active
  \lowercase{\endgroup
  \let :\vcentcolon
  }

\newcommand{\nc}{\newcommand}
\nc{\rnc}{\renewcommand} 
\newcommand{\de}{\mathrm{d}}
\nc{\beq}{\begin{equation}}
\nc{\eeq}{{\end{equation}}} 
\newcommand{\so}{\mathrm{SO}}
\newcommand{\midd}{\textup{ s.t. }}

\nc{\diag}{\operatorname{diag}}
\DeclareMathOperator\End{End} 
\DeclareMathOperator\cProj{Pr} 
\DeclareMathOperator\cproj{pr} 
\DeclareMathOperator\cPoj{P} 

\def\s{\sigma}

\nc{\glneq}{{\raisebox{0.6ex}{$>$}  \hspace*{-1.8ex} \raisebox{-0.6ex}{$<$}}}
\nc{\gleq}{{\raisebox{0.6ex}{$\geq$}\hspace*{-1.8ex} \raisebox{-0.6ex}{$\leq$}}}

\nc{\RR}{{{\mathbb R}}}
\nc{\NN}{{{\mathbb N}}}
\nc{\ZZ}{{{\mathbb Z}}}
\nc{\QQ}{{{\mathbb Q}}}
\nc{\UU}{{{\mathbb U}}}

\nc{\cB}{{\cal B}}
\nc{\cR}{{\cal R}}

\nc{\vholder}[1]{\rule{0pt}{#1}}

\nc{\ob}[1]{#1}

\def\beq{\begin {equation}}
\def\eeq{\end {equation}}

\nc{\region}{\cS\cW}

\allowdisplaybreaks 

\usepackage{etoolbox}
\makeatletter
\patchcmd{\frontmatter@RRAP@format}{(}{}{}{}
\patchcmd{\frontmatter@RRAP@format}{)}{}{}{}
\renewcommand\Dated@name{}
\makeatother

\begin{document}

\title{Characterising the Haar measure on the $p$-adic rotation groups \protect\\ via inverse limits of measure spaces}

\author{Paolo Aniello\,\orcidlink{0000-0003-4298-8275}}
 \email{paolo.aniello@unina.it}
  \affiliation{Dipartimento di Fisica  ``Ettore Pancini'', Universit\`a di Napoli  ``Federico II'', Complesso Universitario di Monte S.~Angelo, via Cintia, I-80126 Napoli, Italy}
  \affiliation{Istituto Nazionale di Fisica Nucleare, Sezione di Napoli,\\
Complesso Universitario di Monte S.~Angelo, via Cintia, I-80126 Napoli, Italy}
\author{Sonia L'Innocente\,\orcidlink{0000-0002-9224-7451}}
 \email{sonia.linnocente@unicam.it}
  \affiliation{School of Science and Technology, University of Camerino,
              Via Madonna delle Carceri 9, I-62032 Camerino, Italy}
\author{Stefano Mancini\,\orcidlink{0000-0002-3797-3987}}
 \email{stefano.mancini@unicam.it}
 \affiliation{School of Science and Technology, University of Camerino,
              Via Madonna delle Carceri 9, I-62032 Camerino, Italy}
 \affiliation{Istituto Nazionale di Fisica Nucleare, Sezione di Perugia,\\ via A.~Pascoli, I-06123 Perugia, Italy}              
\author{Vincenzo Parisi\,\orcidlink{0000-0001-9563-6471}}             
 \email{vincenzo.parisi@unicam.it}   
  \affiliation{School of Science and Technology, University of Camerino,
              Via Madonna delle Carceri 9, I-62032 Camerino, Italy}
  \affiliation{Istituto Nazionale di Fisica Nucleare, Sezione di Perugia,\\ via A.~Pascoli, I-06123 Perugia, Italy}
\author{Ilaria Svampa\,\orcidlink{0000-0002-1389-0319}}
 \email{ilaria.svampa@unicam.it}
  \affiliation{School of Science and Technology, University of Camerino,
              Via Madonna delle Carceri 9, I-62032 Camerino, Italy}
  \affiliation{Istituto Nazionale di Fisica Nucleare, Sezione di Perugia,\\ via A.~Pascoli, I-06123 Perugia, Italy}
  \affiliation{Departament de F\'isica: Grup d'Informaci\'o Qu\`antica,
              Universitat Aut\`onoma de Barcelona, ES-08193 Bellaterra (Barcelona), Spain}
\author{Andreas Winter\,\orcidlink{0000-0001-6344-4870}}
 \email{andreas.winter@uab.cat}
 \affiliation{Departament de F\'isica: Grup d'Informaci\'o Qu\`antica,
              Universitat Aut\`onoma de Barcelona, ES-08193 Bellaterra (Barcelona), Spain}
 \affiliation{ICREA---Instituci\'o Catalana de la Recerca i Estudis Avan\c{c}ats, Pg. Llu\'is Companys 23, ES-08001 Barcelona, Spain}
\affiliation{Institute for Advanced Study, Technische Universit\"at M\"unchen,\\  Lichtenbergstra{\ss}e 2a, D-85748 Garching, Germany}


\begin{abstract}
We determine the Haar measure on the compact $p$-adic special orthogonal groups of rotations $\so(d)_p$ in dimension $d=2,3$, by exploiting the machinery of inverse limits of measure spaces, for every prime $p>2$. We characterise the groups $\so(d)_p$ as inverse limits of finite groups, of which we provide parametrisations and orders, together with an equivalent description through a multivariable Hensel lifting. Supplying these finite groups with their normalised counting measures, we get an inverse family of Haar measure spaces for each $\so(d)_p$. Finally, we constructively prove the existence of the so-called inverse limit measure of these inverse families, which is explicitly computable, and prove that it gives the Haar measure on $\so(d)_p$. Our results pave the way towards the study of the irreducible projective unitary representations of the $p$-adic rotation groups, with potential applications to the recently proposed $p$-adic quantum information theory. 
\end{abstract}

\maketitle

\textbf{Keywords:}{
		\textit{Haar measure; inverse/projective limit; profinite group; $p$-adic rotation group}.}


\section{Introduction}
\label{sec:intro}

Among the classical groups, the special orthogonal ones $\so(d)$ --- over $\RR$ --- are undoubtedly the best-known and the most studied. Notable cases are those with $d\in\{2,3\}$, for the variety of their applications in physics, chemistry and engineering, as well as for the easiness of visualising their action on Euclidean space. Besides those over the real numbers, special orthogonal groups over the $p$-adic fields $\QQ_p$ are remarkably interesting, for primes $p\geq2$. They form a multitude of locally compact symmetry groups, one for each non-trivial quadratic form on $\QQ_p^d$. Unlike the real case, definite quadratic forms over $\QQ_p$ (i.e., those representing the zero only trivially) exist only in dimensions two, three and four~\cite{serre}, and lead to compact, hence profinite, groups (unlike the indefinite forms). The classification of the $p$-adic special orthogonal groups, up to isomorphisms, is complete for indefinite quadratic forms in $d=2,3$ and for all the definite forms~\cite{our1st}. Among them, there exists a unique compact $p$-adic special orthogonal group of degree $3$, $\so(3)_p$, for every prime $p\geq2$. This can be thought of as the group of rotations in $\QQ_p^3$, and its geometric features have been explored in~\cite{our1st}. Along the way, it is natural to address the study of the abelian subgroups of rotations in $\QQ_p^3$ around a given axis: Unlike in the real scenario, there are several compact $p$-adic special orthogonal groups of rotations in $\QQ_p^2$, $\so(2)_{p,\kappa}$, where $\kappa\in\QQ_p^\ast/(\QQ_p^\ast)^2$.

The Haar measure on a locally compact group is of central importance in many fields of mathematics such as harmonic analysis and representation theory, among others. Specifically, all irreducible unitary representations of a compact group occur (and can be studied) as sub-representations of the regular representation, according to the celebrated Peter-Weyl theorem. In turn, the Haar measure plays a fundamental role in the study of the regular representation and, more generally, of the irreducible projective unitary representations. On the other hand, the symmetry group $\so(3)_p$ has an intriguing role in $p$-adic quantum mechanics~\cite{Volovich1} (see also ~\cite{Volovich2,VVZ}): According to Volovich's view of a $p$-adic quantum system, the irreducible projective unitary representations of $\so(3)_p$ can be interpreted as a theory of $p$-adic angular momentum and spin~\cite{our2nd}. 
More specifically, in the prospect of a $p$-adic formulation of quantum computation and information theory, the $p$-adic qubit arises as a two-dimensional irreducible representation of $\so(3)_p$ (see also~\cite{aniello2023} for an alternative `purely $p$-adic' approach).
It is precisely for this reason that the present contribution is devoted to the study of the Haar measure on $\so(3)_p$, and along the way on $\so(2)_{p,\kappa}$. 

It is worth recalling that the mathematician V. S. Varadarajan dedicated part of his research to the investigation of a quantum theory over the field of $p$-adic numbers, also motivated~\cite{VaradarajanCit} by Dirac's mathematical modus operandi~\cite{Dirac}: 
The most powerful method for progressing in modern physics is to develop more and more advanced mathematical tools (such as non-Euclidean geometry or non-commutative algebra) to face new challenges, and, only after that the new mathematical structures have shown to be both consistent and effective, one can proceed to the interpretation of these structures as suitable physical entities.

In fact, during his scientific activity, Varadarajan provided important contributions to the development of Volovich's ideas about a $p$-adic (non-Archimedean) spacetime at a sub-Planck scale, by studying, in particular, the $p$-adic Galilean and Poincar\'e groups and their representations, for the structure and classification of elementary particles~\cite{Varadarajan1,Varadarajan2}, just to cite a few. 

Inspired by the profound ideas of Volovich and Varadarajan, and having in mind, in particular, applications to quantum information theory over the field of $p$-adic numbers, in the present contribution we go a step forward in the study of the compact special orthogonal groups on $\QQ_p^2$ and $\QQ_p^3$.

The Haar measure on $\so(2)_{p,\kappa}$ has already been investigated by means of a Lie-group-theoretical approach, relying on the adoption of a suitable atlas of local charts~\cite{our3rd}. Moreover, the Haar integral on $\so(3)_p$ can be expressed through a certain lifting, which involves a topological and group relation between $\so(3)_p$ and the multiplicative group of $p$-adic quaternions. In this paper, we will construct the Haar measure on $\so(2)_{p,\kappa}$ and $\so(3)_p$ through another universal approach: That of inverse limit of measure spaces. 

To date, the direct product measure on an infinite product of measurable spaces is standard measure theory~\cite{vonNeumann,Halmos}. Product measures naturally generalise to the concept of inverse limit of measure spaces, first introduced in~\cite{Bochner} by Bochner. Sufficient conditions for the existence of an inverse limit measure for an inverse family of measure spaces have subsequently been investigated (see~\cite{Choksi}, to cite one of the first). On the other hand, the left (resp.\ right) Haar measure --- conceived in~\cite{seminalHaar} --- is known to exist and be essentially unique on any locally compact group; left and right Haar measures coincide on compact groups. There exists an essentially unique Haar measure on any profinite group, a fact that can also be proven by an inverse limit reasoning, as argued by Fried and Jarden in~\cite{profinitem}. Indeed, for every locally compact group $G$, the inverse limit of left (resp.\ right.) Haar measures on a suitable inverse family of quotient groups is proven to be the left (resp.\ right.) Haar measure on the inverse limit group $G$~\cite{BourInt}. In the present work, we give a concrete and workable realisation of this abstract result, constructing the Haar measure on the groups $SO(2)_{p,\kappa}$ and $SO(3)_p$ as the inverse limit of counting measures on their quotient groups modulo $p^n$, $n\in\NN$. In fact, we consider an inverse family of measure spaces, which enables us to construct a $\sigma$-algebra on each $p$-adic rotation group which is shown to be Borel, and an inverse limit measure on it which is shown to satisfy all the axioms of a Haar measure. Once calculated the cardinality of the quotient groups modulo $p^n$, the Haar measure can be explicitly evaluated.

The rest of the paper has the following structure. After introducing the background notions concerning measure theory, inverse limits and $p$-adic numbers (Section~\ref{sec:basic}), we give an overview of the main features we know about compact $p$-adic special orthogonal groups of rotations in $\QQ_p^2$ and $\QQ_p^3$ (Subsection~II.A
), which are preparatory for this work. We will assume $p>2$, and briefly describe the scenario for the only even prime $p=2$ at the end, since at times it deviates from the case of odd primes and therefore it is more laborious to deal with. In Section~\ref{sec:invlimtopgru} we explicitly provide the characterisation of $\so(2)_{p,\kappa}$ and $\so(3)_p$ as inverse limits of inverse families of suitable discrete finite groups. We show a possible parametrisation of the latter, through which we can compute their orders; we also give an equivalent description of these finite groups, through a lifting of roots \`a la Hensel of the system of special orthogonal conditions. In Section~\ref{sec:invlimme}, we construct the Haar measure on $\so(2)_{p,\kappa}$ and $\so(3)_p$ as an inverse limit of Haar measure spaces, and evaluate it on every clopen ball in a topology base. Section~\ref{sec:conclusion} concerns the discussion of our main results, conclusions, and prospects.

\section{Basic notions}\label{sec:basic}
In this section, we remind notions that are relevant to our results, starting from the concept of inverse family and limit, moving to basic elements of measure theory, and concluding with $p$-adic special orthogonal groups.

\medskip

Let $(I,\leq)$ be a \emph{(right-)directed partially ordered set}. This is a non-empty set $I$ supplied with a partial order (i.e. a reflexive, transitive and antisymmetric binary relation) $\leq$, such that any finite subset of $I$ has upper bounds in $I$. We first recall the definition of inverse family and inverse limit of sets~\cite{BourSet},  and topological groups~\cite{BourTop} (see~\cite{profinite} for a more categorical approach).
\begin{definition}\label{def:invlimsetgroup}
    Let $\{X_i\}_{i\in I}$ be a family of sets (resp.\ topological groups), and $\{f_{ij}\colon X_j\rightarrow X_i\}_{i\leq j,\ i,j\in I}$ a family of maps (resp.\ continuous group homomorphisms) such that
\begin{enumerate}
    \item $f_{ii}$ is the identity map on $X_i$, for every $i \in I$;
    \item $f_{ik}=f_{ij}\circ f_{jk}$, for every $i\leq j\leq k$,\ \ $i,j,k\in I$.
\end{enumerate}
We call $\big\{\{X_i\}_{i\in I}, \{f_{ij}\colon X_j\rightarrow X_i\}_{i\leq j,\ i,j\in I}\big\}\equiv \{X_i,f_{ij}\}_I$ an \emph{inverse family of sets} (resp.\ \emph{of topological groups}). Let now $\prod\limits_{i\in I}X_i$ be the Cartesian product of the family of sets $\{X_i\}_{i\in I}$.
The \emph{inverse} (or \emph{projective}) \emph{limit} of the inverse family of sets $\{X_i,f_{ij}\}_I$ is 
\begin{align}
    &\varprojlim\{X_i,f_{ij}\}_I\coloneqq\left\{x\in \prod_{i\in I}X_i\midd \cProj_i(x)=f_{ij}\circ \cProj_j(x),\textup{ for every }i\leq j,\ i,j\in I\right\}\notag\\
    &= \left\{(x_i)_{i\in I}\in \prod_{i\in I}X_i\midd x_i=f_{ij}(x_j),\textup{ for every }i\leq j,\ i,j\in I\right\}\subseteq \prod_{i\in I}X_i,\label{invlimset}
\end{align}
where $\cProj_i\colon\prod\limits_{i\in I}X_i\rightarrow X_i,\ x=(x_i)_{i\in I}\mapsto x_i$ is the canonical projection on the $i$-th component. 

The \emph{inverse limit} of the inverse family of topological groups $\{X_i,f_{ij}\}_I$ is the subgroup of the direct product group $\prod\limits_{i\in I}X_i$ as in~\eqref{invlimset}, endowed with the coarsest topology for which all $\cProj_i$ are continuous ($i\in I$), coinciding with the topology induced by the product topology of $\prod\limits_{i\in I}X_i$.
\end{definition}

By abuse of notation, we denote by $\cProj_i$ also the restriction of the canonical projection on $\varprojlim\{X_i,f_{ij}\}_I$
. The inverse limit of an inverse family of sets or topological groups always exists in the same category (this is not true in the broader setting of an inverse family in an arbitrary category). In any category, the definition of inverse limit is given by means of a universal property, so that if an inverse limit exists, it is necessarily unique: If $X$ and $X'$ are two inverse limits of the same inverse family, with projection maps $\{\cProj_i\}_i$ and $\{\cProj'_i\}_i$ respectively, then there exists a \emph{unique} isomorphism $f\colon X\rightarrow X'$ such that $\cProj'_i\circ f = \cProj_i$ for every $i\in I$. 

Now we recall the notion of profinite group (cf.~\cite{serre2,profinite} for a thorough discussion).
\begin{definition}\label{def.2}
    A topological group $G$ is said to be \emph{profinite} if it is the inverse limit of an inverse family of finite groups, each given the discrete topology.
\end{definition}
The next result gives a necessary and sufficient condition for a group to be profinite (see Proposition~0 in~\cite{serre2}, or Theorem~1.1.12 in~\cite{profinite}).
\begin{proposition}\label{prop:profequiv}
    A topological group is profinite if and only if its topology  is (Hausdorff) compact and totally disconnected.
\end{proposition}
Our main aim is to construct an invariant measure on profinite groups, in particular, on compact $p$-adic special orthogonal groups. Therefore, for the reader's convenience, here we collect some basic notions about measure spaces (see, e.g.,~\cite{Halmos,Folland,hewitt1979,Niklas} for both set-theoretical and analytic approaches). We follow the notation and terminology of~\cite{Halmos}.

Let $X$ be a set. By a \emph{ring of sets} $R$ of $X$, we mean a non-empty family of subsets of $X$ closed under finite union and set difference (i.e., relative complementation). We call a family $A$ of subsets of $X$ an \emph{algebra of sets of $X$}, if it is closed under finite union and complementation, i.e., if it is a ring of sets of $X$ and contains $X$ itself. A \emph{$\sigma$-ring} $S$ of $X$ (resp.\ a \emph{$\sigma$-algebra} $\Sigma$ of $X$) is meant to be a ring (resp.\ an algebra) of $X$ closed under countable union, i.e.,  
if $E_\lambda\in S$ (resp.\ $E_\lambda\in\Sigma$), for every $\lambda$ in a countable index set, then $\bigcup_\lambda E_\lambda \in S$ (resp.\ $\bigcup_\lambda E_\lambda \in \Sigma$). If $M$ is a collection of subsets of $X$, then there is a unique smallest $\sigma$-algebra $\Sigma(M)$ containing $M$, namely, the so-called \emph{$\sigma$-algebra generated by $M$}. In particular, if $X$ is a topological space, the $\sigma$-algebra generated by the family of its open sets is called the \emph{Borel $\sigma$-algebra} of $X$, and is usually denoted by $\cB(X)$.

A \emph{measure} $\mu$ on a ring of sets $R$ of $X$ (shortly, a measure $\mu$ on $X$) is a non-negative map $\mu\colon R\rightarrow[0,+\infty]$ such that
\begin{itemize}
    \item[i.] $\mu(\emptyset)=0$;
    \item[ii.] $\mu\left(\bigcup_{\lambda}E_\lambda\right)=\sum_{\lambda}\mu(E_\lambda)$ \quad\mbox{(\emph{$\sigma$-additivity})},\\
    for every countable family $\{E_\lambda\}_{\lambda}$ of pairwise disjoint sets in $R$ such that $\bigcup_{\lambda}E_\lambda\in R$.
\end{itemize}
 We say that $\mu$ is a \emph{probability measure} if it is a measure taking values in $[0,1]$; we say that $\mu$ is a \emph{Borel measure} if it is defined on a Borel $\sigma$-algebra. By a \emph{measurable space} we mean a pair $(X,S)$ where $X$ is a set and $S$ is a covering $\sigma$-ring of $X$; in particular, a \emph{Borel measurable space} is a measurable space with a Borel $\sigma$-algebra as $\sigma$-ring. Moreover, we call $(X,S,\mu)$ a (\emph{Borel}) \emph{measure space} if $(X,S)$ is a (Borel) measurable space and $\mu$ is a (Borel) measure on $S$. 

In our later derivations, Borel measures on locally compact Hausdorff (LCH) spaces will play a prominent role. We recall that a \emph{Radon measure} $\mu$ on a LCH space $X$ is a Borel measure that is finite on all compact sets, \emph{outer regular} on all Borel sets, i.e., 
\beq
 \mu(E)=\inf\left\{\mu(U)\midd E\subseteq U,\,U\textup{ open}\right\},\quad E\in\mathcal{B}(X),
\eeq
and \emph{inner regular} on all open sets~\cite{Folland99}:
\beq
      \mu(U)=\sup\left\{\mu(K)\midd K\subseteq U,\,K\textup{ compact}\right\},\quad \mbox{$U$ open set in $G$}.
\eeq
We come now to the main object of our investigations:
\begin{definition}\label{def:Haarme} 
Let $G$ be a LCH group. A \emph{left} (resp.\ \emph{right}) \emph{Haar measure} on $G$ is a left (resp.\ \emph{right})-invariant \emph{Radon measure} on $G$; namely, a Radon measure $\mu$ on $G$ for which the condition 
\beq \label{eq:condleftrightinv}
\mu(gE)=\mu(E)\quad (\text{resp.}\, \mu(E g)=\mu(E)),
\eeq
holds for every Borel set $E\subseteq G$ and $g\in G$.
\end{definition}

It is a well-known result (see, e.g., Theorem~2.10 and~2.20 in~\cite{Folland}) that any locally compact group $G$ admits an \emph{essentially uniquely defined} left (resp.\ right) Haar measure, i.e., if $\mu$ and $\nu$ are left (resp.\ right) Haar measures on $G$, then there exists  $c\in \mathbb{R}_{>0}$ such that $\mu=c\nu$. Moreover, the left and right Haar measures coincide on every compact group. 

Since our emphasis is on profinite groups, we shall exploit a machinery to construct their Haar measure based on suitable inverse families of measure spaces, defined as follows
(see also Definition~2 in~\cite{Choksi}).
\begin{definition}\label{def:invemeaspac}
    An \emph{inverse family of measure spaces} is a family $\left\{(X_i,S_i,\mu_i),f_{ij}\right\}_I$ of measure spaces such that
    \begin{enumerate}
        \item $\{X_i,f_{ij}\}_I$ is an inverse family of sets;
        \item $f_{ij}$ is \emph{measure preserving}, i.e., for $i<j$, $f_{ij}^{-1}(S_i)\subseteq S_j$ and for $E_i\in S_i$, $\mu_i(E_i) = \mu_j\left(f_{ij}^{-1}(E_i)\right)$.
    \end{enumerate}
\end{definition}

\bigskip

To conclude this section, we now provide a brief account on $p$-adic numbers~\cite{serre,VVZ,Gouvea,Cassels,igusa2000,rooij78}, and reserve the next subsection to recall the main features of $p$-adic special orthogonal groups (see~\cite{our1st,our2nd} for an exhaustive discussion).

For this discussion, let $p\geq2$ be a prime number. Once fixed $p$, we recall that any $x\in\QQ^\ast\equiv\QQ\setminus\{0\}$ can be uniquely written as $x=p^{\nu_p(x)}\frac{a}{b}$, where $\nu_p(x),a,b\in\ZZ$, and $a,b$ are such that $p\nmid ab$; one can then define the so-called \emph{$p$-adic absolute value} on $\QQ$ by setting
\beq\label{eq.3}
|x|_p = p^{-\nu_p(x)}\textup{ for } x\in\QQ^\ast,\qquad |0|_p=0.
\eeq
The space $\QQ_p$ of \emph{$p$-adic numbers} is then defined as  the metric completion of $\QQ$ with respect to (the metric associated with) $|\,\cdot\,|_p$. It is not difficult to show~\cite{Folland} that every $x\in\QQ_p^\ast=\QQ_p\setminus\{0\}$ is (uniquely) represented as a suitable Laurent series of $p$, i.e., 
\beq \label{eq:Lauserp}
x=\sum_{n\geq n_0}x_np^n,
\eeq 
where $n_0\in\ZZ$, and $x_n\in[0,p-1]\cap\ZZ$ for every $n\in\ZZ_{\geq n_0}$. It follows that $\QQ_p$ is in bijective correspondence with the ring $\ZZ/p\ZZ((p))$ 
of the formal Laurent series of $p$ with coefficients in the finite field $\mathbb{F}_p = \ZZ/p\ZZ$, which is in particular the cyclic group of integers modulo $p$. Also, $\QQ_p$ is a field, once supplied with the addition and multiplication operations, defined as for formal series but \textquotedblleft carrying\textquotedblright\ the quotient by $p$ of the coefficient of $p^n$ to the coefficient of $p^{n+1}$. Moreover, the continuous extension of the $p$-adic absolute value $|\cdot|_p$ from $\QQ$ to $\QQ_p$ --- which we still denote with the same symbol --- is given by
\beq
|x|_p=\Bigg|\sum_{n\geq n_0}x_np^n\Bigg|_p=p^{-|n_0|},\quad \forall x\in\QQ_p^\ast,
\eeq
if $x_{n_0}\neq0$. The norm $|\cdot|_p$ satisfies the so-called \emph{strong triangle inequality}, i.e., $|x+y|_p\leq\max\{|x|_p,|y|_p\}$ for every $x,y\in\QQ_p$; then, $(\QQ_p,\,|\cdot|_p)$ is said to be a \emph{non-Archimedean} valued field, and the induced distance is called \emph{ultrametric}. 

In $\QQ_p$, we can single out the so-called \emph{valuation ring} --- with respect to  $|\cdot|_p$ --- of the non-Archimedean field $\QQ_p$; it is the ring
of \emph{$p$-adic integers} $\ZZ_p\coloneqq\{x\in\QQ_p \midd |x|_p\leq 1\}=\{\sum_{n=0}^{+\infty}x_np^n \midd x_n\in[0,p-1]\cap\ZZ\}$,
and is a subring of $\QQ_p$. Also in this case, one has a natural bijective correspondence between $\ZZ_p$ and the ring $\ZZ/p\ZZ[[p]]$ of the formal power series of $p$ with coefficients in $\ZZ/p\ZZ$ (however, the two rings are not isomorphic).
The set $\mathfrak{P}_p\coloneqq\{x\in\QQ_p\midd |x|_p< 1\}=p\ZZ_p\subset\ZZ_p$
is a maximal ideal in $\ZZ_p$ (actually, its unique maximal ideal), and is called the \emph{valuation ideal} of $\QQ_p$ with respect to $|\cdot|_p$. The elements in the set $\UU_p\coloneqq \ZZ_p\setminus\mathfrak{P}_p$  are \emph{precisely} the invertible elements in $\ZZ_p$. Moreover, as is easily shown, $\UU_p$ closes a group, usually referred to as the group of \emph{$p$-adic units}. Any $x\in\QQ_p^\ast$ can be uniquely written like $x=p^{\nu_p(x)}u$, where $u\in\UU_p$ and, by abuse of notation, $\nu_p(x)\in\ZZ$ is the \emph{$p$-adic valuation} of $x$. In particular, the $p$-adic absolute value of $x$ in $\QQ_p$ is then expressed as
\beq 
|x|_p=\begin{cases}
    0, & x=0,\\ 
    p^{-\nu_p(x)}, &  0\neq x = p^{\nu_p(x)}u,\, u\in\UU_p.
\end{cases}
\eeq 
The quotient $\ZZ_p/\mathfrak{P}_p$ is called the \emph{residue class field} of $\QQ_p$ with respect to $|\cdot|_p$ (recall that the quotient of a ring by a maximal ideal is always a field); specifically, $\ZZ_p/\mathfrak{P}_p=\ZZ_p/p\ZZ_p$ is isomorphic to the finite field $\mathbb{F}_p=\ZZ/p\ZZ$.

On $\QQ_p$ we can consider the \emph{(ultra)metric topology} generated by the base of open discs (with respect to the $p$-adic metric induced by the $p$-adic absolute value)
\begin{align}
    D_k(c)&=\{x\in \QQ_p\midd |x-c|_p\leq p^k\} = \{x\in\QQ_p\midd |x-c|_p< p^{k+1}\},
\end{align}
where, in principle, $c\in\QQ_p$ and $k\in\ZZ$. Actually, $\QQ$ is a countable dense subset of the metric space $\QQ_p$, and, hence, $\QQ_p$ has a countable base of open discs with centres $c\in \QQ$. Also, $k$ ranges in a subset of $\ZZ$ without minimum, as all other discs of larger radius will be given by unions of those discs of smaller radius (two discs are either disjoint or one is contained in the other). Hence a base for the (ultra)metric topology on $\QQ_p$ is
\beq \label{eq:topQpbase}
\left\{D_k(c)\midd c\in \QQ,\, k\in\ZZ_{<0}\right\},
\eeq 
and $\QQ_p$ is second countable. Moreover, any disc in $\QQ_p$ (and, a fortiori, in $\QQ_p^d$,  $d\in \NN\equiv\ZZ_{>0}$, equipped with the product topology, coinciding with the topology induced by the product metric) is a \emph{clopen set} --- namely, both open and closed --- each of its elements is a centre, and two discs are either disjoint or one is contained in the other (see Chapter~2 in~\cite{rooij78}, and also Lemma~1 and Corollaries~1,2,3 at p.~6 of~\cite{VVZ}). Then, $\QQ_p$ (and any $\QQ_p^d$), turns out to be a LCH totally disconnected space. In analogy with the real setting, a set $K\subset\QQ_p^d$ is \emph{compact} if and only if it is closed and bounded~\cite{VVZ,rooij78}; e.g., $\ZZ_p$ is compact in $\QQ_p$, since it coincides with the closed disc $D_0(0)$ which has $p$-adic norm bounded by $1$. In fact, each subset $p^m\ZZ_p$, $m\in\ZZ$, of $\QQ_p$ is compact, because $p^m\ZZ_p = \{x\in\QQ_p\midd |x|_p\leq p^{-m}\} = D_{-m}(0)$.

Likewise every normed (or valued) field, $\QQ_p$ is a topological field. The groups $p^m\ZZ_p$, for $m\in\ZZ$, exhaust all the proper closed subgroups of $\QQ_p$. They are topological groups, once given the subspace topology, a base for which is $\{D_k(c)\midd c\in p^m\ZZ,\, k\in J\}$, since $p^m\ZZ$ is a countable dense subset of $p^m\ZZ_p$, and where $J$ is a subset of $\ZZ_{\leq -m}$ without a minimum. In fact, the groups $p^m\ZZ_p$ are profinite by Proposition~\ref{prop:profequiv}. 

The topological (additive) group $\QQ_p$, as well as all its proper closed subgroups, can be characterised as the inverse limit of a suitable inverse family, according to the following result (for a proof, see Appendix~\ref{sec:appinvlimQp}).
\begin{proposition}\label{prop:invlimQpappen}
For every prime $p\geq2$, the following isomorphisms of topological groups hold:
    \beq
\QQ_p\simeq \varprojlim\left\{\QQ_p/p^n\ZZ_p,\,\phi_{nl}\right\}_{\NN}, \qquad p^m\ZZ_p\simeq \varprojlim\left\{p^m\ZZ_p/p^n\ZZ_p,\,\phi_{nl}\right\}_{\ZZ_{> m}},
\eeq 
for every $m\in\ZZ$, where $\QQ_p$ is equipped with the $p$-adic ultrametric topology, $p^m\ZZ_p$ with the subspace topology, and quotient groups with the quotient topologies (coinciding with the discrete topologies); the continuous group homomorphisms $\phi_{nl}$ are defined as
\beq \label{eq:homoaddno}
x_{(l)} + p^l\ZZ_p\mapsto x_{(l)}+p^n\ZZ_p,
\eeq 
for every $l\geq n>m$. 
\end{proposition}
By abuse of notation, $\phi_{nl}$ denotes the inverse limit homomorphism on $\QQ_p/p^l\ZZ_p$ to $\QQ_p/p^n\ZZ_p$, as well as that on the restriction $p^m\ZZ_p/p^l\ZZ_p$ to $p^m\ZZ_p/p^n\ZZ_p$. In a similar fashion, in comparison with the map~\eqref{eq:absuprib}, we also denote by $\cproj_n$, $n>m$, the projection (continuous group homomorphism)
\beq \label{eq:homocontproj}
\cproj_n\colon p^m\ZZ_p\rightarrow p^m\ZZ_p/p^n\ZZ_p,\quad x=\left(x_{(n)}+p^n\ZZ_p\right)_n\mapsto x_{(n)} +p^n\ZZ_p.
\eeq 
The map $\cproj_n$ is like the restriction of $\cProj_i$ in Definition~\ref{def:invlimsetgroup} to the inverse limit subspace of the product space, composed with the group isomorphism from $p^m\ZZ_p$ to $\varprojlim\left\{p^m\ZZ_p/p^n\ZZ_p,\,\phi_{nl}\right\}_{\ZZ_{> m}}$. 

The notation $x_{(n)}+p^n\ZZ_p$ for a coset in $\QQ_p/p^n\ZZ_p$ is often replaced by the canonical projection of $x_{(n)}\in\QQ_p$ into the quotient, $x_{(n)}\bmod p^n\ZZ_p$. In a similar fashion, a coset $z+p^n\ZZ$ in the cyclic group $\ZZ/p^n\ZZ$ is usually rewritten as $z\bmod p^n$ (which is a shorthand notation for $z\bmod p^n\ZZ$).

We stress that $\ZZ_p$ is a ring, however $p^m\ZZ_p$ is a non-unital ring for $m\in\NN$ and is just an additive group for $m\in\ZZ_{<0}$. 
\begin{remark}\label{rem:invlimZpeq}Specialising Proposition~\ref{prop:invlimQpappen} for $m=0$ provides 
\beq 
\ZZ_p\simeq \varprojlim\left\{\ZZ_p/p^n\ZZ_p,\,\phi_{nl}\right\}_{\NN},
\eeq
and observing that $\ZZ_p/p^n\ZZ_p\simeq \ZZ/p^n\ZZ$, we recover the well-known inverse limit characterisation of $\ZZ_p$ (e.g., as $\ZZ_p$ is defined in~\cite{serre}):
\beq \label{eq:invlimZpeq}
\ZZ_p\simeq \Big\{x=\left(x_{(n)}+p^n\ZZ\right)_n \in\prod_{n\in\NN}\left(\ZZ/p^n\ZZ\right)\midd x_{(n+1)}\equiv x_{(n)}\bmod p^n\Big\}.
\eeq 
This is also known from the inverse limit characterisation of the metric completion of certain topological abelian groups as in~\cite{metricompladic}. Indeed, one considers the topological (additive) group $\ZZ$ with $p$-adic ultrametric topology, which can be defined given a sequence of subgroups $p^n\ZZ$, forming a fundamental system of neighbourhoods for $0\in\ZZ$. Therefore, the metric completion of $\ZZ$ is $\ZZ_p\simeq\varprojlim \ZZ/p^n\ZZ$.

Note that $p^n\ZZ_p$ is an ideal of $\ZZ_p$, and $p^n\ZZ$ is an ideal of $\ZZ$, for every $n\in\NN$. Actually, Eq.~\eqref{eq:invlimZpeq} represents an inverse limit of rings: $\{\ZZ/p^n\ZZ\}_{n\in\NN}$ is a family of topological rings (with discrete topology), $\left\{\Phi_{nl}\colon \ZZ/p^l\ZZ\rightarrow \ZZ/p^n\ZZ,\, \Phi_{nl}\left(x_{(l)}\right) \coloneqq x_{(l)}\mod p^n\right\}$ is a family of continuous ring homomorphisms, and a topological ring isomorphism is defined similarly to Eq.~\eqref{eq:fisotop}. We denote by $\cPoj_n$ the canonical projection 
\beq 
\cPoj_n\colon \ZZ_p\rightarrow \ZZ/p^n\ZZ, \quad x\mapsto x_{(n)}\mod p^n,
\eeq 
which is a (continuous) ring homomorphism, for every $n\in\NN$:
\begin{align}
\cPoj_n(x+y) &= \cPoj_n\left(\left(x_{(n)}+p^n\ZZ\right)_n+\left(y_{(n)}+p^n\ZZ\right)_n\right)=\cPoj_n\left(\left(x_{(n)}+y_{(n)}+p^n\ZZ\right)_n\right)=x_{(n)}+y_{(n)}+p^n\ZZ\notag \\&=\cPoj_n\left(\left(x_{(n)}+p^n\ZZ\right)_n\right)+\cPoj_n\left(\left(y_{(n)}+p^n\ZZ\right)_n\right)=\cPoj_n(x)+\cPoj_n(y);\\
\cPoj_n(xy) &= \cPoj_n\left(\left(x_{(n)}+p^n\ZZ\right)_n\left(y_{(n)}+p^n\ZZ\right)_n\right)=\cPoj_n\left(\left(x_{(n)}y_{(n)}+p^n\ZZ\right)_n\right)=x_{(n)}y_{(n)}+p^n\ZZ\notag \\&=\cPoj_n\left(\left(x_{(n)}+p^n\ZZ\right)_n\right)\cPoj_n\left(\left(y_{(n)}+p^n\ZZ\right)_n\right)=\cPoj_n(x)\cPoj_n(y).
\end{align}
\end{remark}

\subsection{\texorpdfstring{$p$}{Lg}-Adic special orthogonal groups}\label{subsec:basicgroup}
$p$-Adic orthogonal groups can be defined as those groups consisting of linear transformations which preserve a \emph{quadratic form} $Q$ on $\QQ_p^d$~\cite{serre,Cassels,Lam}. For our purposes, a quadratic form is a homogeneous function on the $d$-dimensional $\QQ_p$-vector space $V$, written as $Q(\mathbf{x})=\sum_{i,j=1}^da_{ij}x_ix_j=\mathbf{x}^\top A\mathbf{x}$ where $\mathbf{x}=\sum_{i=1}^dx_i\mathbf{e}_i\in V$ is a (column) vector, $(\mathbf{e}_1,\dots,\mathbf{e}_d)$ is a basis of $V$, and $A$ is an $d\times d$ matrix. The latter is said to be the matrix representation of $Q$ with respect to the basis $(\mathbf{e}_1,\dots,\mathbf{e}_d)$. Throughout this work, we will assume quadratic forms to be \emph{non-degenerate}, i.e., their matrix representations to have maximum rank $d$. Since the characteristic of $\QQ_p$ is different from $2$ --- indeed it is $0$ because, likewise $\RR$, the image of $\ZZ$ in $\QQ_p$ is an integral domain isomorphic to $\ZZ$ --- a quadratic form $Q$ naturally induces a \emph{symmetric bilinear form} $b$ on $V$ given by
\beq\label{eq:sclprodQbil} b(\mathbf{x},\mathbf{y})\coloneqq \frac{1}{2}\big(Q(\mathbf{x}+\mathbf{y})-Q(\mathbf{x})-Q(\mathbf{y})\big), \quad \mbox{for $\mathbf{x},\mathbf{y}\in V$},
\eeq
and, vice versa, a symmetric bilinear form $b$ induces a quadratic form via $Q(\mathbf{x}) = b(\mathbf{x}, \mathbf{x})$. Therefore, we have a bijective correspondence between quadratic forms and symmetric bilinear forms on $V$. The scalar product~\eqref{eq:sclprodQbil} between $\mathbf{x}$ and $\mathbf{y}$ can be written in matrix form as $b(\mathbf{x},\mathbf{y}) = \mathbf{x}^\top A \mathbf{y}$.

\begin{definition}
The \emph{orthogonal group} on $V$ with respect to $Q$ is defined as the set of linear maps on $V$ that are symmetries of the quadratic (equivalently, of the symmetric bilinear) form $Q$:
\begin{align}
    \mathrm{O}(V,Q)&=\left\{L\in\End(V)\midd Q(L\mathbf{x})=Q(\mathbf{x})\ \forall \mathbf{x}\in V \right\}\\
    &=\left\{L\in\End(V)\midd b(L\mathbf{x},L\mathbf{y})=b(\mathbf{x},\mathbf{y})\ \forall \mathbf{x},\mathbf{y}\in V \right\}\\
    &\simeq\left\{L\in \mathsf{M}_{d\times d}(\QQ_p)\midd L^\top AL=A\right\},
\end{align}
where the latter group isomorphism is under the identification of $V$ with $\QQ_p^d$ via the basis $(\mathbf{e}_1,\dots,\mathbf{e}_d)$, $\mathbf{x}\leftrightarrow \begin{pmatrix}
    x_1\\\vdots\\x_d
\end{pmatrix}$, which turns the linear maps on $V$ into $d\times d$ matrices. The subgroup of $\mathrm{O}(V\simeq\QQ_p^d,Q)$ consisting of matrices $L$ with unit determinant is the special orthogonal group, i.e.,
\beq 
\so(\QQ_p^d,Q)=\left\{L\in \mathsf{M}_{d\times d}(\QQ_p)\midd L^\top AL=A,\ \det(L)=1\right\}.
\eeq 
\end{definition}

Any $\so(\QQ_p^d,Q)$ (and, more generally, $\mathrm{O}(\QQ_p^d,Q)$) is a topological group --- the multiplication map has polynomial components, the inversion map is continuous by Cramer's rule
--- once supplied with the subspace topology of $\mathsf{M}_{d\times d}(\QQ_p)\simeq \QQ_p^{d^2}$. This coincides with the topology induced by its $p$-adic product metric, i.e., the metric induced by the $p$-adic norm $\|L\|_p=\max_{i,j=1,\dots,d}|\ell_{ij}|_p$, where $L=\begin{pmatrix}
\ell_{ij}
\end{pmatrix}_{ij}\in \so(\QQ_p^d,Q)$. In other words, the topology considered on $\so(\QQ_p^d,Q)$ is the \emph{ultrametric topology} generated by the base of clopen balls 
\beq 
    B_k(L_0)=\{L\in\so(\QQ_p^d,Q)\midd \|L-L_0\|_p\leq p^k\}=\{L\in\so(\QQ_p^d,Q)\midd \|L-L_0\|_p< p^{k+1}\},
\eeq 
where $k$ ranges in a suitable subset of $\ZZ$ without minimum. 
All the topological properties stated after Eq.~\eqref{eq:topQpbase} apply to $\so(\QQ_p^d,Q)$, as subspaces of $\mathsf{M}_{d\times d}(\QQ_p)\simeq \QQ_p^{d^2}$. In particular, the groups $\so(\QQ_p^d,Q)$ are LCH and totally disconnected. Being Hausdorff metric spaces, there exists a countable dense subset $Y$ of $\so(\QQ_p^d,Q)$ such that a base for its ultrametric topology is given by the balls centred at $Y$, and the groups $\so(\QQ_p^d,Q)$ are second countable.

\begin{remark}\label{rem:isohomgrou}
Let $Q$, $Q'$ be two quadratic forms on $V\simeq \QQ_p^d$, and let $A$, $A'$ be the associated matrix representations with respect to some basis. We say that $Q$ is \emph{linearly equivalent} to $Q'$ if there exists an invertible linear map $f\colon V\rightarrow V$ such that $Q'\circ f = Q$, or equivalently, if there exists an invertible matrix $S\in\mathsf{M}_{d\times d}(\QQ_p)$ such that $A'=S^\top AS$. In this case $\mathrm{O}(\QQ_p^d,Q)\simeq \mathrm{O}(\QQ_p^d,Q')$ and $\so(\QQ_p^d,Q)\simeq \so(\QQ_p^d,Q')$, the isomorphism being $\mathrm{O}(\QQ_p^d,Q)\ni L\mapsto S^{-1}LS\in \mathrm{O}(\QQ_p^d,Q')$. This is also a homeomorphism of topological group (it has polynomial components given by the product with $S$ and $S^{-1}$). Up to linear equivalence (change of basis), we can assume every quadratic form $Q$ to be of the kind $Q(\mathbf{x})=\sum_{i=1}^da_ix_i^2$, i.e., with diagonal matrix representation $A=\diag(a_i)$. We say that $Q'$ is a scaling of $Q$ if there exists $t\in\QQ_p^\ast$ such that $Q'=tQ$. In this case, it is clear that $\mathrm{O}(\QQ_p^d,Q)= \mathrm{O}(\QQ_p^d,Q')$, and $\so(\QQ_p^d,Q)= \so(\QQ_p^d,Q')$. 
\end{remark}

For our purposes, it is useful to distinguish between definite and indefinite quadratic forms. We say that a quadratic form is \emph{definite} if it does not represent zero non-trivially --- i.e., $Q(\mathbf{x})=0$ if and only if $\mathbf{x}=\boldsymbol{0}$); we call a quadratic form \emph{indefinite} if it admits a non-zero \emph{isotropic vector} --- namely, there exists $\mathbf{x}\neq\boldsymbol{0}$ such that $Q(\mathbf{x})=0$. This is relevant, e.g., from the topological point of view, for an orthogonal group to be compact or not.

Concluding, quadratic forms, up to linear equivalence and scaling, lead to isomorphic special orthogonal groups;
hence, to list all such groups, we have first to identify the different classes of (equivalent) quadratic forms on $\QQ_p^d$. It is possible to prove that the rank, the discriminant and the so-called Hasse invariant provide a complete set of invariants through which classifying $p$-adic quadratic forms (see Theorem~$7$, n°~$2.3$, Chap.\ IV in~\cite{serre}). 

\begin{remark}
From here on, we will deal only with odd prime integers $p>2$, to facilitate the readability of this work. Indeed, the situation with $p=2$ is somehow peculiar as this is the only even prime, and often requires a longer separated treatment, which will be briefly described at the end. Also, the focus of the present work is on $d\in\{2,3\}$, and the whole machinery that also applies in higher dimension will be remarked at the end.
\end{remark}

As anticipated, we start recalling the classification of definite quadratic forms on $\QQ_p^2$ and $\QQ_p^3$ (see~\cite{our1st} for an explicit derivation). 

\begin{theorem}\label{th:quadform}
\label{theor:defquadforms}
    There are exactly three definite quadratic forms on $\QQ_p^2$, up to linear equivalence and scaling:
\beq 
Q_{-v}(\mathbf{x}) =x_1^2-vx_2^2,\quad Q_p(\mathbf{x})=x_1^2+px_2^2,\quad Q_{up}(\mathbf{x})=u x_1^2+px_2^2.\label{cev1nom}
\eeq 
There is a unique definite quadratic form on $\QQ_p^3$, up to linear equivalence and scaling:
\beq \label{cev2nom}
    Q_+(\mathbf{x})=x_1^2-vx_2^2+px_3^2.
\eeq 
In Eqs.~\eqref{cev1nom} and~\eqref{cev2nom}, we define
\beq\label{eq:vdef}
v\coloneqq
\begin{cases}
 -1 &\text{if } p\equiv3 \mod4,\\
 -u &\text{if } p\equiv1 \mod4,
\end{cases}
\eeq
for a non-square $u\in\mathbb{U}_p$. 
\end{theorem}

The matrix representations of the definite quadratic forms in Theorem~\ref{theor:defquadforms} with respect to the canonical basis are as follows: 
\begin{align}
    &A_{-v}=\diag(1,-v),\quad A_p=\diag(1,p),\quad A_{up}=\diag(u,p),\label{cev1}\\
    &A_+=\diag(1,-v,p)\label{cev2}.
\end{align}

Now we list the special orthogonal groups associated with these $p$-adic definite quadratic forms. These groups turn out to be \emph{all and the only} compact $p$-adic special orthogonal groups of degree two and three.
\begin{corollary}\label{cor:tuttisolicomp}
There are three (up to isomorphisms) compact special orthogonal groups on $\QQ_p^2$ for every prime $p>2$:
\beq
\so(2)_{p,\kappa}\coloneqq\so(\QQ_p^2,Q_\kappa)=\left\{L\in \mathsf{M}_{2\times2}(\QQ_p)\midd L^\top A_\kappa L=A_\kappa,\ \det(L)=1\right\},
\eeq 
where $\kappa\in\{-v,p,up\}$.\\
There is a unique (up to isomorphism) compact special orthogonal group on $\QQ_p^3$ for every prime $p>2$:
\beq 
    \so(3)_p\coloneqq\so(\QQ_p^3,Q_+)=\left\{L\in \mathsf{M}_{3\times3}(\QQ_p)\midd L^\top A_+L=A_+,\ \det(L)=1\right\}.
\eeq 
Indeed,
\begin{align}
    &\so(2)_{p,\kappa}<\mathrm{SL}(2,\ZZ_p),
    \quad \textup{for }\kappa\in \left\{-v,p,up\right\},\label{eq:inclsSO2}\\ 
    &\so(3)_p< \mathrm{SL}(3,\ZZ_p).\label{eq:inclsSO3}\
\end{align}
\end{corollary}
\begin{proof}
The three groups $\so(2)_{p,\kappa}$, and the group $\so(3)_p$ are the special orthogonal groups induced by the quadratic forms in Theorem~\ref{th:quadform}. These groups are indeed compact, as they are closed and bounded as subsets of $\mathsf{M}_{d\times d}(\QQ_p)\simeq \QQ_p^{d^2}$, for $d=2,3$ respectively. In particular, Theorem~5 in~\cite{our1st} shows that the entries of the matrices in $\so(3)_p$ are $p$-adic integers. This goes by contradiction, and consists of solving modular congruences $\mod p^n$ involving the underlying quadratic form: Since the latter is definite, there exists $n\in\NN$ (actually $n=1$) such that the modular congruences have only a trivial solution. A similar proof holds for $\so(2)_{p,\kappa}$ for every $\kappa\in\{-v,p,up\}$. For example, when $\kappa=-v$, the orthogonality condition for $L=\begin{pmatrix}p^{\nu_{ij}}u_{ij}\end{pmatrix}_{i,j=1,2}$ where $\nu_{ij}\in\ZZ,u_{ij}\in\mathbb{U}_p$, leads to 
\beq
\begin{cases}
p^{2\nu_{11}}u_{11}^2-vp^{2\nu_{21}}u_{21}^2=1,\\
p^{2\nu_{12}}u_{12}^2-vp^{2\nu_{22}}u_{22}^2=-v.
\end{cases}
\eeq 
If, by contradiction, $\min\{\nu_{11},\nu_{21}\}=\nu_{11}<0$ (similarly for $\nu_{21}$), then the first equation yields $u_{11}^2-vp^{2(\nu_{21}-\nu_{11})}u_{21}^2\equiv p^{-2\nu_{11}}\mod p^n$. In particular, $u_{11}^2-vp^{2(\nu_{21}-\nu_{11})}u_{21}^2\equiv 0\mod p$, whose only solution is $(u_{11},u_{21})\equiv (0,0)\in(\ZZ/p\ZZ)^2$: We get a contradiction since $u_{ij}\in\mathbb{U}_p$ means $u_{ij}\not\equiv0\mod p$. The same is argued in the second equation, hence $\nu_{ij}\geq0$ for every $i,j=1,2$.

On the other hand, indefinite quadratic forms admit non-trivial roots, hence they lead to unbounded, whence, non-compact special orthogonal groups. 
\end{proof}

\begin{notation}\label{not1}
In what follows, for the sake of conciseness, we will denote all compact $p$-adic special orthogonal groups of degree two and three for primes $p>2$ by $\so(d)_{p(,\kappa)}$, where $d=2,\kappa\in\{-v,p,up\}$ or $d=3$.
\end{notation}

\bigskip

We now recall a useful parametrisation of $\so(2)_{p,\kappa}$ known from Theorem~12 in~\cite{our1st}. An element of $\so(2)_{p,\kappa}$ takes the following matrix form with respect to the canonical basis of $\mathbb{Q}_p^2$:
\beq\label{genelso2}
\mathcal{R}_\kappa(\s )=\begin{pmatrix}
\frac{1-\alpha_\kappa\s ^2}{1+\alpha_\kappa\s ^2} & -\frac{2\alpha_\kappa\s }{1+\alpha_\kappa\s ^2}\\ \frac{2\s }{1+\alpha_\kappa\s ^2} & \frac{1-\alpha_\kappa\s ^2}{1+\alpha_\kappa\s ^2}
\end{pmatrix},\quad \s \in\mathbb{Q}_p\cup\{\infty\},
\eeq
where $\alpha_\kappa\in\{-v,p,\frac{p}{u}\}$ respectively for $\kappa\in\{-v,p,up\}$. Furthermore,
\beq \label{eq:invertparam}
\cR_\kappa\left(-\frac{1}{\alpha_\kappa\s }\right)=-\cR_\kappa(\s ),\qquad \cR_\kappa(\infty)=-\cR_\kappa(0)=-\mathrm{I}_{2\times2}.
\eeq

\begin{remark}\label{rem:paraminter}
Eq.~\eqref{eq:inclsSO2} is confirmed by Remark~14 in~\cite{our1st}, which shows another parametrisation for $\so(2)_{p,\kappa}$, in terms of only integer parameters. In fact, if $\s \in\ZZ_p$, then $1+\alpha_\kappa\s ^2\not\equiv0\mod p$, for every $\alpha_\kappa\in\{-v,p,\frac{p}{u}\}$. In these cases $|1+\alpha_\kappa\s ^2|_p=1$, i.e.,  $(1+\alpha_\kappa\s ^2)^{-1}\in\mathbb{U}_p$ which, multiplied by $1-\alpha_\kappa\s ^2,2\s ,-2\alpha_\kappa\s \in\ZZ_p$, gives $p$-adic integer matrix entries in parametrisation~\eqref{genelso2}.

Now we can distinguish two branches for the parameter $\s \in \QQ_p \cup \{\infty\}$: Either $\s \in\ZZ_p$ or $\s ^{-1}\in p\ZZ_p$ (including $\infty$ formally when $\s =0$). Since we want to exploit Eq.~\eqref{eq:invertparam} for a non-integer parameter, we consider either
$\s \in\ZZ_p$ or: 
\begin{itemize}
    \item[$\diamond$] $\s =\frac{1}{v\tau}$, $\tau\in p\ZZ_p$, for $\alpha_\kappa=-v$, yielding 
    \begin{align}
    \so(2)_{p,-v} &= \left\{\mathcal{R}_{-v}(\s )\midd \s \in\ZZ_p\right\}\cup\left\{-\mathcal{R}_{-v}(\s )\midd \s \in p\ZZ_p\right\};\label{eq:paraminteri1}
\end{align} 
    \item[$\diamond$] $\s  = -\frac{1}{\alpha_\kappa\tau},\,\tau\in \ZZ_p$ for $\alpha_\kappa\in\{p,\frac{p}{u}\}$, yielding
    \beq \label{eq:paraminteri2}
\so(2)_{p,\kappa} = \left\{\pm\mathcal{R}_\kappa(\s )\midd \s \in\ZZ_p\right\}.
\eeq
\end{itemize} 
\end{remark}

Moving to the three dimensional case, Theorem~19 in~\cite{our1st} tells that a rotation of $\so(3)_p$ around $\mathbf{n}\in \QQ_p^3\setminus\{\boldsymbol{0}\}$ takes the following matrix form with respect to an orthogonal basis $(\mathbf{g},\mathbf{h},\mathbf{n})$ of $\QQ_p^3$:
\beq \label{eq:param3x3gen}
\mathcal{R}_{\mathbf{n}}(\s )=\begin{pmatrix}
\frac{1-\alpha\s ^2}{1+\alpha\s ^2} & -\frac{2\alpha\s }{1+\alpha\s ^2}&0\\ \frac{2\s }{1+\alpha\s ^2} & \frac{1-\alpha\s ^2}{1+\alpha\s ^2}   &0\\0&0&1 
\end{pmatrix},
\eeq 
where $\s \in\QQ_p\cup\{\infty\}$ and $\alpha=Q_+(\mathbf{h})/Q_+(\mathbf{g})$ is proportional to $Q_+(\mathbf{n})$. In particular, the rotations around the reference axes of $\QQ_p^3$, are given by choosing different orderings of the canonical basis $(\mathbf{e}_1,\mathbf{e}_2,\mathbf{e}_3)$. A rotation around $\mathbf{e}_1\equiv x$ is located with respect to the basis $(\mathbf{e}_2,\mathbf{e}_3,\mathbf{e}_1)$ by $\alpha=-\frac{p}{v}$, coinciding with $\alpha_p=p$ when $p\equiv3\mod 4$, and with $\alpha_{up}=\frac{p}{u}$ for $p\equiv1\mod 4$; a rotation around $\mathbf{e}_2\equiv y$ is located with respect to $(\mathbf{e}_1,\mathbf{e}_3,\mathbf{e}_2)$ by $\alpha=p=\alpha_p$; and a rotation around $\mathbf{e}_3\equiv z$ is located with respect to $(\mathbf{e}_1,\mathbf{e}_2,\mathbf{e}_3)$ by $\alpha=-v=\alpha_{-v}$. Therefore, the rotations of $\so(3)_p$ around the reference axes with respect to the canonical basis are given by
\beq \label{eq:refaxisSO3}
\cR_x(\xi)=\begin{pmatrix}
    1&0&0\\0&\frac{1+\frac{p}{v}\xi^2}{1-\frac{p}{v}\xi^2} & \frac{2\frac{p}{v}\xi}{1-\frac{p}{v}\xi^2} \\
    0 & \frac{2\xi}{1-\frac{p}{v}\xi^2} 
    & \frac{1+\frac{p}{v}\xi^2}{1-\frac{p}{v}\xi^2}
\end{pmatrix}, \quad \cR_y(\eta)=\begin{pmatrix}
    \frac{1-p\eta^2}{1+p\eta^2} & 0&-\frac{2p\eta}{1+p\eta^2}\\0&1&0\\
    \frac{2\eta}{1+p\eta^2} & 0 & \frac{1-p\eta^2}{1+p\eta^2}
\end{pmatrix},\quad \cR_z(\zeta)=\begin{pmatrix}
    \frac{1+v\zeta^2}{1-v\zeta^2} & \frac{2v\zeta}{1-v\zeta^2} & 0\\\frac{2\zeta}{1-v\zeta^2} & \frac{1+v\zeta^2}{1-v\zeta^2} & 0\\0&0&1
\end{pmatrix},
\eeq 
for $\xi,\eta,\zeta\in\QQ_p\cup\{\infty\}$.
\begin{remark}\label{param:perSO3}
    The set $\so(3)_{p,\mathbf{n}}$ of rotations around a given $\mathbf{n}\in \QQ_p^3\setminus\{\boldsymbol{0}\}$ forms an abelian subgroup of $\so(3)_p$. In particular, $\so(3)_{p,x}\simeq \so(2)_{p,\kappa}$ with $\kappa=p$ for $p\equiv3\mod 4$ and $\kappa=up$ for $p\equiv 1\mod4$, $\so(3)_{p,y}\simeq \so(2)_{p,p}$ and $\so(3)_{p,z}\simeq \so(2)_{p,-v}$, for which Eqs.~\eqref{eq:paraminteri1},~\eqref{eq:paraminteri2} can be used just replacing $\pm\cR_\kappa(\s )$ with $\cR_{\mathbf{n}}(\s ),\cR_{\mathbf{n}}(\infty)\cR_{\mathbf{n}}(\s )$. We have
\beq\label{eq:paramSO3xy}
\so(3)_{p,\mathbf{n}} = \left\{\mathcal{R}_{\mathbf{n}}(\s )\midd \s \in\ZZ_p\right\}\cup \left\{\cR_{\mathbf{n}}(\infty)\mathcal{R}_{\mathbf{n}}(\s )\midd \s \in\ZZ_p\right\},
\eeq
for the $x$- and $y$-axes, while for the $z$-axis,
\beq\label{eq:paramSO3z}
    \so(3)_{p,z} = \left\{\mathcal{R}_z(\s )\midd \s \in\ZZ_p\right\}\cup\left\{\mathcal{R}_z(\infty)\mathcal{R}_z(\s )\midd \s \in p\ZZ_p\right\}.
\eeq 
\end{remark}
Last, in parallel to the real orthogonal case, only certain principal \textquotedblleft angle\textquotedblright\ decompositions of rotations around the reference axes hold for $\so(3)_p$. We recall just one of them (cf.\ Corollary~23 and Theorem~32 in~\cite{our1st}), which will be useful in our later work.
\begin{theorem} \label{theor:cardanorep}
For every prime $p>2$, every $L \in \so(3)_p$ can be written as the following Cardano (aka nautical or Tait–Bryan) type composition,
\beq \label{eq:cardanorep}
\cR_x\cR_y\cR_z,
\eeq 
for some $\cR_\mathbf{n}\in \so(3)_{p,\mathbf{n}}$, $\mathbf{n}\in\{x,y,z\}$.
Moreover, every $L$ has exactly two distinct Cardano decompositions of such kind: 
\beq  \label{eq:cardanorepDOPPIA} 
L=\cR_x(\xi)\cR_y(\eta)\cR_z(\zeta) = \cR_x(\infty)\cR_x(\xi)\, \cR_y(\infty)\cR_y(-\eta)\,\cR_z(\infty)\cR_z(\zeta),
\eeq
for some parameters $\xi,\eta,\zeta\in\QQ_p\cup\{\infty\}$. 
\end{theorem}

\section{Inverse limit characterisation of the compact \texorpdfstring{$p$}{Lg}-adic special orthogonal groups \texorpdfstring{$\so(2)_{p,\kappa}$}{Lg} and \texorpdfstring{$\so(3)_p$}{Lg}}\label{sec:invlimtopgru}
This section is devoted to deriving technical tools, which will be fundamental for the construction of the Haar measure on $\so(2)_{p,\kappa}$ and $\so(3)_p$ from an inverse limit, and for its explicit computability.

In the previous section, we argued that any $p$-adic special orthogonal group is LCH, totally disconnected and second countable, once given the $p$-adic topology. Also, Corollary~\ref{cor:tuttisolicomp} states that $\so(d)_{p(,\kappa)}$ (i.e. $\so(2)_{p,\kappa}$ for $d=2$, $\kappa\in\{-v,p,up\}$, and $\so(3)_p$ for $d=3$; see Notation~\ref{not1}) exhaust all compact $p$-adic special orthogonal groups of degree $d\in\{2,3\}$, for every prime $p>2$. As they are compact (Hausdorff) and totally disconnected, the groups $\so(d)_{p(,\kappa)}$ are profinite (Proposition~\ref{prop:profequiv}), i.e., they are inverse limits of suitable inverse families of finite discrete groups. These inverse families will be indexed by a countable totally ordered set, as $\so(d)_{p(,\kappa)}$ are second countable (Corollary~1.1.13 in~\cite{profinite}).

Recall that the elements of $\ZZ_p$ can be projected $\mod p^n$, $n\in\NN$, via the canonical projection $\cPoj_n$ as in Remark~\ref{rem:invlimZpeq}. Then, in $\mathsf{M}_{d\times d}(\ZZ_p)$ the matrix product is defined through sums and products of entries, for which $\cPoj_n$ are homomorphisms. Therefore, the map
\beq \label{eq:grouphoprpi}
\pi_n(M)=\pi_n\left(\begin{pmatrix}m_{ij} \end{pmatrix}_{ij}\right)\coloneqq \begin{pmatrix} \cPoj_n(m_{ij}) \end{pmatrix}_{ij}=\begin{pmatrix} m_{ij}\mod p^n \end{pmatrix}_{ij}
\eeq 
is a group homomorphism on any group $H$ contained in $\mathsf{M}_{d\times d}(\ZZ_p)$ to some other group $\pi_n(H)$ contained in $\mathsf{M}_{d\times d}(\ZZ/p^n\ZZ)$. Note that $\pi_n(H)$ is a finite group, since the order of $\mathsf{M}_{d\times d}(\ZZ/p^n\ZZ)$ is $\left\lvert\mathsf{M}_{d\times d}(\ZZ/p^n\ZZ)\right\rvert = (p^n)^{d^2}$. An equivalent way of describing this scenario is by considering the normal subgroups $\ker(\pi_n)=\left(\mathrm{I}_{d\times d}+p^n\mathsf{M}_{d\times d}(\ZZ_p)\right)\cap H$ of $H$, and by taking the quotients $H/\ker(\pi_n)=\pi_n(H)$. This argument applies in particular to each $\so(d)_{p(,\kappa)}$, as its matrices have $p$-adic integer entries [see Eqs.~\eqref{eq:inclsSO2},~\eqref{eq:inclsSO3}]. 
\begin{theorem}\label{prop:liminvfacile}
For every prime $p>2$, $d\in\{2,3\}$, and $\kappa\in\{-v,p,up\}$, we have the following topological group isomorphism:
\beq 
\so(d)_{p(,\kappa)}\simeq \varprojlim\left\{G_{(\kappa,)p^n},\,\varphi_{nl}\right\}_{\NN},
\eeq
where $G_{(\kappa,)p^n} \coloneqq \so(d)_{p(,\kappa)}\big/\left(\left(\mathrm{I}_{d\times d}+p^n\mathsf{M}_{d\times d}(\ZZ_p)\right)\cap \so(d)_{p(,\kappa)}\right)$ has discrete topology, $\so(d)_{p(,\kappa)}$ has $p$-adic topology, and with continuous group homomorphism $\varphi_{nl}$ defined as 
\beq
\varphi_{nl}:G_{(\kappa,)p^l}\rightarrow G_{(\kappa,)p^n},\quad (\ell_{ij}\mod p^l)_{ij}\mapsto (\ell_{ij}\mod p^n)_{ij}, \label{eq:mapinverslimit} 
\eeq
for every $n\leq l$, $n,l\in \NN$.
\end{theorem}
\begin{proof}
Specialising the argument around Eq.~\eqref{eq:grouphoprpi} for $\so(d)_{p(,\kappa)}$, we get a group homomorphism
\begin{align}
 \pi_n\colon \so(d)_{p(,\kappa)}&\rightarrow G_{(\kappa,)p^n}=\pi_n\big(\so(d)_{p(,\kappa)}\big) =\so(d)_{p(,\kappa)}\mod p^n,\notag\\
 L=\begin{pmatrix}
\ell_{ij}\end{pmatrix}_{ij}&\mapsto \pi_n(L)\coloneqq \begin{pmatrix}\ell_{ij}\mod p^n\end{pmatrix}_{ij},\label{eq:pin}
\end{align}
for every $n\in\NN$. We have $G_{(\kappa,)p^n} = \so(d)_{p(,\kappa)}\big/\ker(\pi_n)$ where $\ker(\pi_n)= \left(\mathrm{I}_{d\times d}+p^n\mathsf{M}_{d\times d}(\ZZ_p)\right)\cap \so(d)_{p(,\kappa)}$; we supply the finite group $G_{(\kappa,)p^n}$ with the quotient topology, which coincides with the discrete topology (as seen in Appendix~\ref{sec:appinvlimQp} for $\QQ_p/p^n\ZZ_p$) since, e.g., $\pi_n^{-1}(\{\mathrm{I}_{d\times d}+p^n\mathsf{M}_{d\times d}(\ZZ_p)\}\cap \so(d)_{p(,\kappa)})=\ker(\pi_n)=B_{-n}(\mathrm{I}_{d\times d})$ is the open ball of radius $p^{-n+1}$ centred at $\mathrm{I}_{d\times d}$. By construction, $G_{(\kappa,)p^n}$ is a topological group and $\pi_n$ is continuous.

The map $\varphi_{nn}$ as in Eq.~\eqref{eq:mapinverslimit} is the identity map on $G_{(\kappa,)p^n}$. The map $\varphi_{nl}$ is a group homomorphism, as $\Phi_{nl}$ is a ring homomorphism (Remark~\ref{rem:invlimZpeq}), and it is continuous, as its domain has discrete topology. It is easy to check that the maps $\varphi_{nl}$ are coherent, in the sense of axiom 2. in Definition~\ref{def:invlimsetgroup}. Therefore, $\left\{G_{(\kappa,)p^n},\varphi_{nl}\right\}_{\NN}$ is an inverse family of topological groups. Consider the map 
\beq \label{eq:mapsiotop}
F_{p(,\kappa)}\colon \so(d)_{p(,\kappa)}\rightarrow \varprojlim\left\{G_{(\kappa,)p^n},\,\varphi_{nl}\right\}_{\NN},\quad L\mapsto \big(\pi_n(L)\big)_n.
\eeq 
This resembles the map~\eqref{eq:fisotop} applied entry-wise to the matrices of $\so(d)_{p(,\kappa)}$, and with a similar argument to that in Appendix~\ref{sec:appinvlimQp}, one proves that $F$ is a topological group isomorphism. 
\end{proof}

Setting $\pi_n= \varphi_{n\infty}$, then $\pi_n\equiv\varphi_{nl}\circ\pi_l$ for every $n\leq l$. An element of $G_{(\kappa,)p^n}$ will be denoted equivalently by either $\pi_n(L)$ or $L\mod p^n$ (understanding that the reduction $\mod p^n$ is entry-wise), for some $L\in \so(d)_{p(,\kappa)}$. 

\medskip

We conclude this section with the orders of the projected groups $G_{(\kappa,)p^n}$, which will be fundamental to compute the Haar measure of a Borel set of $\so(d)_{p(,\kappa)}$.

\begin{remark}\label{rem:liftstescard}
    The maps $\varphi_{nl}\colon G_{(\kappa,)p^l}\rightarrow G_{(\kappa,)p^n}$ in Eq.~\eqref{eq:mapinverslimit} 
    are surjective (but not injective for $n<l$) homomorphisms of finite groups. The preimages of each of the elements in $G_{(\kappa,)p^n}$ are in bijective correspondence with each other. Indeed, let $N=\varphi_{nl}^{-1}(\mathrm{I}\mod p^n)$: If $\varphi_{nl}(L\mod p^l)=L\mod p^n$ then $\varphi_{nl}^{-1}(L\mod p^n)=N(L\mod p^l)$ of cardinality $\left\lvert \varphi_{nl}^{-1}(L\mod p^n)\right\rvert= \left\lvert N(L\mod p^l) \right\rvert=\left\lvert N\right\rvert$ for every $\pi_n(L)\in G_{(\kappa,)p^n}$.
\end{remark}
We first provide the orders of the finite projections of the compact $p$-adic special orthogonal groups of degree two. 
\begin{proposition}\label{prop:cardparam}
For every prime $p>2$, $\kappa\in\{-v,p,up\}$ and $n\in\NN$, 
\beq \label{eq:cardgpkk} 
\left\lvert G_{\kappa,p^n}\right\rvert=2p^n,\ \   \kappa\in\left\{p,up\right\}, \qquad \qquad
\left\lvert G_{-v,p^n}\right\rvert = p^{n-1}(p+1).
\eeq 
\end{proposition}
\begin{proof}
We exploit the parametrisation in Remark~\ref{rem:paraminter}, which can be projected modulo $p^n$, as the matrix entries and parameters are all in $\ZZ_p$:
\begin{align} 
G_{\kappa,p^n} &= \left\{\pm\mathcal{R}_\kappa(\s)\mod p^n\midd \s\in\ZZ_p\right\}\nonumber\\ &
= \left\{\mathcal{R}_\kappa(\s)\mod p^n\midd \s\in\ZZ/p^n\ZZ\right\}\cup \left\{-\mathcal{R}_\kappa(\s)\mod p^n\midd \s\in\ZZ/p^n\ZZ\right\},\label{eq:semplicep2k}
\end{align}
for $\kappa\in\{p,up\}$
,  while 
\beq \label{eq:semplicepvk}
    G_{-v,p^n} =\left\{\mathcal{R}_{-v}(\s)\mod p^n\midd \s\in \ZZ/p^n\ZZ\right\}\cup\left\{-\mathcal{R}_{-v}(\s)\mod p^n\midd \s\in p(\ZZ/p^n\ZZ)\right\}.
\eeq 
Now the calculus is by integer numbers modulo $p^n$. 

From Eq.~\eqref{eq:semplicep2k} it follows that $\left\lvert G_{\kappa,p^n}\right\rvert\leq 2p^n$ when $\kappa\in\{p,up\}$, while from Eq.~\eqref{eq:semplicepvk} we get $\left\lvert G_{-v,p^n}\right\rvert\leq p^{n-1}(p+1)$. We have $\cR_\kappa(\s)\equiv \cR_\kappa(\tau)\mod p^n$ if and only if both the following conditions are satisfied:
\begin{align}
    & \frac{1-\alpha_\kappa\s^2}{1+\alpha_\kappa\s ^2}\equiv \frac{1-\alpha_\kappa\tau^2}{1+\alpha_\kappa\tau^2}\mod p^n,\quad \textup{equivalent to}\quad \alpha_\kappa\s ^2\equiv \alpha_\kappa\tau^2\mod p^n; \label{eq:unanecessaria}\\
& \frac{2\s }{1+\alpha_\kappa\s ^2}\equiv \frac{2\tau}{1+\alpha_\kappa\tau^2} \mod p^n.\label{eq:unaltranecessaria}
\end{align} 
Plugging~\eqref{eq:unanecessaria} into~\eqref{eq:unaltranecessaria}, we get $\s \equiv \tau\mod p^n$. This means that the matrices within the set $\{\cR_\kappa(\s )\mod p^n\}$ are all distinct by varying the parameter $\s $, hence $\left\lvert \{\mathcal{R}_\kappa(\s )\mod p^n\midd \s \in\ZZ/p^n\ZZ\}\right\rvert =\left\lvert \ZZ/p^n\ZZ\right\rvert=p^n$ for every $\kappa\in\left\{-v,p,up\right\}$, as well as $\left\lvert \{-\mathcal{R}_\kappa(\s )\mod p^n\midd \s \in\ZZ/p^n\ZZ\}\right\rvert =p^n$ for $\kappa\in\left\{p,up\right\}$, while $\left\lvert \left\{-\mathcal{R}_{-v}(\s )\mod p^n\midd \s \in p(\ZZ/p^n\ZZ)\right\} \right\rvert =\left\lvert p(\ZZ/p^n\ZZ)\right\rvert= p^{n-1}$.

On the other hand, a necessary condition for $\cR_\kappa(\s )\equiv -\cR_\kappa(\tau)\mod p^n$ is
\beq\label{eq:disjointsign}
\frac{1-\alpha_\kappa\s ^2}{1+\alpha_\kappa\s ^2}\equiv -\frac{1-\alpha_\kappa\tau^2}{1+\alpha_\kappa\tau^2}\mod p^n,\quad \textup{equivalent to}\quad 
\alpha_\kappa^2\s ^2\tau^2\equiv 1\mod p^n.
\eeq 
This is impossible when $p\mid \alpha_\kappa\s \tau$, i.e., it is always impossible for $\kappa\in\left\{p,up\right\}$, and it is only possible for $\kappa=-v$ when $\s ,\tau\not\equiv0\mod p$, where $\cR_z\left(\tau:=\frac{1}{v\s }\right)=-\cR_z(\s )\mod p^n$ is in the first set of the union in Eq.~\eqref{eq:semplicepvk}. Therefore, the two sets of the unions in Eqs.~\eqref{eq:semplicep2k},~\eqref{eq:semplicepvk} are disjoint, and the order of $G_{\kappa,p^n}$ is given by the sum of the orders of those two sets.
\end{proof}

In order to reach a similar result in $d=3$, we want to make use of an analogous Cardano representation for $G_{p^n}$, like that in Eq.~\eqref{eq:cardanorepDOPPIA} for $\so(3)_p$.
\begin{theorem}\label{theor:cardanodoppiamod}
For every prime $p>2$ and $n\in\NN$, every $\pi_n(L)\in G_{p^n}$ has exactly two distinct Cardano decompositions of the kind $\cR_x\cR_y\cR_z\mod p^n$.
\end{theorem}
The proof of this result is in Appendix~\ref{sec:app3}, and exploits Remark~\ref{param:perSO3}. The twofold Cardano decomposition of $L\in G_{p^n}$ is given in one (depending on $L$) of the six possibilities in Remark~\ref{rem:6cardequality}, and essentially coincides with that in~\eqref{eq:cardanorepDOPPIA} once suitably projected via $\pi_n$.

\begin{proposition}\label{prop:ordinigrup3dmodpk}
For every prime $p>2$, and $n\in\NN$, 
\beq 
\lvert G_{p^n}\rvert =2p^{3n-1}(p+1).
\eeq 
\end{proposition}
\begin{proof}
We recall from Remark~\ref{param:perSO3} that $\so(3)_{p,x}\simeq \so(2)_{p,\kappa}$ with $\kappa=p$ for $p\equiv3\mod 4$ and $\kappa=up$ for $p\equiv 1\mod4$, $\so(3)_{p,y}\simeq \so(2)_{p,p}$ and $\so(3)_{p,z}\simeq \so(2)_{p,-v}$. The respective isomorphisms hold for the images of these groups with respect to $\pi_n$. Therefore, according to Proposition~\ref{prop:cardparam},
\beq 
\lvert \pi_n\left(\so(3)_{p,x}\right)\rvert=\lvert \pi_n\left(\so(3)_{p,y}\right)\rvert = 2p^n, \qquad \lvert \pi_n\left(\so(3)_{p,z}\right)\rvert= p^{n-1}(p+1).
\eeq 
As a direct consequence of the duplicity of the Cardano decomposition in Theorem~\ref{theor:cardanodoppiamod}, we have
\beq 
\lvert G_{p^n}\rvert =\frac{1}{2}\lvert \pi_n\left(\so(3)_{p,x}\right)\rvert \lvert \pi_n\left(\so(3)_{p,y}\right)\rvert \lvert \pi_n\left(\so(3)_{p,z}\right)\rvert.
\eeq 
\end{proof}
Note from Propositions~\ref{prop:cardparam},~\ref{prop:ordinigrup3dmodpk} that
\beq \label{eq:liftings2d}
|G_{\kappa,p^{n+1}}| =p |G_{\kappa,p^n}|,\qquad \quad \lvert G_{p^{n+1}}\rvert=p^3\lvert G_{p^n}\rvert,
\eeq 
for every $n\in\NN$. According to Remark~\ref{rem:liftstescard}, this means that each element of $G_{\kappa,p^n}$ (resp.\ $G_{p^n}$) has a preimage of cardinality $p$ (resp.\ $p^3$) with respect to $\varphi_{n,n+1}$.
\begin{remark}\label{rem:henselift}
For every prime $p>2$, $n\in\NN$ and $\kappa\in\{-v,p,up\}$, let us introduce the groups
\begin{align}
\widetilde{G}_{\kappa,p^n} &
      \coloneqq \left\{\widetilde{L}\in \mathsf{M}_{2\times 2}\left(\ZZ/p^n\ZZ\right)\midd  
            \widetilde{L}^\top \pi_n\left(A_\kappa\right)\widetilde{L} \equiv \pi_n\left(A_\kappa\right)\bmod p^n,\ \det \widetilde{L}\equiv 1 \mod p^n \right\},\label{eq:altrogrequiv2x2}\\
\widetilde{G}_{p^n} &
      \coloneqq \left\{\widetilde{L}\in \mathsf{M}_{3\times 3}\left(\ZZ/p^n\ZZ\right)\midd  
            \widetilde{L}^\top \pi_n\left(A_+\right)\widetilde{L} \equiv \pi_n\left(A_+\right)\bmod p^n,\ \det \widetilde{L}\equiv 1 \mod p^n \right\},\label{eq:altrogrequiv}
\end{align}
of solutions modulo $p^n$ of the defining conditions of $\so(d)_{p(,\kappa)}$ ($d\in\{2,3\}$). It is always true that
\beq \label{eq:unainclusionesem}
G_{(\kappa,)p^n}  \subseteq \widetilde{G}_{(\kappa,)p^n},
\eeq
since if the entries of $L$ form a solution of the defining conditions of $\so(d)_{p(,\kappa)}$ over $\ZZ_p$, then the entries of $\pi_n(L)$ form a solution of the same conditions over $\ZZ/p^n\ZZ$. Conversely, does each solution modulo $p^n$ lift to a $p$-adic integer solution? It is equivalent to asking whether the groups $G_{p^n}$ and $\widetilde{G}_{p^n}$ coincide or not. This question, already put forward in~\cite{our2nd}, has the character of Hensel’s lemma, and here we provide a positive answer (see Appendix~\ref{sec:appliftalaHens}). By Proposition~\ref{prop:HenselSO(2)K} (resp.~\ref{prop:HenselSO(3)}), each element of $\widetilde{G}_{\kappa,p^n}$ (resp.\ $\widetilde{G}_{p^n}$) lifts to exactly $p$ elements (resp.\ $p^3$) in $\widetilde{G}_{\kappa,p^{n+1}}$ (resp.\ $\widetilde{G}_{p^{n+1}}$) --- in agreement with Eq.~\eqref{eq:liftings2d} --- so that (cf.\ Corollary~\ref{cor:HenselSO(3)})
\beq\label{eq:unauguaglsem}
G_{(\kappa,)p^n}= \widetilde{G}_{(\kappa,)p^n}.
\eeq
Eq.~\eqref{eq:unauguaglsem} provides an alternative description of $G_{(\kappa,)p^n}$ --- besides the already known Eq.~\eqref{genelso2} and Cardano decomposition modulo $p^n$--- as in Eqs.~\eqref{eq:altrogrequiv2x2},~\eqref{eq:altrogrequiv} by solving the system of special orthogonal conditions modulo $p^n$, or by lifting solutions [cf.\ systems~\eqref{pv1lift},~\eqref{pv2lift},~\eqref{puplift},~\eqref{eq:solift0},~\eqref{eq:solift0n}].
\end{remark}

\section{Construction of the Haar measure on \texorpdfstring{$\so(2)_{p,\kappa}$}{Lg} and \texorpdfstring{$\so(3)_p$}{Lg} from an inverse limit of measure spaces}\label{sec:invlimme}
The groups $\so(d)_{p(,\kappa)}$ are compact, for every prime $p>2$, $d\in\{2,3\}$, and $\kappa\in\{-v,p,up\}$, hence they admit an essentially unique (left and right) Haar measure. 
In this section, we finally get to construct it, by exploiting the machinery of inverse limit of measure spaces. Before starting, we recall to the reader the following result (Proposition~7 in VII.15 of~\cite{BourInt}), providing another proof of the existence of the Haar measure in terms of inverse limits.
\begin{proposition}\label{prop:genHaarinvBour}
Let $G$ be a locally compact group. Let $(K_i)_{i\in I}$ be a decreasing directed family of compact normal subgroups of $G$ with intersection $\{e\}$. Set $G_i\coloneqq G/K_i$; let $\cProj_i\colon G\rightarrow G_i$ and $f_{ij}\colon G_j\rightarrow G_i$ ($i\leq j$) be the canonical homomorphisms. Then $G$ can be identified with the inverse limit group $\varprojlim\{G_i,f_{ij}\}_{i\in I}$, and the canonical mapping of the latter to $G_i$ is identified with $\cProj_i$. For every $i\in I$, let $\mu_i$ be the left (resp.\ right) Haar measure on $G_i$. Suppose that $f_{ij}(\mu_j) =\mu_i$ for $i\leq j$. Then, there \emph{exists a unique} measure $\mu$ on $G$ such that $\cProj_i(\mu)=\mu_i$ for all $i\in I$; $\mu$ is said to be the \emph{inverse limit measure} of the $\mu_i$s; $\mu$ is the left (resp.\ right) Haar measure on $G$.
\end{proposition}

In what follows, we give a concrete realisation of this abstract result for $\so(2)_{p,\kappa}$ and $\so(3)_p$. We specifically construct an inverse family of Haar measure spaces over the inverse family $\left\{G_{(\kappa,)p^n},\,\varphi_{nl}\right\}_{\NN}$ of topological groups characterising $\so(d)_{p(,\kappa)}$ (see Theorem~\ref{prop:liminvfacile}). This relies on the counting (i.e. Haar) measure on the power set of the finite groups $G_{(\kappa,)p^n}=\so(d)_{p(,\kappa)}\mod p^n$. Indeed, $G_{(\kappa,)p^n}$ is the quotient group of $\so(d)_{p(,\kappa)}$ by $\left(\mathrm{I}_{d\times d}+p^n\mathsf{M}_{d\times d}(\ZZ_p)\right)\cap \so(d)_{p(,\kappa)}$, and $\left\{\left(\mathrm{I}_{d\times d}+p^n\mathsf{M}_{d\times d}(\ZZ_p)\right)\cap \so(d)_{p(,\kappa)}\right\}_{n\in\NN}$ is a decreasing directed family of compact normal subgroups of $\so(d)_{p(,\kappa)}$ with intersection $\{\mathrm{I}_{d\times d}\}$, as in Proposition~\ref{prop:genHaarinvBour}. The inverse family of measure spaces is then used to define a measure on those particular subsets of $\so(d)_{p(,\kappa)}$ which are preimages (or lifts) of subsets of $G_{(\kappa,)p^n}$ for some $n\in\NN$. Finally, the crucial point is the possibility to extend the latter measure to the $\sigma$-algebra of $\so(d)_{p(,\kappa)}$ generated by those sets. We will explicitly verify that the generated $\sigma$-algebra is Borel (knowing the non-Archimedean topology of the $p$-adic rotation groups), and that the inverse limit measure on it is the Haar measure on $\so(d)_{p(,\kappa)}$. Our result will be concrete and workable: The technical tools developed in Section~\ref{sec:invlimtopgru} allow us to calculate the cardinality of $G_{(\kappa,)p^n}$, so as to be able to explicitly evaluate the Haar measure on $\so(d)_{p(,\kappa)}$.

\bigskip

Since the result $\so(d)_{p(,\kappa)}\simeq\varprojlim\left\{G_{(\kappa,)p^n},\,\varphi_{nl}\right\}_{\NN}$, we avoided writing the dependence of the maps $\varphi_{nl}$ and $\pi_n$ on $p$, $d$ and in case $\kappa$, not to overload the notation. We carry this choice forward throughout the section, and introduce those indices back just for the final mathematical objects. $\left\{G_{(\kappa,)p^n},\,\varphi_{nl}\right\}_{\NN}$ is in particular an inverse family of sets, at the basis of Definition~\ref{def:invemeaspac} of inverse family of measure spaces, which we want to construct over the former. Every $G_{(\kappa,)p^n}$ is a finite topological group supplied with discrete topology, i.e., the topology coinciding with the power set of $G_{(\kappa,)p^n}$. Maintaining that structure, any finite set can be turned into a measure space, by taking its power set as $\sigma$-algebra, and the counting measure on it. Then, let 
\beq \label{eq:sigmaalmodpk}
\Sigma_n\coloneqq\{E_n\subseteq G_{(\kappa,)p^n}\}
\eeq 
be the power set of the finite group $G_{(\kappa,)p^n}$. Clearly $\Sigma_n$ is finite, of cardinality $|\Sigma_n|=2^{|G_{(\kappa,)p^n}|}$. The normalised counting measure on any finite group $G_{(\kappa,)p^n}$ is the probability measure
\beq \label{eq:objCoutmeasure}
\mu_n\colon\Sigma_n\rightarrow[0,1],\quad \mu_n(E_n)=\frac{\lvert E_n\rvert}{\lvert G_{(\kappa,)p^n}\rvert},
\eeq 
where $\lvert G_{(\kappa,)p^n}\rvert$ is known from Propositions~\ref{prop:cardparam},~\ref{prop:ordinigrup3dmodpk}, and $\mu_n$ turns out to be the essentially unique (left and right) Haar measure on $(G_{(\kappa,)p^n},\Sigma_n)$, for every $n\in\NN$.

Since $\varphi_{nl}$ is a surjective group homomorphism, then Remark~\ref{rem:liftstescard} holds true and
\beq \label{eq:measurepreserving}
\mu_{(\kappa,)p^l}\big(\varphi_{nl}^{-1}(E_n)\big)=\frac{\lvert \varphi_{nl}^{-1}(E_n)\rvert}{\lvert G_{(\kappa,)p^l}\rvert} =
\frac{\lvert E_n\rvert }{\lvert G_{(\kappa,)p^n}\rvert }
= \mu_n(E_n),
\eeq 
for every $E_n\subseteq G_{(\kappa,)p^n}$ and every $n\leq l$:
\beq 
\mu_l\circ\varphi_{nl}^{-1}=\mu_n.
\eeq 
This means that $\varphi_{nl}$ is measure preserving, hence the family $\left\{\left(G_{(\kappa,)p^n},\,\Sigma_n,\mu_n\right),\varphi_{nl}\right\}_{\NN}$ of probability spaces is an inverse system of (Haar) measure spaces, according to Definition~\ref{def:invemeaspac}. At this point, we run through Section~1 of~\cite{Choksi}.
\begin{definition}\label{def:liftSigma}
For every $n\in\NN$, let $\Sigma_n^\ast$ be the 
preimage of $\Sigma_n$ [as in~\eqref{eq:sigmaalmodpk}] with respect to $\pi_n$ [as in~\eqref{eq:pin}], namely, we set:
\beq \label{eq:def:liftSigma}
\Sigma_n^\ast\coloneqq\pi^{-1}_n(\Sigma_n)=\left\{\pi_n^{-1}(E_n)\midd E_n\in \Sigma_n\right\}.
\eeq 
\end{definition}

It is easy to prove that if $f\colon X\rightarrow Y$ is a surjective map from the set $X$ onto the set $Y$, and $\Sigma$ is a $\sigma$-algebra of $Y$, then $f^{-1}(\Sigma)$ is a $\sigma$-algebra of $X$. As a consequence, $\Sigma_n^\ast$ is a finite ($\sigma$-)algebra of $\so(d)_{p(,\kappa)}$, as is shown explicitly by means of the following two points:
\begin{itemize}
    \item for every (finite) family $\left\{E_\lambda=\pi_n^{-1}\big({E_n}^{(\lambda)}\big)\right\}_\lambda$ of sets in $\Sigma_n^\ast$, $\bigcup_\lambda E_\lambda=    \pi_n^{-1}\left(\bigcup_\lambda E_n^{(\lambda)}\right)\in\Sigma_n^\ast$, because $\bigcup_\lambda E_n^{(\lambda)}\in \Sigma_n$;
    \item for every $E=\pi_n^{-1}(E_n)\in \Sigma_n^\ast$, $ \so(d)_{p(,\kappa)}\setminus E = \pi_n^{-1}(G_{p^n}\setminus E_n) \in \Sigma_n^\ast$, because $G_{p^n}\setminus E_n\in\Sigma_n$.
\end{itemize}
 \begin{definition}\label{def.mes}
For every $n\in\NN$, we denote by $\mu_n^\ast$ the set function defined on $\Sigma_n^\ast$ as
\beq \label{eq:measulivelli}
\mu_n^\ast(E)\coloneqq(\mu_n\circ\pi_n)(E)=\mu_n(E_n),
\eeq
for $E=\pi_n^{-1}(E_n)\in\Sigma_n^\ast$, where $\Sigma_n^\ast$ is as in~\eqref{eq:def:liftSigma} and $\mu_n$ as in Eq.~\eqref{eq:objCoutmeasure}.
 \end{definition}
The set function $\mu_n^\ast$ is well defined, since  $\pi_n(E)=\pi_n\big(\pi_n^{-1}(E_n)\big)=E_n$ by surjectivity of $\pi_n$. Moreover, $\mu_n^\ast$ inherits from $\mu_n$ the properties of a probability  measure: $\mu_n^\ast(E)\in[0,1]$ for every $E\in\Sigma_n^\ast$, $\mu_n^\ast(\emptyset)=0$, and $\mu_n^\ast$ is ($\sigma$-)additive on $\Sigma_n^\ast$. In fact, let $\left\{E_\lambda=\pi_n^{-1}\big({E_n}^{(\lambda)}\big)\right\}_\lambda$ be a (finite) family of pairwise disjoint sets in $\Sigma_n^\ast$: If $E_\lambda\cap E_{\lambda'}=\emptyset$ then ${E_n}^{(\lambda)}\cap
{E_n}^{(\lambda')}=\emptyset$, and by the ($\sigma$-)additivity of $\mu_n$,
\beq
\mu_n^\ast\bigg(\bigcup_\lambda E_{\lambda}\bigg)=
\mu_n^\ast\left(\pi_n^{-1}\bigg(\bigcup_\lambda {E_n}^{(\lambda)}\bigg)\right)=\mu_n\left(\bigcup_\lambda {E_n}^{(\lambda)}\right)=\sum_\lambda\mu_n\left({E_n}^{(\lambda)}\right)=\sum_\lambda \mu_n^\ast(E_\lambda).
\eeq
Hence, $\left(\so(d)_{p(,\kappa)},\Sigma_n^\ast,\mu_n^\ast\right)$ is a probability space, for every $n\in\NN$. Also,
\beq \label{eq:unacontenutelaltra}
\Sigma_n^\ast\subset\Sigma_l^\ast
\eeq 
for every $n< l$, since $\Sigma_n^\ast\ni E=\pi_n^{-1}(E_n)=\pi_l^{-1}(\varphi_{nl}^{-1}(E_n))=\pi_l^{-1}(E_l)$ for $E_l\coloneqq \varphi_{nl}^{-1}(E_n) \in\Sigma_l$. Thus
\beq \label{eq:measconsist}
\mu_n^\ast(E)=\mu_l^\ast(E)
\eeq 
for every $E\in\Sigma_n^\ast\subset\Sigma_l^\ast$, as $\varphi_{nl}$ is measure preserving~\eqref{eq:measurepreserving}.
\begin{definition}\label{def.A}
We denote by
\beq \label{eq:def.A}
A\coloneqq\bigcup_{n\in\NN}\Sigma_n^\ast, 
\eeq 
the union of all the $\sigma$-algebras $\Sigma_n^\ast$ [as in~\eqref{eq:def:liftSigma}] of $\so(d)_{p(,\kappa)}$. 
\end{definition}
It is clear that $A$ is  a countable set, being the countable union of finite sets. Furthermore, it is not hard to prove that $A$ is an algebra of sets of $\so(d)_{p(,\kappa)}$:
\begin{itemize}
    \item if $E=\pi_n^{-1}(E_n),\,F=\pi_l^{-1}(E_l)\in A$ with $n\leq l$ (the case $n> l$ is analogous), then $E\cup F=\pi_l^{-1}\big(\varphi_{nl}^{-1}(E_n)\cup E_l\big)\in \Sigma_l^\ast\subset A$. With a similar reasoning, by taking the maximum over the subscripts, $A$ is closed under finite union;
    \item if $E=\pi_n^{-1}(E_n)\in A$, then $\so(d)_{p(,\kappa)}\setminus E =\pi_n^{-1}(G_{(\kappa,)p^n}\setminus E_n)\in \Sigma_n^\ast\subset A$.
\end{itemize}

\begin{remark}
$A$ is not a $\sigma$-algebra of $\so(d)_{p(,\kappa)}$ (it is not closed under countable union), as it is the union of a countable sequence of $\sigma$-algebras one contained in the other~\cite{broughton}.
\end{remark}

\begin{definition}
Let $A$ be the algebra of sets~\eqref{eq:def.A} of $\so(d)_{p(,\kappa)}$. We denote by $\mu$ the set function on $A$ defined by
\beq \label{eq:def.mu}
\mu(E)\coloneqq\mu_n^\ast(E),\quad\mbox{for every $E\in \Sigma_n^\ast$},
\eeq 
where $\mu_n^\ast$ is the measure in Definition~\ref{def.mes}, and $\Sigma_n^\ast$ the $\sigma$-algebra on $\so(d)_{p(,\kappa)}$ as in~\ref{eq:def:liftSigma}, for every $n\in\NN$.
\end{definition}
The map $\mu$ is well defined on $A$, for the above discussion around Eq.~\eqref{eq:measconsist}. Moreover, $\mu(E)\in[0,1]$ for every $E\in A$ and $\mu(\emptyset)=0$, as the $\mu_n^\ast$s themselves satisfy these properties. 

A map $\mu$ constructed as above may not be $\sigma$-additive in the more general scenario where the initial inverse family is indexed by a directed set with measures $\mu_n$ on rings of sets. This would imply that $\mu$ cannot be extended to a measure on the $\sigma$-algebra generated by $A$~\cite{Halmos}. But this is not the case in this work, as the following proposition states.

\begin{proposition}
The map $\mu$ defined in~\eqref{eq:def.mu} is $\sigma$-additive.
\end{proposition}
\begin{proof}
Let $\{E_\lambda\in A\}_\lambda$ be a countable family of pairwise disjoint sets such that $\bigcup_\lambda E_\lambda \in A$. By virtue of Eq.~\eqref{eq:unacontenutelaltra}, by taking the maximum $m$ over the subscripts $n_\lambda$ in $E_\lambda = \pi^{-1}_{n_\lambda}(E_{n_\lambda})$, we can write $\bigcup_{\lambda=1}^\infty E_\lambda=\pi_l^{-1}(E_l)$, for some $E_l\in\Sigma_l$ and some $l\geq m$. Since $n_\lambda\leq l$ for every $\lambda$, it follows that $\mu(E_\lambda)=\mu_l^\ast(E_\lambda)$ for every $\lambda$ by Eq.~\eqref{eq:measconsist}, and
\beq 
\mu \left(\bigcup_\lambda E_\lambda\right) = \mu_l^\ast\left(\bigcup_\lambda E_\lambda\right)=\sum_\lambda \mu_l^\ast(E_\lambda)=\sum_\lambda\mu(E_\lambda)
\eeq 
by $\sigma$-additivity of $\mu_l^\ast$.
\end{proof}

This proves that $\mu$ is a probability measure on the algebra $A$ of $\so(d)_{p(,\kappa)}$. However, to get a measure space, we need to introduce the $\sigma$-algebra generated by the algebra $A$. 
\begin{definition}\label{def:defSigmagenA}
For every prime $p>2$, $d\in\{2,3\}$, and $\kappa\in\{-v,p,up\}$, let $\Sigma_{p(,\kappa)}(A)$ be the $\sigma$-algebra of $\so(d)_{p(,\kappa)}$ generated by the algebra $A$ as in~\eqref{eq:def.A}.
\end{definition}

By Theorem~C p.~23 in~\cite{Halmos}, since $A$ is countable, then so is $\Sigma_{p(,\kappa)}(A)$.
\begin{remark}\label{rem:nale}
Since $\mu$ is a finite measure on $A$, then there exists a unique $\sigma$-additive extension of $\mu$ to a measure $\overline{\mu}_{p(,\kappa)}$ on $\Sigma_{p(,\kappa)}(A)$ which coincides with $\mu$ on $A$ (Theorem~A p.~54 in~\cite{Halmos}). Then, the measure space $\big(\so(d)_{p(,\kappa)}, \Sigma_{p(,\kappa)}(A), \overline{\mu}_{p(,\kappa)}\big)$ is what is called the \emph{inverse limit of the inverse family of measure spaces} $\left\{\left(G_{(\kappa,)p^n},\,\Sigma_n,\mu_n\right),\varphi_{nl}\right\}_{\NN}$~\cite{Choksi}, for every prime $p>2$, $d\in\{2,3\}$, and $\kappa\in\{-v,p,up\}$.
\end{remark}
We construct such an extension $\overline{\mu}_{p(,\kappa)}$ as follows. For every $E\subseteq \so(d)_{p(,\kappa)}$, consider $\pi_n^{-1}\big(\pi_n(E)\big)\in \Sigma_n^\ast\subset A$ for every $n\in \NN$. Clearly $E\subseteq \pi_n^{-1}\big(\pi_n(E)\big)$ for every $n\in \NN$, as well as $\pi_{n+1}^{-1}\big(\pi_{n+1}(E)\big)\subseteq \pi_n^{-1}\big(\pi_n(E)\big)$ because $\pi_{n+1}(\cR)\in \pi_{n+1}(E)$, $\cR\in\so(d)_{p(,\kappa)}$, implies $\varphi_{n,n+1}\big(\pi_{n+1}(\cR)\big)\in \varphi_{n,n+1}\big( \pi_{n+1}(E)\big)$, equivalent to $\pi_n(\cR)\in\pi_n(E)$. Hence
\beq \label{eq:seqdecredidef}
E\subseteq\dots\subseteq\pi_{n+1}^{-1}\big(\pi_{n+1}(E)\big)\subseteq \pi_n^{-1}\big(\pi_n(E)\big)\subseteq\dots\subseteq \pi_1^{-1}\big(\pi_1(E)\big),
\eeq 
i.e., $\left\{\pi_n^{-1}\big(\pi_n(E)\big)\right\}_{\NN}$
is a decreasing sequence of sets in $A$ containing $E$. The limit of such a sequence is $\bigcap_n\pi_n^{-1}\big(\pi_n(E)\big)$, it belongs to $\Sigma_{p(,\kappa)}(A)$~\cite{Halmos} but does not necessarily coincide with $E$. Therefore, we give the following
\begin{definition}
For every prime $p>2$ and $\kappa\in\{-v,p,up\}$, let $\Sigma_{p(,\kappa)}(A)$ be as in Definition~\ref{def:defSigmagenA}. For every $E\in\Sigma_{p(,\kappa)}(A)$, we define the set function
\beq\label{eq:defolmu}
\overline{\mu}_{p(,\kappa)}(E)\coloneqq\inf_n\left\{\mu  \big(\pi_n^{-1}\left(\pi_n(E)\right)\big)\right\}=\inf_n\left\{\mu_n\big(\pi_n(E)\big)\right\}.
\eeq 
\end{definition}
This coincides with the standard construction 
\beq \label{eq:defolmuST} \overline{\mu}_{p(,\kappa)}(E)=\inf\left\{\mu(F)\midd E\subseteq F\in A\right\} \eeq 
of the extension of a measure $\mu$ to $\Sigma_{p(,\kappa)}(A)$. Indeed,  $E\subseteq F=\pi_n^{-1}(F_n)$ is equivalent to $\pi_n(E)\subseteq F_n$ and, by varying $F_n$ at fixed $n$, we have
\beq
\pi^{-1}_n\big(\pi_n(E)\big)=\pi_n^{-1}\bigg(\bigcap\limits_{F_n\supseteq \pi_n(E)}F_n\bigg)=\bigcap\limits_{F_n\supseteq \pi_n(E)}\pi_n^{-1}(F_n).
\eeq
Therefore $\mu\big(\pi^{-1}_n\left(\pi_n(E)\right)\big)\leq \mu (\pi_n^{-1}(F_n))$ for every $F_n\in\Sigma_n$ such that $E\subseteq \pi_n^{-1}(F_n)$, at any fixed $n$, since $\mu$ is monotone as a measure (Theorem~A p.~37 in~\cite{Halmos}). Thus $\overline{\mu}_{p(,\kappa)}(E)=\inf\left\{\mu(F)\midd E\subseteq F\in A\right\}=\inf_n\big(\inf_{F_n}\left\{\mu(\pi_n^{-1}(F_n))\midd E\subseteq \pi_n^{-1}(F_n)\right\}\big)=\inf_n\left\{\mu  \big(\pi_n^{-1}\left(\pi_n(E)\right)\big)\right\}$.

The map $\overline{\mu}_{p(,\kappa)}$ defined on $\Sigma_{p(,\kappa)}(A)$ takes values in $[0,1]$, since it is bounded by probabilities $\mu\big(\pi_n^{-1}\left(\pi_n(E)\right)\big)$. And this standard construction is known to be a $\sigma$-additive measure on $\Sigma_{p(,\kappa)}(A)$ (cf.\ Theorem~C p.~47 in~\cite{Halmos}).
Moreover 
\beq 
\overline{\mu}_{p(,\kappa)}(E)=\mu(E)\quad \text{for}\ E\in A,
\eeq 
in fact if $A\ni E=\pi_m^{-1}(E_m)$, then $\pi_n^{-1}\big(\pi_n(E)\big)=\pi_n^{-1}\big(\varphi_{mn}^{-1}(E_m)\big)=\pi_m^{-1}(E_m)=E$ for every $n\geq m$, i.e., the sets of sequence~\eqref{eq:seqdecredidef} all coincide with $E$ for every $n\geq m$. Hence $\mu\big(\pi_n^{-1}\left(\pi_n(E)\right)\big)=\mu_m^\ast \big(\pi_n^{-1}\left(\pi_n(E)\right)\big)$ by~\eqref{eq:measconsist}, and this is $\mu_m^\ast(E)=\mu(E)$ for $n\geq m$. Since $\mu_m^\ast$ is monotone, then $\mu(E)\leq \mu\big(\pi_n^{-1}\left(\pi_n(E)\right)\big)$ for every $n\in\NN$, and $\overline{\mu}_{p(,\kappa)}(E)=\inf_n\left\{\mu\big(\pi_n^{-1}\left(\pi_n(E)\right)\big)\right\}=\mu(E)$.

We conclude that $\overline{\mu}_{p(,\kappa)}$ on $\Sigma_{p(,\kappa)}(A)$ is the unique extension of $\mu$ we were looking for in Remark~\ref{rem:nale}. We have proved the following result.
\begin{theorem}\label{thr:invlimeas} For every prime $p>2$, and $\kappa\in\left\{-v,p,up\right\}$, the triples
\beq 
\big(\so(2)_{p,\kappa},\Sigma_{p,\kappa}(A),\overline{\mu}_{p,\kappa}\big),\qquad
\big(\so(3)_p,\Sigma_p(A),\overline{\mu}_p\big),
\eeq 
with $\Sigma_{p(,\kappa)}(A)$ as in Definition~\ref{def:defSigmagenA} and $\overline{\mu}_{p(,\kappa)}$ as in~\eqref{eq:defolmu}, are (probability) measure spaces, and they are the inverse limits of the inverse families of measure spaces $\left\{\left(G_{(\kappa,)p^n},\Sigma_n,\mu_n\right),\varphi_{nl}\right\}_{\NN}$.
\end{theorem}

\bigskip

Now, we verify that the inverse limit measure $\overline{\mu}_{p(,\kappa)}$ on the $\sigma$-algebra $\Sigma_{p(,\kappa)}(A)$ of $\so(d)_{p(,\kappa)}$ satisfies all the defining conditions of the Haar measure.
\begin{proposition}\label{prop:borelsalg}
For every prime $p>2$, $d\in\{2,3\}$ and $\kappa\in\{-v,p,up\}$, the $\sigma$-algebra $\Sigma_{p(,\kappa)}(A)$ [as in~\eqref{def:defSigmagenA}] coincides with the Borel $\sigma$-algebra $\cB(\so(d)_{p(,\kappa)})$ of $\so(d)_{p(,\kappa)}$. 
\end{proposition}
\begin{proof}
First, the collection of preimages of singletons of $G_{(\kappa,)p^n}$ with respect to $\pi_n$, is a topology base for $\so(d)_{p(,\kappa)}$: For every $n\in\NN$,
\begin{align}
B_{-n}(\mathcal{R}_0)&=\{\mathcal{R}\in \so(d)_{p(,\kappa)}\midd \|\mathcal{R}-\mathcal{R}_0\|_p\leq p^{-n}\}\notag \\
&=\{\mathcal{R}\in \so(d)_{p(,\kappa)}\midd \lvert (\mathcal{R}-\mathcal{R}_0)_{ij}\rvert_p \leq  p^{-n}\ \text{for every}\ i,j\}\notag\\
& =\{\mathcal{R}\in \so(d)_{p(,\kappa)}\midd\pi_n(\mathcal{R})\equiv \pi_n(\mathcal{R}_0)\mod p^n\}\notag\\
&= \pi_n^{-1}\big(\pi_n(\mathcal{R}_0)\big).\label{eq:discopreimage}
\end{align}
Actually, it is enough for $\cR_0$ to run over a countable dense subset of $\so(d)_{p(,\kappa)}$ to get a topology base, and $\so(d)_{p(,\kappa)}$ is second countable.

Every $E\in A$ of the kind $E=\pi_n^{-1}(E_n)$ for some $E_n\in \Sigma_n$ is open, as $\pi_n$ is continuous. Indeed, $E$ is a finite union of open balls in the topology base:
\beq
    E=\pi_n^{-1}(E_n)=\bigcup_{M\in E_n}\pi_n^{-1}(M)=\bigcup_{\mathcal{R}_0} B_{-n}(\mathcal{R}_0),
\eeq
where the union is performed on whatever choice of $\mathcal{R}_0\in\pi_n^{-1}(M)$ for each $M\in E_n$, and $E_n\subseteq G_{(\kappa,)p^n}$ is finite. Thus, $A$ is the collection of all finite unions of open balls in the topology base.

Now, $\Sigma_{p(,\kappa)}(A)$ is the $\sigma$-algebra generated by $A$ through countable unions of its sets and complementations. The former are countable unions of open balls of $\so(d)_{p(,\kappa)}$. But in principle, the topology of a space is generated by arbitrary unions of the sets in the topology base. Anyway, $\so(d)_{p(,\kappa)}$ is second countable, so the topology is generated by all possible (no more than) countable unions on $A$. Hence the open balls and $A$ generate the same $\sigma$-algebra: $\Sigma_{p(,\kappa)}(A)$ is the Borel $\sigma$-algebra of $\so(d)_{p(,\kappa)}$.
\end{proof}

\medskip

Proposition~\ref{prop:borelsalg} means that $\overline{\mu}_{p(,\kappa)}$ is a Borel measure. Furthermore, $\overline{\mu}_{p(,\kappa)}$ is a probability measure, finite on every set in $\Sigma_{p(,\kappa)}(A)$, and in particular on every compact set.

Since $\so(d)_{p(,\kappa)}$ is a (locally) compact, second countable and Hausdorff, every Borel measure on $\so(d)_{p(,\kappa)}$ that is finite on compact sets is regular (i.e., both outer and inner regular on all Borel sets), hence Radon (cf.\ Theorem~7.8 in~\cite{Folland99}). In particular, this implies that $\overline{\mu}_{p(,\kappa)}$ is a Radon measure on $\so(d)_{p(,\kappa)}$.

\begin{proposition}
For every prime $p>2$ and $\kappa\in\{-v,p,up\}$, the measure $\overline{\mu}_{p(,\kappa)}$ [as in~\eqref{eq:defolmu}] is both left- and right-invariant:
    \beq 
    \overline{\mu}_{p(,\kappa)}(\mathcal{R}E)=\overline{\mu}_{p(,\kappa)}(E\mathcal{R})=\overline{\mu}_{p(,\kappa)}(E),
    \eeq 
for every $E\in\Sigma_{p(,\kappa)}(A),\,\mathcal{R}\in \so(d)_{p(,\kappa)}$.
\end{proposition}
\begin{proof}
We prove that $\overline{\mu}_{p(,\kappa)}$ is left-invariant; right-invariance goes analogously (and it is implied by the left-invariance, since $\so(d)_{p(,\kappa)}$ is compact). First, we focus on the elements $E=\pi_n^{-1}(E_n)\in\Sigma_n^\ast\in A$, $n\in\NN$. Their left-translation by $\mathcal{R}\in \so(d)_{p(,\kappa)}$ is $\mathcal{R}E$, whose components are
\beq
\pi_l(\mathcal{R}E)=\pi_l(\mathcal{R})\pi_l\big(\pi_n^{-1}(E_n)\big)=\begin{cases}
    \pi_l(\mathcal{R})\varphi_{nl}^{-1}(E_n),& l> n;\\
    \pi_n(\mathcal{R})E_n,& l=n;\\
    \pi_l(\mathcal{R})\varphi_{ln}(E_n),&l< n.
\end{cases}
\eeq 
We have $\varphi_{nl}\big(\pi_l(\mathcal{R}E)\big)=\pi_n(\mathcal{R}E)=\pi_n(\mathcal{R})E_n$ for every $l>n$, that is, $\pi_l(\mathcal{R})\varphi_{nl}^{-1}(E_n)\subseteq \varphi_{nl}^{-1}\big(\pi_n(\mathcal{R})E_n\big)$. Moreover $\lvert \pi_l(\mathcal{R})\varphi_{nl}^{-1}(E_n)\rvert=\lvert \varphi_{nl}^{-1}(E_n)\rvert=\lvert \varphi_{nl}^{-1}\big(\pi_n(\mathcal{R})E_n\big)\rvert$, thus $\pi_n(\mathcal{R}E)=\varphi_{nl}^{-1}\big(\pi_n(\mathcal{R}E)\big)$ for every $l>n$. In words, the $l$-th component of $\cR E$, $l\geq n$, is the preimage of the $n$-th component with respect to $\varphi_{nl}$. It follows that $\mathcal{R}E=\pi_n^{-1}\big(\pi_n(\mathcal{R}E)\big)\in \Sigma_n^\ast \subset A$. Then $\mu(\mathcal{R}E)=\mu_n^\ast(\mathcal{R}E)=\mu_n\big(\pi_n(\mathcal{R})E_n\big) =\mu_n(E_n)=\mu(E)$ by the left-invariance of $\mu_n$ on $\Sigma_n$ under the action of $G_{(\kappa,)p^n}$.

For every $E\in \Sigma_{p(,\kappa)}(A)$, we define $\mu_{p(,\kappa)}^{(\mathcal{R})}(E)\coloneqq \overline{\mu}_{p(,\kappa)}(\mathcal{R}E)$ for a given $\mathcal{R}\in \so(d)_{p(,\kappa)}$. Clearly $\mu_{p(,\kappa)}^{(\mathcal{R})}\geq 0$ as $\overline{\mu}_{p(,\kappa)}$ is so, $\mu_{p(,\kappa)}^{(\mathcal{R})}(\emptyset)=\overline{\mu}_{p(,\kappa)}(\mathcal{R}\emptyset)=\overline{\mu}_{p(,\kappa)}(\emptyset)=0$, and, if $\{E_\lambda\}_\lambda$ is a countable family of pairwise disjoint sets in $\Sigma_{p(,\kappa)}(A)$, then $\mu_{p(,\kappa)}^{(\mathcal{R})}\left(\bigcup_\lambda E_\lambda\right)=\overline{\mu}_{p(,\kappa)}\left(\mathcal{R}\bigcup_\lambda E_\lambda\right)=\overline{\mu}_{p(,\kappa)}\left(\bigcup_\lambda \mathcal{R}E_\lambda\right)$. We see that $\{\mathcal{R}E_\lambda\}_\lambda$ is still a disjoint family of sets in $\Sigma_{p(,\kappa)}(A)$: If $E_\lambda\cap E_\nu=\emptyset$ for some $\lambda,\nu$, then $\emptyset = \mathcal{R}\emptyset=\mathcal{R}(E_\lambda\cap E_\nu)=(\mathcal{R}E_\lambda)\cap(\mathcal{R}E_\nu)$, since the group action (simple left-multiplication by a matrix) is bijective. Therefore, every $\mu_{p(,\kappa)}^{(\mathcal{R})}$ inherits the $\sigma$-additivity of $\overline{\mu}_{p(,\kappa)}$. If $E\in A$, then $\mu_{p(,\kappa)}^{(\mathcal{R})}(E)=\overline{\mu}_{p(,\kappa)}(\mathcal{R}E)=\overline{\mu}_{p(,\kappa)}(E)=\mu(E)$ by left-translation invariance of $\overline{\mu}_{p(,\kappa)}\equiv\mu$ on $A$. All of this means that $\mu_{p(,\kappa)}^{(\mathcal{R})}$ is a $\sigma$-additive extension of $\mu$ on $\Sigma_{p(,\kappa)}(A)$, for each $\mathcal{R}\in \so(d)_{p(,\kappa)}$. By uniqueness of such extension, it must be $\overline{\mu}_{p(,\kappa)}\equiv \mu_{p(,\kappa)}^{(\mathcal{R})}$ for each $\mathcal{R}\in \so(d)_{p(,\kappa)}$, providing $\mu_{p(,\kappa)}^{(\mathcal{R})}(E)=\overline{\mu}_{p(,\kappa)}(E)=\overline{\mu}_{p(,\kappa)}(\mathcal{R}E)$ for every $E\in \Sigma_{p(,\kappa)}(A),\,\mathcal{R}\in \so(d)_{p(,\kappa)}$. 
\end{proof}

\medskip

The above results are a proof of what follows.
\begin{theorem}\label{thr:HaarmSO}
The inverse-limit measure $\overline{\mu}_{p(,\kappa)}$, defined as in~\eqref{eq:defolmu} on the Borel $\sigma$-algebra $\Sigma_{p(,\kappa)}(A)=\mathcal{B}(\so(d)_{p(,\kappa)})$, is the (left and right) Haar measure on $\so(d)_{p(,\kappa)}$, for every prime $p>2$, $d\in\{2,3\}$ and $\kappa\in\left\{-v,p,up\right\}$.
\end{theorem}

\bigskip

We defined the Haar measures $\overline{\mu}_{p(,\kappa)}$ as probability measures, in fact their normalisation is
\beq \overline{\mu}_{p(,\kappa)} \big(\so(d)_{p(,\kappa)}\big) = \mu_n^\ast\big(\pi_n^{-1}(G_{(\kappa,)p^n})\big) = \mu_n(G_{(\kappa,)p^n})=\frac{\left \lvert G_{(\kappa,)p^n}\right\rvert}{\left \lvert G_{(\kappa,)p^n}\right\rvert}=1.
\eeq 
As an application, we provide the Haar measure of the open balls in the topology base for $\so(d)_{p(,\kappa)}$ already considered above. 

\begin{proposition}\label{prop:misbasinvlim}
Let $p>2$ be a prime, $d\in\{2,3\}$ and $\kappa\in\{-v,p,up\}$. For every $\cR_0\in \so(d)_{p(,\kappa)}$ and every $n\in\NN$, 
\begin{align}  
d=2:\qquad &\overline{\mu}_{p,\kappa}\big(B_{-n}(\cR_0)\big)=\frac{p^{-n}}{2},\ \   \kappa\in\left\{p,up\right\}, \qquad 
\overline{\mu}_{p,-v}\big(B_{-n}(\cR_0)\big) = \frac{p^{1-n}}{p+1},\label{eq:cardbocgpkk}\\
d=3:\qquad &\overline{\mu}_p\big(B_{-n}(\cR_0)\big) = \frac{p^{1-3n}}{2(p+1)}.
\end{align}
\end{proposition}
\begin{proof}
Eq.~\eqref{eq:discopreimage} is $B_{-n}(\cR_0) =\pi_n^{-1}\big(\pi_n(\mathcal{R}_0)\big)\in \Sigma_n^\ast$ for every $n\in\NN$. Therefore, the Haar measure $\overline{\mu}_{p(,\kappa)}$ on $\so(d)_{p(,\kappa)}$ [cf.~\eqref{eq:defolmu}] reduces to the measure~\eqref{eq:measulivelli} on $\Sigma_n^\ast$. We conclude that, for every $\cR_0\in \so(d)_{p(,\kappa)}$,
\beq 
\overline{\mu}_{p(,\kappa)}\big(B_{-n}(\cR_0)\big)=\mu_n^\ast \big(\pi_n^{-1}\big(\pi_n(\mathcal{R}_0)\big)\big) = \mu_n\big(\pi_n(\mathcal{R}_0)\big) = \frac{\left\lvert \pi_n(\mathcal{R}_0)\right\rvert}{\left\lvert G_{(\kappa,)p^n}\right\rvert} = \left\lvert G_{\kappa,p^n}\right\rvert^{-1}.
\eeq 
The value of these measures is given by Propositions~\ref{prop:cardparam},~\ref{prop:ordinigrup3dmodpk}. 
\end{proof}

\section{Discussion}\label{sec:conclusion}
This work is inspired by Volovich's original idea~\cite{Volovich1} that the existence of a shortest measurable length --- i.e., the so-called Planck length --- entails a non-Archimedean structure of spacetime. According to this hypothesis, at the Planck regime spacetime does not consist of infinitely divisible intervals, but only of isolated points, which essentially results into a totally disconnected topological structure. Pursuing this idea to its logical conclusions naturally leads to the exploration of $p$-adic models of quantum mechanics. Within this framework, in~\cite{our1st,our2nd} we have begun to develop a theory of angular momentum and spin via a thorough study of the geometric features of the special orthogonal groups $\so(2)_{p,\kappa}$ and $\so(3)_p$. 

In the present contribution, our main aim was to provide a construction of the invariant measure (Haar measure) on the compact two- and three-dimensional $p$-adic rotation groups. In particular, this effort serves a dual purpose: On the one hand, it enables us to study the irreducible projective unitary  representations (via the Peter-Weyl theorem) of the special orthogonal groups in dimensions $2$ and $3$; on the other hand, it paves the way for the study of $p$-adic qubit models, which ultimately fit into our ideal program devoted to the foundation of a $p$-adic theory of quantum information.

The strategy we followed in this work essentially relies on the observation that, as $\so(2)_{p,\kappa}$ and $\so(3)_p$ are profinite groups, they are isomorphic 
to the inverse limit of an inverse family of finite groups. Over the latter, one considers an inverse family of Haar measure spaces, to construct the inverse limit measure and to prove that it is the Haar measure on the inverse limit groups $\so(2)_{p,\kappa}$ and $\so(3)_p$. This strategy is known to be generalisable to all profinite~\cite{profinitem} groups, and also to all locally compact groups~\cite{BourInt}. Our main aim was to obtain a concrete result especially in the case of the compact $p$-adic rotation groups, because of their remarkable role in the context of $p$-adic quantum mechanics, and, moreover, to provide an explicit determination of the Haar measure on these groups. This is achieved by knowing the order of the finite quotients $G_{(\kappa,)p^n}= \so(d)_{p(,\kappa)}\mod p^n$. To this end, we provided a parametrisation of $G_{(\kappa,)p^n}$, together with an interesting characterisation via a multivariable Hensel lifting of roots. These tools are also useful in the study of the representations of $\so(2)_{p,\kappa}$ and $\so(3)_p$ which factorise on some quotient $G_{(\kappa,)p^n}$~\cite{our2nd}.

It is worth remarking that the inverse limit strategy is not the only approach one can pursue to determine the Haar measure on the $p$-adic rotation groups. Indeed, very recently, a general formula for the Haar measure on every (locally compact, second countable, Hausdorff) $p$-adic Lie group was obtained in~\cite{our3rd}. The Haar measure on $\so(2)_{p,\kappa}$, and the Haar integral of $\so(3)_p$ are then derived by means of a suitable application of this general formula. For the group $\so(2)_{p,\kappa}$ --- where explicit calculations of $p$-adic Haar integrals can be carried out --- we have verified in Appendix~\ref{sec:equivmisure2} that the Haar measures obtained in these two different approaches do coincide (up to a positive multiplicative constant, due to normalisation, by essential uniqueness of the Haar measure).

\medskip

In this work, we focused on odd primes $p>2$. Now we want to describe the special case of even prime $p=2$, which exhibits some peculiarities. To start with, we recall that there is a unique definite quadratic form $Q_+$ on $\QQ_2^3$, associated with a unique compact $2$-adic special orthogonal group $\so(3)_2<\mathrm{SL}(3,\ZZ_2)$  defined by $Q_+$~\cite{our1st}. Moreover, there are seven (rather than three) definite quadratic forms $Q_\kappa$ on $\QQ_2^2$ --- labelled by their determinant $\kappa\in\{1,\pm2,\pm5,\pm10\}$ --- yielding seven compact $2$-adic groups $\so(2)_{2,\kappa}$.

Let us focus on the bidimensional case first. 
Here, we observe that while $\so(2)_{2,\kappa}<\mathrm{SL}(2,\ZZ_2)$, for all $\kappa\in\{1,\pm2,5,\pm10\}$, the case $\kappa=-5$ presents an exception. Indeed, by means of an argument similar to the proof of Corollary~\ref{cor:tuttisolicomp}, one can prove that $\so(2)_{2,-5}<\mathrm{SL}(2,2^{-1}\ZZ_2)$. In contrast to $\ZZ_2$, $2^{-1}\ZZ_2$ is not a ring but an additive group and, hence, the maps $\phi_{kn}\colon 2^{-1}\ZZ_2/2^n\ZZ_2\rightarrow \ZZ_2/2^k\ZZ_2$ in Eq.~\eqref{eq:homoaddno}, and $\cproj_k$ in Eq.~\eqref{eq:homocontproj}, are \emph{never} ring homomorphisms. Indeed, Theorem~\ref{prop:liminvfacile} provides an inverse limit of topological groups for $\so(2)_{2,\kappa}$ for all $\kappa\in\{1,\pm2,5,\pm10\}$, but \emph{not} for $\so(2)_{2,-5}$. On the other hand, a way to characterise $\so(2)_{2,-5}$ as an inverse limit of discrete finite groups is to inject $\so(2)_{2,-5}$ as a subgroup of $\so(3)_2$; namely, there exists $\mathbf{n}\in\QQ_2^3\setminus\{\boldsymbol{0}\}$ such that ${Q_+}_{\lvert\mathbf{n}^\perp}$ is equivalent to $Q_{-5}$ (cf.\ Proposition 21 in~\cite{our1st}), and $\so(2)_{2,-5}$ is the restriction to the orthogonal complement $\mathbf{n}^\perp$ of the abelian subgroup in $\so(3)_2$ of rotations around $\QQ_2\mathbf{n}$ with respect to an orthogonal basis $(\mathbf{g},\mathbf{h},\mathbf{n})$. As the entries of the matrices in $\so(3)_2$ are $2$-adic integers, a change of basis from $(\mathbf{g},\mathbf{h},\mathbf{n})$ to the canonical one in $\QQ_2^3$ provides a group $\so(2)_{2,-5}'<\so(3)_2$ isomorphic to $\so(2)_{2,-5}$, consisting of $3\times 3$ matrices with $2$-adic integer entries, for which the inverse limit in Theorem~\ref{prop:liminvfacile} holds true. At this point, the above construction of the Haar measure as an inverse limit of discrete measure spaces works on $\so(2)_{2,\kappa}$, $\kappa\in\{1,\pm2,5,\pm10\}$, and on $\so(2)_{2,-5}'$. Finally, the Haar measure on $\so(2)_{2,-5}'$ is transferred to the Haar measure on $\so(2)_{2,-5}$ by means of the pushforward via their topological group isomorphism. One uses a similar parametrisation as in Remark~\ref{rem:paraminter} to calculate the orders of $G_{\kappa,2^n}$, $\kappa\in\{1,\pm2,5,\pm10\}$, and of $\so(2)_{2,-5}'\mod 2^n$, for $n\in\NN$, so as to explicitly be able to evaluate the Haar measure of clopen balls as in Proposition~\ref{prop:misbasinvlim}. Again, these values are consistent with those computed by the normalised integral Haar measure as in Appendix~\ref{sec:equivmisure2}.

Moving to the three-dimensional case, $\so(3)_2<\mathrm{SL}(3,\ZZ_2)$ is still characterised as an inverse limit of topological groups as in Theorem~\ref{prop:liminvfacile}, since its matrix entries are in the ring $\ZZ_2$. Therefore, a measure $\overline{\mu}_2$ constructed as in~\eqref{eq:defolmu} provides again the (left and right) Haar measure on $\so(3)_2$. However, the evaluation of this measure on a Borel set of $\so(3)_2$ requires knowing the order of the groups $G_{2^n}$. This is a hard task, since none of the possible forms of principal (Euler or Cardano) \textquotedblleft angle\textquotedblright\ decomposition,
familiar from the real Euclidean case, exists for $p = 2$ (Remark~28 in~\cite{our1st}). As already seen, an alternative approach is through a multivariable Hensel lifting of roots: One defines $\widetilde{G}_{2^n}$ as the group of solutions modulo $2^n$ of the defining conditions of $\so(3)_2$ as in Eq.~\eqref{eq:altrogrequiv}, and studies whether or not this coincides with $G_{2^n}$. If this was the case, then the order of $G_{2^{n+1}}$ would be obtained from the number of liftings of each element of $\widetilde{G}_{2^n}$ to $\widetilde{G}_{2^{n+1}}$ (cf.\ Appendix~\ref{sec:appliftalaHens}). However, this is not the case, as one finds counterexamples of elements in $\widetilde{G}_{2^n}$ that do not lift to elements in $\widetilde{G}_{2^{n+1}}$. As expected algebraically, the above discussion shows that the situation for $p=2$ is peculiarly different from that for odd primes $p>2$ (see $\kappa=-5$ in two dimensions, or the non-existence of principal \textquotedblleft angle\textquotedblright\ decompositions in three dimensions and the failure of the Hensel lifting strategy). This circumstance is not evident in the $p$-adic Lie group approach discussed in~\cite{our3rd}.

Concerning higher dimensions, Theorem~6 at pp.~36-37 of~\cite{serre} states that no quadratic form on $\QQ_p^d$ is definite for $d \geq5$, for every prime $p\geq2$; hence, the only remaining case is $d=4$. By Corollary at p.~39 of~\cite{serre}, there is a unique definite quadratic form on $\QQ_p^4$ for every prime $p\geq2$, say $Q_+^{(4)}$ as in~\cite{our1st}, leading to one compact group $\so(4)_p$. One can show that $\so(4)_p<\mathrm{SL}(4,\ZZ_p)$ for $p>2$, for which the same inverse limit of groups $\so(4)_p\mod p^n$ as in Theorem~\ref{prop:liminvfacile} holds, while the entries of the matrices of $\so(4)_2$ are in $2^{-1}\ZZ_2$. Hence, for $p=2$ we still lack an inverse family of discrete finite groups whose inverse limit is isomorphic to $\so(4)_2$ --- as for $\so(2)_{2,-5}$. In general, this is given by an inverse family of quotient groups by a decreasing directed family of compact normal subgroups whose intersection is the identity~\cite{BourInt}. The last step is to calculate the order of those finite groups, for every prime $p\geq2$. In conclusion, the proposed construction of the Haar measure as an inverse limit also applies to $\so(4)_p$, and, given the above ingredients, it will be also explicitly computable.


\section*{Acknowledgments}
IS is supported by the Istituto Nazionale di Fisica Nucleare (INFN), Sezione di Perugia, and by the Spanish MICIN (project 
PID2022-141283NB-I00) with the support of FEDER funds. AW is supported by the European Commission QuantERA grant ExTRaQT
(Spanish MICIN project PCI2022-132965), by the Spanish MICIN
(project PID2022-141283NB-I00) with the support of FEDER funds,
by the Spanish MICIN with funding from European Union NextGenerationEU
(PRTR-C17.I1) and the Generalitat de Catalunya, by the Spanish MINECO
through the QUANTUM ENIA project: Quantum Spain, funded by the European
Union NextGenerationEU within the framework of the ``Digital Spain
2026 Agenda'', by the Alexander von Humboldt Foundation, and the
Institute for Advanced Study of the Technical University Munich.

\appendix

\section{Inverse limit characterisation of \texorpdfstring{$\mathbb{Q}_p$}{Lg} and of its closed subgroups}\label{sec:appinvlimQp}
This appendix section contains the proof of Proposition~\ref{prop:invlimQpappen}. To this end, for every prime $p\geq2$, we will always have in mind the bijective correspondences of $\QQ_p$ with $\ZZ/p\ZZ((p))$ and of $\ZZ_p$ with $\ZZ/p\ZZ[[p]]$, i.e., we will always write the elements of $\QQ_p$ and $\ZZ_p$ as formal Laurent and power series of $p$ respectively. Indeed, the following argument is inspired by Exercise~(3) at p.~65 of~\cite{Fuchs}, and by Exercise~5.25 at p.~255 of~\cite{Rotman}.

\medskip

We start by proving
\beq \label{eq:invlimQp}
\QQ_p\simeq \varprojlim\left\{\QQ_p/p^n\ZZ_p,\,\phi_{nl}\right\}_{\NN},
\eeq 
where each of the maps $\phi_{nl}$, $n\leq l$, is explicitly given by
\begin{align}
\phi_{nl}\colon &\QQ_p/p^l\ZZ_p\rightarrow \QQ_p/p^n\ZZ_p,
\nonumber\\
&x_{-j}p^{-j}+\dots+ x_{l-1}p^{l-1}+p^l\ZZ_p\mapsto x_{-j}p^{-j}+\dots +x_{n-1}p^{n-1}+p^n\ZZ_p,
\end{align}
and it is well defined, with $p^l\ZZ_p\subseteq p^n\ZZ_p$ for $n\leq l$, and then $\QQ_p/p^n\ZZ_p\subseteq \QQ_p/p^l\ZZ_p$.  

First, we observe that $\{\QQ_p/p^n\ZZ_p\}_{\NN}$ is a family of groups, since $p^n\ZZ_p$ is a normal subgroup of $\QQ_p$ (because $\QQ_p$ is an additive abelian group). Also, $\{\phi_{nl}\}_{\NN}$ is a family of group homomorphisms: For every $\sum_{k=-i}^{l-1}x_kp^k+p^l\ZZ_p,\, \sum_{k=-j}^{l-1}y_kp^k+p^l\ZZ_p \in \QQ_p/p^l\ZZ_p$, we have
\begin{align}
\phi_{nl}\left(\sum_{k=-i}^{l-1}x_kp^k+\sum_{k=-j}^{l-1}y_kp^k+p^l\ZZ_p\right)
&=\phi_{nl}\left(\sum_{k=-\max\{i,j\}}^{l-1}z_kp^k+p^l\ZZ_p\right)\nonumber\\
& = \sum_{k=-\max\{i,j\}}^{n-1}z_kp^k+p^n\ZZ_p,
\end{align}
where $z_k$ takes into account \textquotedblleft carryings\textquotedblright, that is, that $x_k+y_k\in \ZZ/p\ZZ$ and possible multiples of $p$ from the sum $x_k+y_k$ in $\ZZ$ contribute to the coefficient of $p^{k+1}$, and
\begin{align}
    \phi_{nl}\left(\sum_{k=-i}^{l-1}x_kp^k+p^l\ZZ_p\right)+\phi_{nl}\left(\sum_{k=-j}^{l-1}y_kp^k+p^l\ZZ_p\right) &= \sum_{k=-i}^{n-1}x_kp^k+\sum_{k=-j}^{n-1}y_kp^k+p^n\ZZ_p\nonumber\\
    & = \sum_{k=-\max\{i,j\}}^{n-1}z_kp^k+p^n\ZZ_p.
\end{align}
Clearly $\phi_{nn}=id$ (any Laurent series $\mod p^n\ZZ_p$ truncated $\mod p^n\ZZ_p$ is itself) and $\phi_{nl}=\phi_{nm}\circ \phi_{ml}$, for every $n\leq m\leq l,\ n,m,l\in\NN$ (truncating $\mod p^n\ZZ_p$ any Laurent series $\mod p^l\ZZ_p$ is equal to first truncating it $\mod p^m$ and then again $\mod p^n\ZZ_p$).

We denote by $\cproj_n$ the canonical projection 
\beq \label{eq:absuprib}
\cproj_n\colon\QQ_p\rightarrow\QQ_p/p^n\ZZ_p,\quad x\mapsto x+p^n\ZZ_p;
\eeq 
we consider the quotient topology on $\QQ_p/p^n\ZZ_p$, whose open sets are those $\mathcal{V}\subseteq \QQ_p/p^n\ZZ_p$ such that $\cproj_n^{-1}(\mathcal{V})$ are open in $\QQ_p$ (with the $p$-adic ultrametric topology). By definition, $\cproj_n$ is continuous. Any quotient group with quotient topology is a topological group, so is $\QQ_p/p^n\ZZ_p$ --- and the translation map $T_{a+p^n\ZZ_p}\colon \QQ_p/p^n\ZZ_p\rightarrow \QQ_p/p^n\ZZ_p$, $x+p^n\ZZ_p\mapsto (x+a)+p^n\ZZ_p$ is a homeomorphism for every $a\in\QQ_p$. We further show that the quotient topology on $\QQ_p/p^n\ZZ_p$ coincides with its discrete topology, by showing that singletons are open: $\cproj_n^{-1}\left(\{0+p^n\ZZ_p\}\right)=\{x\in\QQ_p\midd \cproj_n(x)=0+p^n\ZZ_p\} = \{x\in\QQ_p\midd |x|_p< p^{-n+1}\}=p^n\ZZ_p$ is an open ball; then, for every $x\in\QQ_p$, we have $\{x+p^n\ZZ_p\} = T_{x+p^n\ZZ_p}(0+p^n\ZZ_p)$, which is open since $0+p^n\ZZ_p$ is so and $T_{x+p^n\ZZ_p}$ is a homeomorphism. Moreover, each $\phi_{nl}$ is a continuous group homomorphism, as its domain is supplied with discrete topology. All of this proves that $\left\{\QQ_p/p^n\ZZ_p,\,\phi_{nl}\right\}_{\NN}$ is an inverse family of topological groups.

By Definition~\ref{def:invlimsetgroup}, an element in the inverse limit group in~\eqref{eq:invlimQp} is a sequence $\left(x_{(n)}+p^n\ZZ_p\right)_{n\in\NN}\in \prod_n\left(\QQ_p/p^n\ZZ_p\right)$ such that $x_{(n)}+p^n\ZZ_p=\phi_{nl}(x_{(l)}+p^l\ZZ_p)$ for all $n\leq l$, i.e., such that $x_{(n)}+p^n\ZZ_p=\phi_{n,n+1}(x_{(n+1)}+p^{n+1}\ZZ_p)$ for all $n\in\NN$ (since the index set $\NN$ is totally ordered). Last condition can be equivalently rewritten as 
\beq \label{eq:invlimcoscon}
x_{(n+1)}\equiv x_{(n)}\mod p^n\ZZ_p,
\eeq 
which implies $\lim\limits_{n\rightarrow\infty}\left\lvert x_{(n+1)}- x_{(n)}\right\rvert_p\leq \lim\limits_{n\rightarrow\infty}p^{-n}=0$. This means that $\left(x_{(n)}\right)_{n\in\NN}$ is a Cauchy sequence in $\QQ_p$, the latter being a complete space (once supplied with the $p$-adic metric). Hence, the Cauchy sequence converges in $\QQ_p$, say $\lim\limits_{n\rightarrow\infty}x_{(n)}\coloneqq x\in \QQ_p$. Eventually, $\varprojlim\left\{\QQ_p/p^n\ZZ_p,\,\phi_{nl}\right\}_{\NN}$ is a topological group, once endowed with the subspace topology of the product topology on $\prod_n\left(\QQ_p/p^n\ZZ_p\right)$, where each $\QQ_p/p^n\ZZ_p$ has a discrete topology.

We move to prove that the two topological groups in Eq.~\eqref{eq:invlimQp} are indeed isomorphic. We introduce the following map,
\beq \label{eq:fisotop}
f\colon\QQ_p\rightarrow \varprojlim\left\{\QQ_p/p^n\ZZ_p,\,\phi_{nl}\right\}_{\NN},\qquad x\mapsto \left(x+ p^n\ZZ_p\right)_n,
\eeq 
which is a group homomorphism: $f(x+y)=\left((x+y)+p^n\ZZ_p\right)_n = \left(x+p^n\ZZ_p\right)_n+\left(y+p^n\ZZ_p\right)_n=f(x)+f(y)$ for every $x,y\in\QQ_p$. Furthermore, we prove that $f$ is bijective. Consider the map \beq 
    g\colon \varprojlim\left\{\QQ_p/p^n\ZZ_p,\,\phi_{nl}\right\}_{\NN}\rightarrow \QQ_p, \qquad  \left(x_{(n)}+p^n\ZZ_p\right)_n\mapsto \lim_{n\rightarrow\infty}x_{(n)},
\eeq 
which is well defined: Suppose that another set of representatives in $\QQ_p$ is considered for the same element $\left(x_{(n)}+p^n\ZZ_p\right)_n$ in the inverse limit, say $\left(y_{(n)}\right)_n$ such that $y_{(n)}\equiv x_{(n)}\mod p^n\ZZ_p$. The limit of this Cauchy sequence is the same: $\lim\limits_{n\rightarrow\infty} \left\lvert y_{(n)}-x\right\rvert_p = \lim\limits_{n\rightarrow\infty}\left\lvert y_{(n)}-x_{(n)}+x_{(n)} -x\right\rvert_p\leq \lim\limits_{n\rightarrow\infty}\max\left\{\left\lvert y_{(n)}-x_{(n)}\right\rvert_p,\,\left\lvert x_{(n)}-x\right\rvert_p\right\}=0$ by the strong triangle inequality, i.e., $\lim\limits_{n\rightarrow\infty}y_{(n)}=\lim\limits_{n\rightarrow\infty}x_{(n)}=x$. On the one hand, for every $x\in\QQ_p$, 
$g\big(f(x)\big) = g\left(\left(x+ p^n\ZZ_p\right)_n\right)=\lim\limits_{n\rightarrow\infty}x=x$. On the other hand, for every $\left(x_{(n)}+ p^n\ZZ_p\right)_n\in \varprojlim\left\{\QQ_p/p^n\ZZ_p,\,\phi_{nl}\right\}_{\NN}$, $f\left(g\left(\left(x_{(n)}+ p^n\ZZ_p\right)_n\right)\right) = f(x) = \left(x+ p^n\ZZ_p\right)_n$. Condition~\eqref{eq:invlimcoscon} implies $x_{(l)}\equiv x_{(n)}\mod p^n\ZZ_p$ for all $l\geq n$, hence $\lim\limits_{l\rightarrow\infty}x_{(l)}  \equiv \lim\limits_{l\rightarrow\infty}x_{(n)}\mod p^n\ZZ_p$, that is $x\equiv x_{(n)}\mod p^n\ZZ_p$, for all $n\in\NN$. In conclusion $f\left(g\left(\left(x_{(n)}+ p^n\ZZ_p\right)_n\right)\right) = \left(x_{(n)}+ p^n\ZZ_p\right)_n$, and we proved that $f$ and $g$ are inverse of each other.

The map $f$ is continuous by construction: $\prod_n \cproj_n\colon \QQ_p\rightarrow \prod_n\left(\QQ_p/p^n\ZZ_p\right)$ is continuous since all its components $\cproj_n$ are so (indeed, the product topology on $\prod_n\left(\QQ_p/p^n\ZZ_p\right)$ is the coarsest topology for which all the projections on the factors are continuous), and $f$ is $\prod_n \cproj_n$ whose codomain is restricted to its image with subspace topology. Finally, we prove that $f^{-1}=g$ is continuous, by showing that the preimage of any base set of $\QQ_p$ is open. We first consider $D_k(0)$ for $k<0$, and get
\begin{align}
f\big(D_k(0)\big) &= \left\{\left(x+p^n\ZZ_p\right)_n\midd x\in D_k(0)=p^{-k}\ZZ_p\right\}\nonumber\\
& = \left\{\left(p^{-k}x_{-k}+\dots+p^n\ZZ_p\right)_n\in\varprojlim\left\{\QQ_p/p^n\ZZ_p,\,\phi_{nl}\right\}_{\NN}\right\}  \nonumber\\
& = \varprojlim\left\{\QQ_p/p^n\ZZ_p,\,\phi_{nl}\right\}_{\NN}\cap \left(\prod_{n=1}^{-k}\left(0+p^n\ZZ_p\right)\times \prod_{n>-k}\left(\QQ_p/p^n\ZZ_p\right)\right),
\end{align}
which is open on the subspace topology of $\varprojlim\left\{\QQ_p/p^n\ZZ_p,\,\phi_{nl}\right\}_{\NN}$, as $\prod_{n=1}^{-k}\left\{0+p^n\ZZ_p\right\}\times \prod_{n>-k}\left(\QQ_p/p^n\ZZ_p\right)$ is open in the product topology of $\prod_n\left(\QQ_p/p^n\ZZ_p\right)$ --- it is the product of a finite number of singletons (open in the discrete topology) times infinitely many whole spaces $\QQ_p/p^n\ZZ_p$. For any other open set $D_k(c)$, $k<0$, $c\in\QQ\subset\QQ_p$, we have $D_k(c) = t_c\big(D_k(0)\big)$ where $t_c\colon\QQ_p\rightarrow\QQ_p$, $x\mapsto x+c$ is the homeomorphism of translation in $\QQ_p$, for every $c$; thus $D_k(c)$ is open. We have proved that $f$, as in Eq.~\eqref{eq:fisotop}, is an isomorphism of topological groups.

\bigskip

The same argument can be repeated by replacing each assurance of $\QQ_p$ with any of its proper closed subgroups $p^m\ZZ_p$, $m\in\ZZ$, considered with subspace ultrametric topology. We just point out that $p^m\ZZ_p$ is complete, being a closed subspace of the complete space $\QQ_p$; in fact, a Cauchy sequence in $p^m\ZZ_p\subset \QQ_p$ converges in $\QQ_p$, and its limit actually belongs to $p^m\ZZ_p$. Lastly, here the inverse family is indexed by $\ZZ_{>m}$, to ensure that the group $p^k\ZZ_p$ we are quotienting by is a proper normal subgroup of $p^m\ZZ_p$, even in the case $m>0$. This concludes the proof of
\beq 
p^m\ZZ_p\simeq \varprojlim\left\{p^m\ZZ_p/p^n\ZZ_p,\,\phi_{nl}\right\}_{\ZZ_{>m}},
\eeq 
that is last statement of Proposition~\ref{prop:invlimQpappen}.

\section{Cardano decomposition of \texorpdfstring{$G_{p^n}=\pi_n\big(\so(3)_p\big)$}{Lg}}\label{sec:app3}
In order to calculate the Haar measure of Borel sets of $\so(3)_p$, for every prime $p>2$, we need to know the orders of the projected groups $G_{p^n}$, for every $n\in\NN$. This can be achieved by exploiting the projection via $\pi_n$ of the Cardano representation of $\so(3)_p$ (Theorem~\ref{theor:cardanorep}). 

Let $G_{\mathbf{n},p^n}\coloneqq \pi_n\left(\so(3)_{p,\mathbf{n}}\right)<G_{p^n}$, for every $n\in\NN$, $\mathbf{n}\in\QQ_p^3\setminus\{\boldsymbol{0}\}$.  We exploit Eqs.~\eqref{eq:paramSO3xy},~\eqref{eq:paramSO3z}, writing rotations around the reference axes as in Eq.~\eqref{eq:refaxisSO3}:
\begin{align}
    G_{x,p^n} =& \left\{\mathcal{R}_x(\xi)\mod p^n\midd \xi\in\ZZ/p^n\ZZ\right\}\cup \left\{\cR_x(\infty)\mathcal{R}_x(\xi)\mod p^n\midd \xi\in\ZZ/p^n\ZZ\right\}\\
    = & \left\{\begin{pmatrix}
    1&0&0\\0&a(\xi) & \frac{p}{v}c(\xi) \\
    0 & c(\xi)
    & a(\xi)
\end{pmatrix}\mod p^n\midd \xi\in\ZZ/p^n\ZZ\right\}\nonumber\\
    &\cup \left\{\begin{pmatrix}
    1&0&0\\0&-a(\xi) & -\frac{p}{v}c(\xi) \\
    0 & -c(\xi)
    & -a(\xi)
\end{pmatrix}\mod p^n\midd \xi\in\ZZ/p^n\ZZ\right\},\nonumber\\
G_{y,p^n}=& \left\{\mathcal{R}_y(\eta)\mod p^n\midd \eta\in\ZZ/p^n\ZZ\right\}\cup \left\{\cR_y(\infty)\mathcal{R}_y(\eta)\mod p^n\midd \eta\in\ZZ/p^n\ZZ\right\}\\
= &  \left\{\begin{pmatrix}
    e(\eta) & 0&-pg(\eta)\\0&1&0\\
    g(\eta) & 0 & e(\eta)
\end{pmatrix}\mod p^n\midd \eta\in\ZZ/p^n\ZZ\right\} \nonumber\\
    &\cup \left\{\begin{pmatrix}
    -e(\eta) & 0&pg(\eta)\\0&1&0\\
    -g(\eta) & 0 & -e(\eta)
\end{pmatrix}\mod p^n\midd \eta\in\ZZ/p^n\ZZ\right\},\nonumber\\
G_{z,p^n} =& \left\{\mathcal{R}_z(\zeta)\mod p^n\midd \zeta\in\ZZ/p^n\ZZ\right\}\cup\left\{\mathcal{R}_z(\infty)\mathcal{R}_z(\zeta)\mod p^n\midd \zeta\in p(\ZZ/p^n\ZZ)\right\}\\
= & \left\{\begin{pmatrix}
    l(\zeta) & vm(\zeta) & 0\\m(\zeta) & l(\zeta) & 0\\0&0&1
\end{pmatrix}\mod p^n\midd \zeta\in\ZZ/p^n\ZZ\right\}\nonumber\\
    &\cup\left\{\begin{pmatrix}
    -l(\zeta) & -vm(\zeta) & 0\\-m(\zeta) & -l(\zeta) & 0\\0&0&1
\end{pmatrix}\mod p^n\midd \zeta\in p(\ZZ/p^n\ZZ)\right\},\nonumber
\end{align}
where
\begin{align}
&a(\xi)\coloneqq\frac{1+\frac{p}{v}\xi^2}{1-\frac{p}{v}\xi^2},\quad c(\xi)\coloneqq\frac{2\xi}{1-\frac{p}{v}\xi^2},\\
& e(\eta)\coloneqq \frac{1-p\eta^2}{1+p\eta^2},\quad g(\eta)\coloneqq\frac{2\eta}{1+p\eta^2},\\ & l(\zeta)\coloneqq \frac{1+v\zeta^2}{1-v\zeta^2},\quad m(\zeta)\coloneqq \frac{2\zeta}{1-v\zeta^2}.
\end{align}
We will refer to the first set of each of these three unions as the \textquotedblleft first branch\textquotedblright\ and to the second one as the \textquotedblleft second branch\textquotedblright. This has set the bases for the following proof of Theorem~\ref{theor:cardanodoppiamod}.

\begin{proof}
Theorem~\ref{theor:cardanorep} states that every matrix $L$ in $\so(3)_p$ has exactly two distinct Cardano decompositions of the kind $\cR_x\cR_y\cR_z$. Since $\pi_n$ is a group homomorphism, also every $\pi_n(L)\in G_{p^n}$ can be written in at least two distinct compositions of the kind $\cR_x\cR_y\cR_z\mod p^n$, for every $n\in\NN$. One could ask if this Cardano representation for $G_{p^n}$ is exactly twofold, or if there are more than two distinct triples $(\cR_x\mod p^n,\cR_y\mod p^n,\cR_z\mod p^n)\in G_{x,p^n}\times G_{y,p^n}\times G_{z,p^n}$ whose products give the same $L\mod p^n$. Theorem~\ref{theor:cardanodoppiamod} states that the answer is no, and this is what we are going to prove. We shall analyse all the possibilities for the branches of the three rotations involved in the Cardano representation. A triple $ijk$ with $i,j,k\in\{1,2\}$ will denote a Cardano composition $\cR_x\cR_y\cR_z$, where $\cR_x$, $\cR_y$, $\cR_z\mod p^n$ are taken from the $i$-th, $j$-th, and $k$-th branch, respectively. There are $2^3=8$ possible Cardano compositions $\cR_x\cR_y\cR_z$ with respect to the branches of each of the three involved rotations. Thus, there are $36$ possibilities for equating two Cardano compositions modulo $p^n$. This same procedure was already developed for $n=1$ in \cite{our2nd}, and here we generalise the proof for all $n\in\NN$. We shall use the fact that $a(\xi),e(\eta)\equiv1\not\equiv0\mod p$ are invertible in $\ZZ/p^n\ZZ$, for all $\xi,\eta\in\ZZ/p^n\ZZ$.

We start with $111\equiv222$, i.e., 
\beq 
\cR_x(\xi)\cR_y(\eta)\cR_z(\zeta)\equiv \cR_x(\infty)\cR_x(\xi') \cR_y(\infty)\cR_y(\eta')\cR_z(\infty)\cR_z(\zeta')\mod p^n,
\eeq 
for some $\xi,\xi',\eta,\eta',\zeta,\zeta'\in\ZZ/p^n\ZZ$, $\zeta'\equiv0\mod p$. This is
\beq 
\begin{aligned}
&\begin{pmatrix}
    e(\eta)l(\zeta)&e(\eta)vm(\zeta)&-pg(\eta)\\
    a(\xi)m(\zeta) +\frac{p}{v}c(\xi)g(\eta)l(\zeta) &a(\xi)l(\zeta) +\frac{p}{v}c(\xi)g(\eta)vm(\zeta)&\frac{p}{v}c(\xi)e(\eta)\\
    c(\xi)m(\zeta)+ a(\xi)g(\eta)l(\zeta)& c(\xi)l(\zeta)+ a(\xi)g(\eta)vm(\zeta)&a(\xi)e(\eta)
  \end{pmatrix}\equiv\\
  &\equiv \begin{pmatrix}
    e(\eta')l(\zeta')&e(\eta')vm(\zeta')&pg(\eta')\\
    a(\xi')m(\zeta')-\frac{p}{v}c(\xi')g(\eta')l(\zeta') &a(\xi')l(\zeta')-\frac{p}{v}c(\xi')g(\eta')vm(\zeta')&\frac{p}{v}c(\xi')e(\eta')\\
    c(\xi')m(\zeta')-a(\xi')g(\eta')l(\zeta')& c(\xi')l(\zeta')-a(\xi')g(\eta')vm(\zeta')&a(\xi')e(\eta')
  \end{pmatrix}\mod p^n,
\end{aligned}
\eeq 
and we get the following system of modular congruences:
\begin{align}
    &e(\eta)l(\zeta)\equiv e(\eta')l(\zeta')\mod p^n;\label{eq:entry11}\\
    &e(\eta)vm(\zeta)\equiv e(\eta')vm(\zeta')\mod p^n;\label{eq:entry12}\\
    &-pg(\eta)\equiv pg(\eta')\mod p^n;\label{eq:entry13}\\
    &a(\xi)m(\zeta) +\frac{p}{v}c(\xi)g(\eta)l(\zeta)\equiv a(\xi')m(\zeta')-\frac{p}{v}c(\xi')g(\eta')l(\zeta')\mod p^n;\label{eq:entry21}\\
    & a(\xi)l(\zeta) +\frac{p}{v}c(\xi)g(\eta)vm(\zeta)\equiv a(\xi')l(\zeta')-\frac{p}{v}c(\xi')g(\eta')vm(\zeta')\mod p^n;\label{eq:entry22}\\
    &\frac{p}{v}c(\xi)e(\eta)\equiv\frac{p}{v}c(\xi')e(\eta')\mod p^n;\label{eq:entry23}\\
    &c(\xi)m(\zeta)+ a(\xi)g(\eta)l(\zeta)\equiv c(\xi')m(\zeta')-a(\xi')g(\eta')l(\zeta')\mod p^n;\label{eq:entry31}\\
    & c(\xi)l(\zeta)+ a(\xi)g(\eta)vm(\zeta)\equiv c(\xi')l(\zeta')-a(\xi')g(\eta')vm(\zeta')\mod p^n;\label{eq:entry32}\\
    & a(\xi)e(\eta)\equiv a(\xi')e(\eta')\mod p^n.\label{eq:entry33}
\end{align}
Eq.~\eqref{eq:entry13} is equivalent to $p(\eta+\eta')(1+p\eta\eta')\equiv0\mod p^n$. For $n>1$, it provides $\eta'\equiv-\eta\mod p^{n-1}$, and it follows that $p{\eta'}^2\equiv p\eta^2\mod p^n$ and $e(\eta')\equiv e(\eta)\mod p^n$. When $n=1$, Eq.~\eqref{eq:entry13} is trivial, and $e(\eta),e(\eta')\equiv1\mod p$. Then, Eqs.~\eqref{eq:entry11},~\eqref{eq:entry12} give $l(\zeta')\equiv l(\zeta),m(\zeta')\equiv m(\zeta)\mod p^n$, i.e., $\cR_z(\zeta')\equiv \cR_z(\zeta)\mod p^n$, with $\zeta'\equiv\zeta\in p(\ZZ/p^n\ZZ)$. This implies that $l(\zeta)\equiv1\not\equiv0\mod p$ is invertible in $\ZZ/p^n\ZZ$. Eqs.~\eqref{eq:entry23},~\eqref{eq:entry33} give $a(\xi')\equiv a(\xi)$, $pc(\xi')\equiv pc(\xi)\mod p^n$. Now the remaining equations become as follows:
\begin{align}
    &pc(\xi)l(\zeta)\big(g(\eta)+g(\eta')\big)\equiv0\mod p^n;\\
    &pc(\xi)m(\zeta)\big(g(\eta)+g(\eta')\big)\equiv0\mod p^n;\\
    & m(\zeta)\big(c(\xi)-c(\xi')\big)\equiv -a(\xi)l(\zeta)\big(g(\eta)+g(\eta')\big)\mod p^n;\label{eq:ultima1}\\
    & l(\zeta)\big(c(\xi)-c(\xi')\big)\equiv -a(\xi)vm(\zeta)\big(g(\eta)+g(\eta')\big)\mod p^n.\label{eq:ultima2}
\end{align}
The first two equations are satisfied. The last two equations give $g(\eta)+g(\eta')\equiv -\frac{m(\zeta)}{a(\xi)l(\zeta)}\big(c(\xi)-c(\xi')\big)$ and $c(\xi)-c(\xi')\equiv -\frac{a(\xi)vm(\zeta)}{l(\zeta)}\big(g(\eta)+g(\eta')\big)\mod p^n$. Plugging the former into the latter, $\big(c(\xi)-c(\xi')\big)\big(l(\zeta)^2-vm(\zeta)^2\big)\equiv0\mod p^n$, where $l(\zeta)^2-vm(\zeta)^2\equiv\det\cR_z(\zeta)\equiv1\mod p^n$, therefore $c(\xi)\equiv c(\xi')\mod p^n$. Hence $g(\eta')\equiv-g(\eta)\mod p^n$ and $\eta'\equiv-\eta\mod p^n$. Summing up, we have found the unique solution 
\beq \label{eq:111=222}
\cR_x(\xi')\equiv \cR_x(\xi),\quad \cR_y(\eta')=\cR_y(-\eta),\quad \cR_z(\zeta')\equiv \cR_z(\zeta)\mod p^n.
\eeq 
Very similar calculations hold for $112\equiv221$, $121\equiv212$, and $211\equiv122$.

We move on with $111\equiv221$, i.e.,
\beq 
\cR_x(\xi)\cR_y(\eta)\cR_z(\zeta)\equiv \cR_x(\infty)\cR_x(\xi') \cR_y(\infty)\cR_y(\eta')\cR_z(\zeta')\mod p^n,
\eeq 
for some $\xi,\xi',\eta,\eta',\zeta,\zeta'\in\ZZ/p^n\ZZ$. This yields
\begin{align}
    & e(\eta)l(\zeta)\equiv -e(\eta')l(\zeta')\mod p^n;\label{eq:entry11bis}\\
    & e(\eta)vm(\zeta)\equiv-e(\eta')vm(\zeta')\mod p^n;\label{eq:entry12bis}\\
    & -pg(\eta)\equiv pg(\eta')\mod p^n;\label{eq:entry13bis}\\
    & a(\xi)m(\zeta) +\frac{p}{v}c(\xi)g(\eta)l(\zeta) \equiv -a(\xi')m(\zeta')+\frac{p}{v}c(\xi')g(\eta')l(\zeta') \mod p^n;\label{eq:entry21bis}\\
    & a(\xi)l(\zeta) +\frac{p}{v}c(\xi)g(\eta)vm(\zeta)\equiv -a(\xi')l(\zeta')+\frac{p}{v}c(\xi')g(\eta')vm(\zeta')\mod p^n;\label{eq:entry22bis}\\
    &\frac{p}{v}c(\xi)e(\eta)\equiv \frac{p}{v}c(\xi')e(\eta')\mod p^n;\label{eq:entry23bis}\\
    &c(\xi)m(\zeta)+ a(\xi)g(\eta)l(\zeta)\equiv -c(\xi')m(\zeta')+a(\xi')g(\eta')l(\zeta')\mod p^n;\label{eq:entry31bis}\\
    & c(\xi)l(\zeta)+ a(\xi)g(\eta)vm(\zeta)\equiv -c(\xi')l(\zeta')+ a(\xi')g(\eta')vm(\zeta')\mod p^n;\label{eq:entry32bis}\\
    &a(\xi)e(\eta)\equiv a(\xi')e(\eta')\mod p^n.\label{eq:entry33bis}
\end{align}
Again $e(\eta')\equiv e(\eta)\mod p^n$, by Eq.~\eqref{eq:entry13bis} for $n>1$, and just congruent to $1$ when $n=1$. So Eqs.~\eqref{eq:entry11bis},~\eqref{eq:entry12bis} give $l(\zeta')\equiv-l(\zeta),m(\zeta')\equiv-m(\zeta)\mod p^n$, i.e., $\cR_z(\zeta')\equiv\cR_z(\infty)\cR_z(\zeta)\mod p^n$. The latter is equivalent to $(\zeta+\zeta')(1-v\zeta\zeta')\equiv0\mod p^n$. The former is equivalent to $1\equiv (v\zeta\zeta')^2\mod p^n$, which is impossible if $\zeta\equiv0\mod p$ or $\zeta'\equiv0\mod p$. Hence, we assume $\zeta,\zeta'\not\equiv0\mod p$ while solving
\beq \label{eq:system2inter}
\left\{\begin{aligned}
    & (1+v\zeta\zeta')(1-v\zeta\zeta')\equiv0\mod p^n;\\
    & (\zeta+\zeta')(1-v\zeta\zeta')\equiv0\mod p^n.
\end{aligned}\right.
\eeq 
If $1-v\zeta\zeta'\equiv p^ju_j\mod p^n$ for some $u_j\not\equiv0\mod p$ and $0\leq j<n$, then system~\eqref{eq:system2inter} leads to $\zeta^2\equiv v^{-1}\mod p^{n-j}$, which is impossible since $v$ is not a square. Therefore, for the system to possibly have solutions, it must be $1-v\zeta\zeta'\equiv0\mod p^n$, providing $\zeta'\equiv\frac{1}{v\zeta}\mod p^n$. We have obtained $\cR_z(\zeta')\equiv \cR_z\left(\frac{1}{v\zeta}\right)\equiv \cR_z(\infty)\cR_z(\zeta)\mod p^n$; in other words, $\cR_z(\zeta')\equiv \cR_z(\infty)\cR_z(\zeta)\mod p^n$ is possible if (and only if) both $\cR_z(\zeta)$ and $\cR_z(\zeta')$ are in the first branch, under the transformation $\zeta\mapsto\frac{1}{v\zeta}$. We are left again with Eqs.~\eqref{eq:ultima1},~\eqref{eq:ultima2}. Since $\zeta\not\equiv0\mod p$, $m(\zeta)$ is invertible modulo $p^n$, so these equations provide $c(\xi)-c(\xi')\equiv -\frac{a(\xi)l(\zeta)}{m(\zeta)}\big(g(\eta)+g(\eta')\big)$ and $g(\eta)+g(\eta')\equiv -\frac{l(\zeta)}{a(\xi)vm(\zeta)}\big(c(\xi)-c(\xi')\big)\mod p^n$. As before, they give $c(\xi')\equiv c(\xi)$ and $g(\eta')\equiv-g(\eta)\mod p^n$, and globally we have found
\beq \label{eq:111=221}
\begin{aligned}
    \cR_x(\xi')\equiv \cR_x(\xi),\quad \cR_y(\eta')=\cR_y(-\eta),\quad \cR_z(\zeta')\equiv \cR_z(\infty)\cR_z(\zeta)\mod p^n, \\
    \zeta,\zeta'\not\equiv0\mod p.
\end{aligned}
\eeq 
Very similar calculations hold for $121\equiv211$.
\begin{remark}\label{rem:6cardequality}
So far, we have shown six different modular congruences of Cardano representations with respect to certain triples of branches, for every $n\in\NN$. Hence, for every $n\in\NN$, and every prime $p>2$, every $\pi_n(L)\in G_{p^n}$ has at least two distinct Cardano representations of the kind $\cR_x\cR_y\cR_z\mod p^n$, that is (depending on $\pi_n(L)$) one of the following six:
\beq\notag \hspace{-0.5cm}\begin{array}{lll}
    111\leftrightarrow 221 &\cR_x(\xi)\cR_y(\eta)\cR_z(\zeta)\equiv \cR_x(\infty)\cR_x(\xi)\cR_y(\infty)\cR_y(-\eta)\cR_z(\infty)\cR_z(\zeta)\mod p^n,  & \zeta\not\equiv0 \mod p,\\ 
    111\leftrightarrow222\quad  & \cR_x(\xi)\cR_y(\eta)\cR_z(\zeta)\equiv \cR_x(\infty)\cR_x(\xi)\cR_y(\infty)\cR_y(-\eta)\cR_z(\infty)\cR_z(\zeta)\mod p^n,\quad & \zeta\equiv0\mod p;\\
    112 \leftrightarrow 221 &\cR_x(\xi)\cR_y(\eta)\cR_z(\infty)\cR_z(\zeta)\equiv \cR_x(\infty)\cR_x(\xi)\cR_y(\infty)\cR_y(-\eta)\cR_z(\zeta)\mod p^n,  & \zeta\equiv0\mod p;\\
    121 \leftrightarrow 211 &\cR_x(\xi)\cR_y(\infty)\cR_y(\eta)\cR_z(\zeta)\equiv \cR_x(\infty)\cR_x(\xi)\cR_y(-\eta)\cR_z(\infty)\cR_z(\zeta)\mod p^n,  & \zeta\not\equiv0\mod p;\\
    121 \leftrightarrow 212 &\cR_x(\xi)\cR_y(\infty)\cR_y(\eta)\cR_z(\zeta)\equiv \cR_x(\infty)\cR_x(\xi)\cR_y(-\eta)\cR_z(\infty)\cR_z(\zeta)\mod p^n,  & \zeta\equiv0\mod p;\\
    122 \leftrightarrow211 & 
    \cR_x(\xi)\cR_y(\infty)\cR_y(\eta)\cR_z(\infty)\cR_z(\zeta) \equiv \cR_x(\infty)\cR_x(\xi)\cR_y(-\eta)\cR_z(\zeta)\mod p^n,  & \zeta\equiv0\mod p.
\end{array}\eeq
Interpreting this table, given a Cardano representation of $L\mod p^n$ realised by the triple 
$(\cR_x\mod p^n,\cR_y\mod p^n,\cR_z\mod p^n)\in G_{x,p^n}\times G_{y,p^n}\times G_{z,p^n}$ of parameters $\xi,\eta,\zeta$, respectively, then $L\mod p^n$ admits at least another distinct Cardano representation with parameters $\xi',\eta',\zeta'$, respectively: That obtained by 
\begin{itemize}
    \item changing the branch of the $x$-rotation, with parameter $\xi'\equiv\xi\mod p^n$, i.e., 
    \beq 
    \cR_x(\xi)\mapsto \cR_x(\infty)\cR_x(\xi)\mod p^n;
    \eeq 
     \item changing the branch of the $y$-rotation, with parameter $\eta'\equiv-\eta\mod p^n$, i.e., 
     \beq 
     \cR_y(\eta)\mapsto\cR_y(\infty)\cR_y(-\eta)\mod p^n;
     \eeq 
     \item changing the branch of the $z$-rotation if $\zeta\equiv0\mod p$, with parameter $\zeta'\equiv\zeta\mod p^n$; or fixing the branch of the $z$-rotation if $\zeta\not\equiv0\mod p$, with parameter $\zeta'\equiv\frac{1}{v\zeta}\mod p^n$, giving in any case 
     \beq 
     \cR_z(\zeta)\mapsto \cR_z(\infty)\cR_z(\zeta)\mod p^n.
     \eeq 
\end{itemize}
This is exactly what happens for $\so(3)_p$, as in Eq.~\eqref{eq:cardanorepDOPPIA}.
\end{remark}
Now we are going to show that those in Remark~\ref{rem:6cardequality} are the only congruences of Cardano compositions $\cR_x\cR_y\cR_z\mod p^n$ with two different triples $(\cR_x\mod p^n,\cR_y\mod p^n,\cR_z\mod p^n)$. We start by showing that $\pi_n(L)\in G_{p^n}$ cannot admit two distinct Cardano representations with respect to the same three branches. When $111\equiv111$, we get the following congruences: 
\begin{align}
    &e(\eta)l(\zeta)\equiv e(\eta')l(\zeta')\mod p^n;\label{eq:entry11tris}\\
    &e(\eta)vm(\zeta)\equiv e(\eta')vm(\zeta')\mod p^n;\label{eq:entry12tris}\\
    &-pg(\eta)\equiv -pg(\eta')\mod p^n;\label{eq:entry13tris}\\
    &a(\xi)m(\zeta) +\frac{p}{v}c(\xi)g(\eta)l(\zeta)\equiv a(\xi')m(\zeta')+\frac{p}{v}c(\xi')g(\eta')l(\zeta')\mod p^n;\label{eq:entry21tris}\\
    & a(\xi)l(\zeta) +\frac{p}{v}c(\xi)g(\eta)vm(\zeta)\equiv a(\xi')l(\zeta')+\frac{p}{v}c(\xi')g(\eta')vm(\zeta')\mod p^n;\label{eq:entry22tris}\\
    &\frac{p}{v}c(\xi)e(\eta)\equiv\frac{p}{v}c(\xi')e(\eta')\mod p^n;\label{eq:entry23tris}\\
    &c(\xi)m(\zeta)+ a(\xi)g(\eta)l(\zeta)\equiv c(\xi')m(\zeta')+a(\xi')g(\eta')l(\zeta')\mod p^n;\label{eq:entry31tris}\\
    & c(\xi)l(\zeta)+ a(\xi)g(\eta)vm(\zeta)\equiv c(\xi')l(\zeta')+a(\xi')g(\eta')vm(\zeta')\mod p^n;\label{eq:entry32tris}\\
    & a(\xi)e(\eta)\equiv a(\xi')e(\eta')\mod p^n.\label{eq:entry33tris}
\end{align}
Now Eq.~\eqref{eq:entry13tris} for $n>1$ gives $\eta'\equiv\eta\mod p^{n-1}$, and again $e(\eta')\equiv e(\eta)\mod p^n$, while this is trivial for $n=1$. Eqs.~\eqref{eq:entry11tris},~\eqref{eq:entry12tris} provide $\cR_z(\zeta')\equiv\cR_z(\zeta)\mod p^n$; Eqs.~\eqref{eq:entry23tris},~\eqref{eq:entry33tris} yield $a(\xi')\equiv a(\xi)$ and $pc(\xi')\equiv pc(\xi)\mod p^n$, trivial for $n=1$. Then, Eqs.~\eqref{eq:entry21tris},~\eqref{eq:entry22tris} are satisfied, while the remaining Eqs.~\eqref{eq:entry31tris},~\eqref{eq:entry32tris} become
\begin{align}
    & m(\zeta)\big(c(\xi)-c(\xi')\big)\equiv -a(\xi)l(\zeta)\big(g(\eta)-g(\eta')\big)\mod p^n;\label{eq:ultima1tris}\\
    & l(\zeta)\big(c(\xi)-c(\xi')\big)\equiv -a(\xi)vm(\zeta)\big(g(\eta)-g(\eta')\big)\mod p^n.\label{eq:ultima2tris}
\end{align}
If $\zeta\equiv0\mod p$ we proceed as in the end of case $111\equiv222$, while if $\zeta\not\equiv0\mod p$ as in $111\equiv 221$; anyway we get $c(\xi')\equiv c(\xi)$ and $g(\eta')\equiv g(\eta)\mod p^n$. In conclusion, $\cR_x(\xi)\cR_y(\eta)\cR_z(\zeta)\equiv \cR_x(\xi')\cR_y(\eta')\cR_z(\zeta')\mod p^n$ if and only if
\beq 
\cR_x(\xi')\equiv \cR(\xi),\quad \cR_y(\eta')\equiv \cR_y(\eta),\quad \cR_z(\zeta')\equiv \cR_z(\zeta)\mod p^n.
\eeq 
Very similar calculations hold for the congruences of Cardano compositions with respect to the same branches, $121\equiv 121$, $211\equiv 211$, $221\equiv 221$, $112\equiv 112$, $122\equiv 122$, $212\equiv212$ and $222\equiv222$. 

There are $36 - 6 - 8 = 22$ possibilities left of equating two triples $ijk$. In any case, as seen above, the modular congruence of the $(1,3)$-entries gives $e(\eta')\equiv e(\eta)\mod p^n$, congruent to $1\mod p$. For $111\equiv 112$, the modular congruence of the first matrix entries gives $l(\zeta')\equiv-l(\zeta)\mod p^n$, which is impossible [refer to Eq.~\eqref{eq:disjointsign}] since in the second branch $\zeta'\equiv0\mod p$. The same happens for $111\equiv212$, $112\equiv211$, $112\equiv222$, $121\equiv122$, $121\equiv222$, $122\equiv 212$, $211\equiv212$, $212\equiv222$, $221\equiv122$, $221\equiv222$. On the other hand, for $111\equiv 121$, the modular congruence of the last matrix entries gives $a(\xi')\equiv -a(\xi)\mod p^n$, which is again impossible. The same happens for $111\equiv122$, $111\equiv211$, $112\equiv121$, $112\equiv122$, $112\equiv212$, $121\equiv221$, $122\equiv222$, $211\equiv221$, $211\equiv222$, $212\equiv221$.
\end{proof}

Now, the order of the group $G_{p^n}$, $n\in\NN$, is straightforward to calculate as in Proposition~\ref{prop:ordinigrup3dmodpk}, by knowing the duplicity of the Cardano decomposition of each matrix in $G_{p^n}$ and the orders of the groups $G_{x,p^n},G_{y,p^n},G_{z,p^n}$.

\section{Lifting \`a la Hensel}\label{sec:appliftalaHens}
In remark~\ref{rem:henselift} we introduce the groups $\widetilde{G}_{\kappa,p^n},\widetilde{G}_{p^n}$ of solutions modulo $p^n$ of the defining conditions of $\so(d)_{p(,\kappa)}$, and we ask whether the inclusion~\eqref{eq:unainclusionesem} is an equality. 

Let us start answering this question in the bidimensional case. Let $L=\left(\ell_{ij}\right)_{i,j=1,2}\in \so(2)_{p,\kappa}$, and let $A_{\kappa}=\diag(a_1,a_2)\in\mathsf{M}_{2\times2}(\ZZ_p)$ be the matrix representation with respect to the canonical basis of the quadratic forms defining $\so(2)_{p,\kappa}$, as in Eq.~\eqref{cev1}. The defining conditions for $\so(2)_{p,\kappa}$ are explicitly
\beq \label{eq:sys2x2int}
\left\{\begin{aligned}
    &a_1\ell_{11}^2+a_2\ell_{21}^2=a_1,\\
    & a_1\ell_{12}^2+a_2\ell_{22}^2=a_2,\\
    &a_1\ell_{11}\ell_{12}+a_2\ell_{21}\ell_{22}=0,\\
    &\ell_{11}\ell_{22}-\ell_{12}\ell_{21}=1,
\end{aligned}\right.
\eeq 
where $\ell_{ij}\in\ZZ_p$ allows to project them modulo $p^n$, providing the defining conditions of $\widetilde{G}_{\kappa,p^n}$. When $n=1$, the solutions of system~\eqref{eq:sys2x2int} modulo $p$ give the following groups:
\begin{align}
    &\widetilde{G}_{-v,p}= 
    \left\{\begin{pmatrix}a&vb\\b&a\end{pmatrix}\bmod p\midd a,b\in\ZZ/p\ZZ,\,a^2-vb^2\equiv1\bmod p\right\};\label{eq:tilde-v}\\
    &\widetilde{G}_{p,p}=\widetilde{G}_{\frac{p}{u},p}=\left\{\begin{pmatrix}s&0\\c&s\end{pmatrix}\bmod p\midd c\in\ZZ/p\ZZ,\ s\in\{\pm1\}\right\}.\label{eq:tildeup}
\end{align}
The solutions forming $\widetilde{G}_{-v,p}$ are derived in Section~IV.A of~\cite{our2nd}, while in Appendix~A of~\cite{our2nd} it is proved that $\widetilde{G}_{-v,p}\simeq \{(a,b)\in\ZZ/p\ZZ\midd a^2-vb^2\equiv1\mod p\}\simeq \ZZ/(p+1)\ZZ$; the other groups are easily found. Comparing them with the groups $G_{\kappa,p}$, parametrised as in Eqs.~\eqref{eq:semplicep2k},~\eqref{eq:semplicepvk}, we see that
\beq 
\widetilde{G}_{p,p}=G_{p,p}=\widetilde{G}_{\frac{p}{u},p}=G_{\frac{p}{u},p},\qquad \quad\widetilde{G}_{-v,p}=G_{-v,p}.\label{eq:cardnet}
\eeq 
To understand if these equalities are kept over $\ZZ/p^n\ZZ$ for every $n\in\NN$, we need to understand if each solution modulo $p^n$ to the system~\eqref{eq:sys2x2int} lifts to a solution of the same system modulo $p^{n+1}$, until converging to a $p$-adic integer solution of the same system in $\ZZ_p$. The multivariable version of Hensel's lemma --- Theorem~3.3 of~\cite{conradMULvar} (see also~\cite{serre,fisher}) --- cannot answer our question starting from $n=1$, because the $4\times 4$ Jacobian matrix associated with the system~\eqref{eq:sys2x2int} has zero determinant. Therefore, we prove the following result by brute force.
\begin{proposition}\label{prop:HenselSO(2)K}
If $L=(\ell_{ij})_{i,j=1,2}\in\mathsf{M}_{2\times2}(\ZZ_p)$ is a solution to~\eqref{eq:sys2x2int} modulo $p^n$, there exists $Z=(z_{ij})_{i,j=1,2}\in\mathsf{M}_{2\times2}(\ZZ_p)$ such that $L+p^nZ$ is solution of the same system modulo $p^{n+1}$, for every $n\in\NN$. Any solution $L$ modulo $p^n$ admits exactly $p$ distinct lifted solutions $L+p^n Z$ modulo $p^{n+1}$.
\end{proposition}
\begin{proof}
The fact that $L=(\ell_{ij})_{i,j=1,2}\in\mathsf{M}_{2\times2}(\ZZ_p)$ is a solution to~\eqref{eq:sys2x2int} modulo $p^n$ means
\beq \label{eq:sys2x2modn}
\left\{\begin{aligned}
    &a_1\ell_{11}^2+a_2\ell_{21}^2\equiv a_1+\lambda_1p^n\mod p^{n+1},\\
    & a_1\ell_{12}^2+a_2\ell_{22}^2\equiv a_2+\lambda_2p^n\mod p^{n+1},\\
    &a_1\ell_{11}\ell_{12}+a_2\ell_{21}\ell_{22}\equiv \lambda_3p^n\mod p^{n+1},\\
    &\ell_{11}\ell_{22}-\ell_{12}\ell_{21}\equiv 1+\lambda_dp^n\mod p^{n+1},
\end{aligned}\right.
\eeq 
for some $\lambda_1,\lambda_2,\lambda_3,\lambda_d\in\ZZ_p$ determined by $L$. We plug $L+p^nZ$ in system~\eqref{eq:sys2x2int} modulo $p^{n+1}$:
\beq \label{eq:sys2x2modn+1}
\left\{\begin{aligned}
    &a_1(\ell_{11}+p^nz_{11})^2+a_2(\ell_{21}+p^nz_{21})^2\equiv a_1\mod p^{n+1},\\
    & a_1(\ell_{12}+p^nz_{12})^2+a_2(\ell_{22}+p^nz_{22})^2\equiv a_2\mod p^{n+1},\\
    &a_1(\ell_{11}+p^nz_{11})(\ell_{12}+p^nz_{12})+a_2(\ell_{21}+p^nz_{21})(\ell_{22}+p^nz_{22})\equiv 0\mod p^{n+1},\\
    &(\ell_{11}+p^nz_{11})(\ell_{22}+p^nz_{22})-(\ell_{12}+p^nz_{12})(\ell_{21}+p^nz_{21})\equiv 1\mod p^{n+1}.
\end{aligned}\right.
\eeq 
All products between two terms containing a factor $p^n$ cancel, since $p^{2n}\equiv0\mod p^{n+1}$ for $n\in\NN$, and plugging~\eqref{eq:sys2x2modn} into~\eqref{eq:sys2x2modn+1} we obtain
\beq\label{sys2x2modnz}
\left\{\begin{aligned}
    &\lambda_1+2a_1\ell_{11}z_{11}+2a_2\ell_{21}z_{21}\equiv 0\mod p,\\
    &\lambda_2+2a_1\ell_{12}z_{12}+2a_2\ell_{22}z_{22}\equiv 0\mod p,\\
    &\lambda_3+a_1\ell_{11}z_{12}+a_1\ell_{12}z_{11}+a_2\ell_{21}z_{22}+a_2\ell_{22}z_{21}\equiv 0\mod p,\\
    &\lambda_d+\ell_{11}z_{22}+\ell_{22}z_{11}-\ell_{12}z_{21}-\ell_{21}z_{12}\equiv 0\mod p.
\end{aligned}\right.
\eeq 
We now look for solutions $Z$ at given $L$. We start from $\kappa=-v$, where [cf.\ Eqs.~\eqref{eq:tilde-v},~\eqref{cev1}] we have $\ell_{11}\equiv \ell_{22}\equiv a,\ell_{21}\equiv b, \ell_{12}\equiv vb, a^2-vb^2\equiv1, a_1\equiv1,a_2\equiv-v\mod p$, giving
\beq
\left\{\begin{aligned}
    &az_{11}-vbz_{21}\equiv -\frac{\lambda_1}{2}\mod p,\\
    &az_{22}-bz_{12}\equiv \frac{\lambda_2}{2v}\mod p,\\
    &vbz_{11}+az_{12}-vaz_{21}-vbz_{22}\equiv -\lambda_3\mod p,\\
    &az_{11}-bz_{12}-vbz_{21}+az_{22}\equiv -\lambda_d\mod p.
\end{aligned}\right.
\eeq 
If $b\equiv 0\mod p$ then $a^2\equiv1\mod p$ and
\beq\label{pv1lift}
\left\{\begin{aligned}
    &z_{11}\equiv -\frac{a}{2}\lambda_1\mod p,\\
    &z_{22}\equiv \frac{a}{2v}\lambda_2\mod p,\\
    &z_{12}\equiv vz_{21}-a\lambda_3\mod p,
\end{aligned}\right.
\eeq 
while if $b\not\equiv0\mod p$ then
\beq\label{pv2lift}
\left\{\begin{aligned}
    &z_{21}\equiv \frac{a}{vb}z_{11}+\frac{\lambda_1}{2vb}\mod p,\\
    &z_{12}\equiv \frac{a}{b}z_{11}+\frac{a^2}{2b}\lambda_1+\frac{a^2-1}{2vb}\lambda_2-a\lambda_3\mod p,\\
    &z_{22}\equiv z_{11}+\frac{a}{2}\lambda_1+\frac{a}{2v}\lambda_2-b\lambda_3\mod p,
\end{aligned}\right.
\eeq 
where, in both Eqs.~\eqref{pv1lift} and~\eqref{pv2lift}, the $\lambda_i$s are fixed by $L$ as in~\eqref{eq:sys2x2modn}, with the additional condition 
\beq
v\lambda_1-\lambda_2-2v\lambda_d\equiv0\mod p.
\eeq

Now we go back to system~\eqref{sys2x2modnz} for $\kappa\in\{p,\frac{p}{u}\}$, where [cf.\ Eqs.~\eqref{eq:tildeup},~\eqref{cev1}] we have $\ell_{11}\equiv\ell_{22}\equiv s,\ \ell_{12}\equiv0,\ell_{21}\equiv c,a_1\in\{1,u\},\,a_2\equiv0\mod p$, so
\beq\label{puplift}
\left\{\begin{aligned}
    &z_{11}\equiv-s\frac{\lambda_1}{2a_1}\mod p\\
    &z_{12}\equiv -s\frac{\lambda_3}{a_1}\mod p\\
    &z_{22}\equiv s\frac{\lambda_1}{2a_1}-\frac{c}{a_1}\lambda_3-s\lambda_d\mod p
\end{aligned}\right.
\eeq 
where the $\lambda_i$s are fixed by $L$ as in~\eqref{eq:sys2x2modn} with the additional  condition 
\beq
\lambda_2\equiv 0\mod p.
\eeq 
Fixed $L$ solution modulo $p^n$, there are at most $p$ distinct liftings to solutions $L+p^nZ$ modulo $p^{n+1}$, because Eqs.~\eqref{pv1lift},~\eqref{pv2lift},~\eqref{puplift} depend on one free parameter, either $z_{21}\mod p$ or $z_{11}\mod p$. Furthermore $L+p^nZ\equiv L+p^nZ'\mod p^{n+1}$ if and only if $z_{ij}\equiv z_{ij}'\mod p$ for every $i,j=1,2$: Different $Z$s modulo $p$ provide different liftings modulo $p^{n+1}$ of a same solution $L$ modulo $p^n$. In other words, if $L$ lifts to a solution modulo $p^{n+1}$, then it has exactly $p$ distinct liftings $L+p^nZ$ which are solutions modulo $p^{n+1}$.

We need to see whether or not Eq.~\eqref{sys2x2modnz} admits solutions $Z$ at given $L$. Since the $\lambda_i$s are determined by $L$, the answer only depends on the conditions
\begin{align}
v\lambda_1-\lambda_2-2v\lambda_d\equiv0\mod p,\qquad &\textup{for } \kappa=-v,\label{eq:condki1}\\
\lambda_2\equiv 0\mod p,\qquad &\textup{for } \kappa\in\left\{p,\frac{p}{u}\right\}. \label{eq:condki2}
\end{align}
A solution $L$ modulo $p^n$ does not lift to a solution modulo $p^{n+1}$ if Eq.~\eqref{eq:condki1} or~\eqref{eq:condki2} provides a non-trivial constraint on the $\lambda_i$s located by $L$. But now we show that this is never the case, by induction. For $n=1$, if, by contradiction, there existed $L$ solution modulo $p$ which does not lift to $L+pZ$ solution modulo $p^2$, then it would be $|\widetilde{G}_{\kappa,p^2}|< p|\widetilde{G}_{\kappa,p}|=p|G_{\kappa,p}|=|G_{\kappa,p^2}|$, where the first equality is by Eq.~\eqref{eq:cardnet} and the second one by Eq.~\eqref{eq:liftings2d}; this is in contradiction with $G_{\kappa,p^2} \subseteq \widetilde{G}_{\kappa,p^2}$. If we assume that $|G_{\kappa,p^n}|=|\widetilde{G}_{\kappa,p^n}|$ and that some $L\in \widetilde{G}_{\kappa,p^n}$ does not lift to $\widetilde{G}_{\kappa,p^{n+1}}$, then we would have $|\widetilde{G}_{\kappa,p^{n+1}}| < p|\widetilde{G}_{\kappa,p^n}|=p|G_{\kappa,p^n}|= |G_{\kappa,p^{n+1}}|$, which is again a contradiction.
\end{proof}

\bigskip

We now move to the three-dimensional case. Let $L=\left(\ell_{ij}\right)_{i,j=1}^3\in \so(3)_p$, and let $A_+=\diag(a_1,a_2,a_3)=\diag(1,-v,p)\in\mathsf{M}_{3\times3}(\ZZ_p)$ be the matrix representation with respect to the canonical basis of the quadratic form defining $\so(3)_p$ [cf.\ Eq.~\eqref{cev2}]. The defining conditions for $\so(3)_p$ are explicitly
\beq\label{eq:syscond}
\left\{\begin{aligned}
&\sum_{i=1,2,3}a_i\ell_{ij}^2=a_j,\quad \quad \quad j=1,2,3,\\
&\sum_{i=1,2,3}a_i\ell_{ij}\ell_{im}=0,\quad \ \ \  (j,m)\in\{(1,2),\,(1,3),\,(2,3)\},\\
&{\ell_{11}}{\ell_{22}}{\ell_{33}}
+{\ell_{12}}{\ell_{23}}{\ell_{31}}+{\ell_{13}}{\ell_{21}}{\ell_{32}}+\\
&-{\ell_{11}}{\ell_{23}}{\ell_{32}}-{\ell_{12}}{\ell_{21}}{\ell_{33}}-{\ell_{13}}{\ell_{22}}{\ell_{31}}=1,
\end{aligned}\right.
\eeq
where $\ell_{ij}\in\ZZ_p$ allows to project them modulo $p^n$, providing the defining conditions of $\widetilde{G}_{p^n}$. Again, to understand whether each element in $\widetilde{G}_{p^n}$ lift to elements in $\widetilde{G}_{p^{n+1}}$, the multivariable version of Hensel's lemma --- Theorem~3.8 of~\cite{conradMULvar} cannot be exploited, because the $7\times 9$ Jacobian matrix associated with the system~\eqref{eq:syscond} has at most rank $6$, implying that any $7\times 7$ submatrix has zero determinant. In Remark IV.7 of~\cite{our2nd} one deduces that $\widetilde{G}_p=G_p$, of order $2p^2(p+1)$, as an element in $\widetilde{G}_p$ (i.e., a matrix solution of the system modulo $p$) is of the form
\beq \label{eq:paramsoluzsismodp}
\widetilde{L}=\begin{pmatrix}
    a&svb&0\\ b&sa&0\\c&d&s
\end{pmatrix}
\eeq 
for some $a,b,c,d\in\ZZ/p\ZZ$ such that $a^2-vb^2\equiv1\mod p$ and $s\in\{\pm1\}$. Here we give the general answer, in a very similar fashion to the bidimensional case above.

\begin{proposition}\label{prop:HenselSO(3)}
If $L=\left(\ell_{ij}\right)_{i,j=1}^3\in\mathsf{M}_{3\times3}(\ZZ_p)$ is a solution to~\eqref{eq:syscond} modulo $p^n$, there exists $Z=\left(z_{ij}\right)_{i,j=1}^3\in\in\mathsf{M}_{3\times3}(\ZZ_p)$ such that $L+p^nZ$ is solution of the same system modulo $p^{n+1}$, for every $n\in\NN$.\\
Any solution $L$ modulo $p^n$ admits exactly $p^3$ distinct lifted solutions $L+p^nZ$ modulo $p^{n+1}$.
\end{proposition}
\begin{proof}
The fact that $L$ is a solution to~\eqref{eq:syscond} modulo $p^n$ means
\beq\label{eq:kfixedell}
\left\{\begin{aligned}
&\sum_{i=1,2,3}a_i\ell_{ij}^2\equiv a_j+\lambda_jp^n\mod p^{n+1},\quad \text{for}\ j=1,2,3;\\
&\sum_{i=1,2,3}a_i\ell_{ij}\ell_{im}\equiv \lambda_{j+m+1}p^n\mod p^{n+1},\\ &\quad \quad \quad \text{for}\ (j,m)\in\{(1,2),\,(1,3),\,(2,3)\};\\
&{\ell_{11}}{\ell_{22}}{\ell_{33}}
+{\ell_{12}}{\ell_{23}}{\ell_{31}}+{\ell_{13}}{\ell_{21}}{\ell_{32}}+\\
&-{\ell_{11}}{\ell_{23}}{\ell_{32}}-{\ell_{12}}{\ell_{21}}{\ell_{33}}-{\ell_{13}}{\ell_{22}}{\ell_{31}}\equiv 1+\lambda_dp^n\mod p^{n+1};
\end{aligned}\right.
\eeq
for some $\lambda_1,\lambda_2,\lambda_3,\lambda_4,\lambda_5,\lambda_6,\lambda_d\in\ZZ_p$ determined by $L$.\\
We plug $L+p^nZ$ in the system of equations modulo $p^{n+1}$, and look for solutions $Z$ at given $L$:
\beq\label{eq:sislift}
\left\{\begin{aligned}
&\sum_{i=1,2,3}a_i(\ell_{ij}+p^nz_{ij})^2\equiv a_j\ \text{mod}\ p^{n+1},\quad \text{for}\ j=1,2,3;\\
&\sum_{i=1,2,3}a_i(\ell_{ij}+p^nz_{ij})(\ell_{im}+p^nz_{im})\equiv 0\ \text{mod}\ p^{n+1},\\
&\quad \quad \quad \text{for}\ (j,m)\in\{(1,2),\,(1,3),\,(2,3)\};\\
&(\ell_{11}+p^nz_{11})(\ell_{22}+p^nz_{22})(\ell_{33}+p^nz_{33})
+\\&+(\ell_{12}+p^nz_{12})(\ell_{23}+p^nz_{23})(\ell_{31}+p^nz_{31})+\\
&+(\ell_{13}+p^nz_{13})(\ell_{21}+p^nz_{21})(\ell_{32}+p^nz_{32})+\\&-(\ell_{11}+p^nz_{11})(\ell_{23}+p^nz_{23})(\ell_{32}+p^nz_{32})+\\
&-(\ell_{12}+p^nz_{12})(\ell_{21}+p^nz_{21})(\ell_{33}+p^nz_{33})+\\&-(\ell_{13}+p^nz_{13})(\ell_{22}+p^nz_{22})(\ell_{31}+p^nz_{31})\equiv 1\mod p^{n+1}.
\end{aligned}\right.
\eeq 
All products between two terms containing a factor $p^n$ cancel, and plugging~\eqref{eq:kfixedell} into~\eqref{eq:sislift} we get
\beq\label{eq:daquipe2diff}
\left\{\begin{aligned}
&a_j+p^n\big(\lambda_j+2\sum_{i=1,2,3}a_i\ell_{ij}z_{ij}\big)\equiv a_j\ \text{mod}\ p^{n+1},\quad \text{for}\ j=1,2,3;\\
&p^n\big[\lambda_{j+m+1}+\sum_{i=1,2,3}a_i(\ell_{ij}z_{im}+\ell_{im}z_{ij})\big]\equiv 0\ \text{mod}\ p^{n+1},\\
&\quad \quad \quad \text{for}\ (j,m)\in\{(1,2),\,(1,3),\,(2,3)\};\\
&1+p^n\big[\lambda_d+\ell_{11}\ell_{22}z_{33}+\ell_{11}\ell_{33}z_{22}+\ell_{22}\ell_{33}z_{11}+\ell_{12}\ell_{23}z_{31}+\\
&+\ell_{12}\ell_{31}z_{23}+\ell_{23}\ell_{31}z_{12}+\ell_{13}\ell_{21}z_{32}+\ell_{13}\ell_{32}z_{21}+\ell_{21}\ell_{32}z_{13}+\\
&-\ell_{11}\ell_{23}z_{32}-\ell_{11}\ell_{32}z_{23}-\ell_{23}\ell_{32}z_{11}-\ell_{12}\ell_{21}z_{33}-\ell_{12}\ell_{33}z_{21}+\\
&-\ell_{21}\ell_{33}z_{12}-\ell_{13}\ell_{22}z_{31}-\ell_{13}\ell_{31}z_{22}-\ell_{22}\ell_{31}z_{13}\big]\equiv 1\mod p^{n+1}.
\end{aligned}\right.
\eeq 
As $a_3=p$, the double products involving $a_3$ in the first equations vanish, as well as all the terms with $a_3$ in the second equations. Then,~\eqref{eq:daquipe2diff} is equivalent to
\beq \label{eq:sysp3}
\left\{\begin{aligned}
&az_{11}-vbz_{21}\equiv -\frac{\lambda_1}{2}\mod p;\\
&bz_{12}-az_{22}\equiv -s\frac{\lambda_2}{2v}\mod p;\\
&\lambda_3\equiv0\mod p;\\
&az_{12}+svbz_{11}-vbz_{22}-svaz_{21}\equiv-\lambda_4\mod p;\\
&az_{13}-vbz_{23}\equiv-\lambda_5\mod p;\\
&bz_{13}-az_{23}\equiv-s\frac{\lambda_6}{v}\mod p;\\
&z_{33}\equiv -saz_{11}+bz_{12}-s(bd-sac)z_{13}+\\
&+svbz_{21}-az_{22}+s(ad-svbc)z_{23}-s\lambda_d\mod p;
\end{aligned}\right.
\eeq
where we exploited Eq.~\eqref{eq:paramsoluzsismodp}. If $b\equiv0\mod p$, then $a^2\equiv1\mod p$ and
\beq \label{eq:solift0}
\left\{\begin{aligned}
&z_{11}\equiv -\frac{a}{2}\lambda_1\mod p;\\
&z_{22}\equiv \frac{sa}{2v}\lambda_2\mod p;\\
&z_{12}\equiv svz_{21}-a\lambda_4\mod p;\\
&z_{13}\equiv-a\lambda_5\mod p;\\
&z_{23}\equiv \frac{sa}{v}\lambda_6\mod p;\\
&z_{33}\equiv -saz_{11}+acz_{13}-az_{22}+sadz_{23}-s\lambda_d\\
&\quad\equiv \frac{s}{2}\lambda_1-\frac{s}{2v}\lambda_2-c\lambda_5+\frac{d}{v}\lambda_6-s\lambda_d\mod p;
\end{aligned}\right.\eeq
where the $\lambda_i$s are given by $L$ as in~\eqref{eq:kfixedell}, together with the condition
\beq 
\lambda_3\equiv0\mod p.
\eeq 
If $b\not\equiv0\mod p$, then~\eqref{eq:sysp3} rewrites as follows:
\begin{align}
&z_{21}\equiv \frac{a}{vb}z_{11}+\frac{\lambda_1}{2vb}\mod p;\label{eq:z21modp}\\
&z_{12}\equiv \frac{a}{b}z_{22}-s\frac{\lambda_2}{2vb}\mod p;\label{eq:z12modp}\\
&\lambda_3\equiv0\mod p;\label{eq:k3modp}\\
&az_{12}-vbz_{22}\equiv sv(az_{21}-bz_{11})-\lambda_4\mod p;\label{eq:z22modp}\\
&z_{23}\equiv \frac{az_{13}+\lambda_5}{vb}\mod p;\label{eq:z23modp}\\
&z_{13}\equiv \frac{a}{b}z_{23}-s\frac{\lambda_6}{vb}\mod p;\label{eq:z13modp}\\
&z_{33}\equiv -saz_{11}+bz_{12}-s(bd-sac)z_{13}+\nonumber\\
&+svbz_{21}-az_{22}+s(ad-svbc)z_{23}-s\lambda_d\mod p.\label{eq:z33DET}
\end{align}
Plugging~\eqref{eq:z23modp} into~\eqref{eq:z13modp} we get
\beq 
z_{13}\equiv \frac{a^2z_{13}+a\lambda_5-sb\lambda_6}{vb^2}\mod p,
\eeq 
that is
\beq 
z_{13}\equiv sb\lambda_6-a\lambda_5\mod p.
\eeq 
From~\eqref{eq:z23modp} we get
\beq 
z_{23}\equiv\frac{sa}{v}\lambda_6-b\lambda_5\mod p.
\eeq 
Plugging Eqs.~\eqref{eq:z21modp} and~\eqref{eq:z12modp} in~\eqref{eq:z22modp}, the following are equivalent:
\begin{align}
    &a\left(\frac{a}{b}z_{22}-s\frac{\lambda_2}{2vb}\right)-vbz_{22}\equiv sva\left(\frac{a}{vb}z_{11}+\frac{\lambda_1}{2vb}\right)-svbz_{11}-\lambda_4\mod p;\\
&\frac{a^2-vb^2}{b}z_{22}\equiv s\frac{a^2-vb^2}{b}z_{11}+\frac{sa}{2b}\lambda_1+\frac{sa}{2vb}\lambda_2-\lambda_4\mod p;\\
&z_{22}\equiv sz_{11}+\frac{sa}{2}\left(\lambda_1+\frac{\lambda_2}{v}\right)-b\lambda_4\mod p.
\end{align}
Now $z_{12}$ in terms of $z_{11}$ is
\begin{align}
z_{12}&\equiv \frac{a}{b}\left[sz_{11}+\frac{sa}{2}\left(\lambda_1+\frac{\lambda_2}{v}\right)-b\lambda_4\right]-s\frac{\lambda_2}{2vb}\mod p\notag\\
&\equiv \frac{sa}{b}z_{11}+\frac{sa^2}{2b}\lambda_1+\frac{sb}{2}\lambda_2-a\lambda_4\mod p.
\end{align}
Lastly, we derive an expression for $z_{33}$ from~\eqref{eq:z33DET}:
\begin{align}
z_{33}&\equiv -saz_{11}+b\left[\frac{sa}{b}z_{11}+\frac{sa^2}{2b}\lambda_1+\frac{sb}{2}\lambda_2-a\lambda_4\right]+\notag\\
&+(ac-sbd)(sb\lambda_6-a\lambda_5)+svb\left(\frac{a}{vb}z_{11}+\frac{\lambda_1}{2vb}\right)+\notag\\
&-a\left[sz_{11}+\frac{sa}{2}\left(\lambda_1+\frac{\lambda_2}{v}\right)-b\lambda_4\right]+(sad-vbc)\left(\frac{sa}{v}\lambda_6-b\lambda_5\right)-s\lambda_d\notag\\
&\equiv \frac{s}{2}\lambda_1 +\frac{s}{2}\left(b^2-\frac{a^2}{v}\right)\lambda_2 -c(a^2-vb^2)\lambda_5
+d\left(\frac{a^2}{v}-b^2\right)\lambda_6-s\lambda_d\notag\\
&\equiv \frac{s}{2}\lambda_1-\frac{s}{2v}\lambda_2-c\lambda_5+\frac{d}{v}\lambda_6-s\lambda_d\mod p.
\end{align}
We collect the results obtained for $Z\in\ZZ_p$ when $b\not\equiv0\mod p$:
\beq \label{eq:solift0n}
\left\{\begin{aligned}
&z_{12}\equiv \frac{sa}{b}z_{11}+\frac{sa^2}{2b}\lambda_1+\frac{sb}{2}\lambda_2-a\lambda_4\mod p;\\
& z_{13}\equiv sb\lambda_6-a\lambda_5\mod p;\\
&z_{21}\equiv \frac{a}{vb}z_{11}+\frac{\lambda_1}{2vb};\\
&z_{22}\equiv sz_{11}+\frac{sa}{2}\left(\lambda_1+\frac{\lambda_2}{v}\right)-b\lambda_4\mod p;\\
&z_{23}\equiv \frac{sa}{v}\lambda_6-b\lambda_5\mod p;\\
&z_{33}\equiv \frac{s}{2}\lambda_1-\frac{s}{2v}\lambda_2-c\lambda_5+\frac{d}{v}\lambda_6-s\lambda_d\mod p;
\end{aligned}\right.
\eeq 
where the $\lambda_i$s are given by $\boldsymbol{\ell}$ as in~\eqref{eq:kfixedell}, together with the condition
\beq 
\lambda_3\equiv0\mod p.
\eeq 
As argued in the above proof for the bidimensional case, if some $L$ solution modulo $p^n$ lifts to a solution modulo $p^{n+1}$, then actually it has exactly $p^3$ distinct liftings $L+p^nZ$ which are solutions modulo $p^{n+1}$, because~\eqref{eq:solift0} depends on the free parameters $z_{21},z_{31},z_{32}\mod p$ and~\eqref{eq:solift0n} on $z_{11},z_{31},z_{32}\mod p$. However, whether $L$ fits or not depends on the condition
\beq \label{eq:condlift>2}
    \lambda_3\equiv0\mod p,
\eeq
as the $\lambda_i$s are determined by $L\in\mathsf{M}_{3\times3}(\ZZ_p)$.
This imposes the constraint on the solution $L$ of the system~\eqref{eq:syscond} modulo $p^n$, to satisfy also an equation of the same system modulo $p^{n+1}$, as expected from the lifting \`a la Hensel of a multiple root. Indeed, the above condition in the respective equation of~\eqref{eq:kfixedell} provides
\beq \label{eq:condlift>2b}
\ell_{13}^2-v\ell_{23}^2+p\ell_{33}^2\equiv p\mod p^{n+1}.
\eeq 
One repeats the same argument at the end of the proof in the bidimensional case (by induction, and locally by contradiction) to show that Eq.~\eqref{eq:condlift>2} --- or equivalently~\eqref{eq:condlift>2b} --- must be satisfied by every solution modulo $p^n$, or in other words, that every solution $L$ modulo $p^n$ lifts to solutions modulo $p^{n+1}$.
\end{proof}

\begin{corollary}\label{cor:HenselSO(3)}
The group $G_{(\kappa,)p^n}=\so(d)_{p(,\kappa)}\mod p^n$ coincides with the group $\widetilde{G}_{(\kappa,)p^n}$ as in~\eqref{eq:altrogrequiv2x2},~\eqref{eq:altrogrequiv}, for every $n\in\NN$:
\beq 
G_{\kappa,p^n}= \widetilde{G}_{\kappa,p^n}.
\eeq 
\end{corollary}
\begin{proof}
First, $G_{(\kappa,)p}=\widetilde{G}_{(\kappa,)p}$ by Eqs.~\eqref{eq:cardnet} and Remark~IV.7 of~\cite{our2nd}. In general, we have the inclusion $G_{(\kappa,)p^n}\subseteq \widetilde{G}_{(\kappa,)p^n}$ from Eq.~\eqref{eq:unainclusionesem}. On the other hand, Propositions~\ref{prop:HenselSO(2)K} and~\ref{prop:HenselSO(3)} state that each element in $\widetilde{G}_{(\kappa,)p^n}$ lifts to elements in $\widetilde{G}_{(\kappa,)p^{n+1}}$, and so on until converging to elements in $\so(d)_{p(,\kappa)}$ for $n\rightarrow \infty$. These latter can be projected via $\pi_n$, getting elements in $G_{p^n}$. In this way, also $\widetilde{G}_{p^n}\subseteq G_{p^n}$ is proved. 

An equivalent proof is as follows: Since $G_{(\kappa,)p^n}\subseteq \widetilde{G}_{(\kappa,)p^n}$, one has $G_{\kappa,p^n}= \widetilde{G}_{\kappa,p^n}$ if and only if $\left\lvert G_{\kappa,p^n}\right\rvert =\left\lvert \widetilde{G}_{\kappa,p^n}\right\rvert$. By Proposition~\ref{prop:HenselSO(2)K} one has $\lvert \widetilde{G}_{\kappa,p^n}\rvert = p\lvert \widetilde{G}_{\kappa,p^{n-1}}\rvert = p^{n-1}\lvert \widetilde{G}_p\rvert= p^{n-1}\lvert G_p\rvert =\left\{\begin{aligned}
&2p^n,\ \kappa\in\left\{p,\frac{p}{u}\right\},\\
& p^{n-1}(p+1),\ \kappa=-v, \end{aligned}\right.$ which coincides with $\lvert G_{\kappa,p^n}\rvert$ in Proposition~\ref{prop:cardparam}; and by Proposition~\ref{prop:HenselSO(3)} $\lvert \widetilde{G}_{p^n}\rvert = p^3\lvert \widetilde{G}_{p^{n-1}}\rvert = (p^3)^{n-1}\lvert \widetilde{G}_p\rvert= p^{3(n-1)}\lvert G_p\rvert =p^{3(n-1)}2p^2(p+1)=2p^{3n-1}(p+1)$ which is equal to $\lvert G_{p^n}\rvert$ in Proposition~\ref{prop:ordinigrup3dmodpk}. Indeed, the number of liftings of any $\widetilde{L}\in\widetilde{G}_{\kappa,p^n}$ to $\widetilde{G}_{\kappa,p^{n+1}}$ is equal to the cardinality of the preimage of any $\pi_n(L)\in G_{\kappa,p^n}$ with respect to $\varphi_{n,n+1}$.
\end{proof}

\section{Comparison of Haar measures on \texorpdfstring{$\so(2)_{p,\kappa}$}{Lg}}\label{sec:equivmisure2}
So far, two different approaches have been developed to find the Haar measure on the compact $p$-adic rotation groups in dimensions two and three. On the one side, this paper provides an inverse limit characterisation; on the other hand, an integral Haar measure on $p$-adic Lie groups was derived in~\cite{our3rd}, and applied to $p$-adic rotation groups. In particular, explicit calculations can be carried out for the integral Haar measure on $\so(2)_{p,\kappa}$. In this appendix, we want to make a comparison between these two formulations of Haar measure on $SO(2)_{p,\kappa}$, for every prime $p>2$ and $\kappa\in\{-v,p,up\}$.

Like on every compact group, the Haar measure on $\so(2)_{p,\kappa}$ is essentially unique. Thus, for every prime $p>2$, we want to explicitly show that $\overline{\mu}_{p,\kappa}$ on $\so(2)_{p,\kappa}$ coincides --- up to a positive multiplicative constant, due to normalisation --- with the Haar measure $\mu_2^{(\kappa)}$ given in~\cite{our3rd}. The latter is
\beq\label{haarmes2}
\mu_2^{(\kappa)}(E)=\int_{\varphi_{(\kappa)}(E)}\frac{1}{|1+\alpha_\kappa\s ^2|_{p}}\de \s ,
\eeq
for every Borel set $E\in\mathcal{B}(\so(2)_{p,\kappa})$, where $\de\s $ denotes the Haar measure on $\QQ_p$, while $\varphi_{(\kappa)}\colon\so(2)_{p,\kappa}\setminus\{-\mathrm{I}\}\rightarrow\QQ_p$ is the \emph{coordinate map} on $\so(2)_{p,\kappa}$ defined in such a way that  $\varphi^{-1}_{(\kappa)}(\s )= \mathcal{R}_\kappa(\s )$  (cf.\ Eq.~\eqref{eq:invertparam}). To this end, it is enough to show that the measure of any open ball in a topology base for $\so(2)_{p,\kappa}$ provides the same result in both the two approaches. Indeed, the topology base generates the topology of $\so(2)_{p,\kappa}$, which in turn provides its Borel $\sigma$-algebra. 

First, we want to normalise the integral measure $\mu_2^{(\kappa)}$ to evaluate to one on the whole group $\so(2)_{p,\kappa}$, likewise $\overline{\mu}_{p,\kappa}$. We just need to redefine the Haar measure in Eq.~\eqref{haarmes2} by dividing the second member by $\mu_2^{(\kappa)}\big(\so(2)_{p,\kappa}\big)$. To this end, we present a technical result, whose proof is pedagogical for the resolution of simple integrals over $\QQ_p$.
\begin{lemma}\label{prop.5.1}
For every $k\in\ZZ_{\leq0}$, 
    \beq \label{eq:piuingen}
    \int_{D_k(0)}\frac{1}{|1+\alpha_\kappa \s ^2|_p}\de \s =p^k.
    \eeq 
\end{lemma}
\begin{proof}
Any integral on $D_k(0)\subset\QQ_p$ can be decomposed as a sum of integrals on the disjoint concentric circles of radii $\leq p^k$ centred at $0$ in $\QQ_p$, which cover the whole $D_k(0)$~\cite{VVZ}. Indeed, if $S_m(0)\coloneqq\{x\in\QQ_p\midd |x|_p=p^m \}$, $m\in\ZZ$, then
\beq\label{eq.69}
\int_{D_k(0)}\frac{1}{|1+\alpha_\kappa \s ^2|_p}\de \s = \sum_{m\leq k}\int_{S_m(0)}\frac{1}{|1+\alpha_\kappa \s ^2|_p}\de \s .
\eeq
Since $k\in\ZZ_{\leq0}$ implies $D_k(0)\subseteq\ZZ_p$, then $\s \in\ZZ_p$ in the last integrals, case in which Remark~\ref{rem:paraminter} tells us that $|1+\alpha_\kappa\s ^2|_p=1$. Therefore, 
\beq
\int_{D_k(0)}\frac{1}{|1+\alpha_\kappa \s ^2|_p}\de \s = \sum_{m\leq k}\int_{S_m(0)}\de \s =\sum_{m\leq k} p^m\left(1-\frac{1}{p}\right),
\eeq
according to Example~2 at p.~40 of~\cite{VVZ}. By the change of index $N\coloneqq k-m$, we get 
\beq 
\sum_{m\leq k} p^m = \sum_{N\geq 0}p^{k-N} =  p^k\sum_{N\geq0}\left(\frac{1}{p}\right)^N.
\eeq 
The last sum is the geometric series of common ratio $\frac{1}{p}<1$, which converges to $\frac{1}{1-\frac{1}{p}}$, therefore
\beq
\int_{D_k(0)}\frac{1}{|1+\alpha_\kappa \s ^2|_p}\de \s =\left(1-\frac{1}{p}\right)p^k\frac{1}{1-\frac{1}{p}} = p^k.
\eeq
\end{proof}
We now give the normalised Haar measure on $\so(2)_{p,\kappa}$ in the integral approach.
\begin{theorem}
For every prime $p>2$, the Haar measure in~\cite{our3rd} normalised to one on $\so(2)_{p,\kappa}$ is given by
\beq \label{normhaarmes2}
\widetilde{\mu}_2^{(\kappa)}(E)=\frac{1}{\mu_2^{(\kappa)}\big(\so(2)_{p,\kappa}\big)}\int_{\varphi_{(\kappa)}(E)}\frac{1}{|1+\alpha_\kappa\s ^2|_p}\de\s ,
\eeq 
where 
\beq \label{eq:normlafacts}
\mu_2^{(\kappa)}\big(\so(2)_{p,\kappa}\big) = \begin{cases}
    1+\frac{1}{p},\ &\textup{ if } \ \kappa=-v,\\
    2, & \textup{ if }  \kappa\in\left\{p,up\right\}.
\end{cases}
\eeq 
\end{theorem}
\begin{proof}
One just needs to compute the integral 
\beq
\mu_2^{(\kappa)}\big(\so(2)_{p,\kappa}\big)=\int_{\QQ_p}\frac{1}{|1+\alpha_\kappa\s ^2|_p}\de\s.
\eeq
When $p\mid \alpha_\kappa$, according to Remark~\ref{rem:paraminter}, we write
\begin{align}
    \mu_2^{(\kappa)}\big(\so(2)_{p,\kappa}\big)&=\int_{\{\s \in\ZZ_p\}}\frac{1}{|1+\alpha_\kappa\s ^2|_p}\de\s + \int_{\left\{\s =-\frac{1}{\alpha_\kappa \tau}\midd \tau\in\ZZ_p\setminus\{0\}\right\}}\frac{1}{|1+\alpha_\kappa\s ^2|_p}\de\s\nonumber\\
    & = \int_{\ZZ_p}\frac{1}{|1+\alpha_\kappa\s ^2|_p}\de\s + \int_{\ZZ_p\setminus\{0\}}\frac{1}{|1+\frac{1}{\alpha_\kappa\tau ^2}|_p}\frac{\de\tau}{\left\lvert\alpha_\kappa \tau^2\right\rvert_p}\nonumber\\
    & = \int_{\ZZ_p}\frac{1}{|1+\alpha_\kappa\s ^2|_p}\de\s+\int_{\ZZ_p}\frac{1}{|1+\alpha_\kappa\tau ^2|_p}\de\tau\nonumber\\
    & = 2\int_{D_0(0)}\frac{1}{|1+\alpha_\kappa\s ^2|_p}\de\s = 2,
\end{align}
where we used the change of variable formula for $p$-adic integrals (see Proposition~7.4.1 in~\cite{igusa2000}) in the second equality, the fact that a singleton has zero Haar measure in the compact and infinite (uncountable) group $\ZZ_p$ in the third equality,  and the results in Lemma~\ref{prop.5.1} for $\ZZ_p=D_0(0)$ in the last equality. We perform the same steps when $\alpha_\kappa=-v$, with the only initial difference that $\{\s \in\QQ_p\}=\{\s \in\ZZ_p\}\cup \left\{\s =\frac{1}{v \tau}\midd \tau\in p\ZZ_p\setminus\{0\}\right\}$.
\end{proof}

We can now proceed to show that the two measures $\overline{\mu}_{p,\kappa}$ and $\widetilde{\mu}_2^{(\kappa)}$ do coincide. We need to compare the values in Eq.~\eqref{eq:cardbocgpkk} with
\begin{align}
\widetilde{\mu}_2^{(\kappa)}\big(B_{-n}(\cR_0)\big) = \widetilde{\mu}_2^{(\kappa)}\big(B_{-n}(\mathrm{I})\big) &= \int_{\so(2)_{p,\kappa}}\chi_{\big(B_{-n}(\mathrm{I})\big)}\de \widetilde{\mu}_2^{(\kappa)}\nonumber\\ &= \frac{1}{\mu_2^{(\kappa)}\big(\so(2)_{p,\kappa}\big)}
\int_{\varphi_{(\kappa)}\big(B_{-n}(\mathrm{I})\big)}\frac{1}{|1+\alpha_\kappa\s ^2|_{p}}\de \s ,\label{eq.78}
\end{align}
for every $\cR_0\in \so(2)_{p,\kappa}$ by translation invariance, where, as usual, we denote by $\chi_{\big(B_{-n}(\mathrm{I})\big)}$ the \emph{indicator function} of $B_{-n}(\mathrm{I})$, namely, $\chi_{\big(B_{-n}(\mathrm{I})\big)}(\cR)=1$ for $\cR\in B_{-n}(\mathrm{I})$, $\chi_{\big(B_{-n}(\mathrm{I})\big)}(\cR)=0$ for $\cR\not\in B_{-n}(\mathrm{I})$.
On the other hand, to compute the last integral in Eq.~\eqref{eq.78}, we can exploit the following result.
\begin{lemma}\label{prop.boccoord}
For every $\kappa\in\{-v,p,up\}$ and $n\in\NN$, the image $\varphi_{(\kappa)}\big(B_{-n}(\mathrm{I})\big)=\left\{\s \in\QQ_p\midd \cR_\kappa(\s )\in B_{-n}(\mathrm{I})\right\}$ of a ball $B_{-n}(\mathrm{I})$ through the coordinate map $\varphi_{(\kappa)}$ is 
    \beq\label{eq.boccoord}
\varphi_{(\kappa)}\big(B_{-n}(\mathrm{I})\big)= D_{-n}(0).
\eeq 
\end{lemma}
\begin{proof}
The condition $\cR_\kappa(\s ) \in B_{-n}(\mathrm{I})$ is equivalent to $\left\|\cR_\kappa(\s )-\mathrm{I}\right\|_p\leq p^{-n}$. If $\s \in\ZZ_p$, then $|1+\alpha_\kappa\s ^2|_p=1$ and 
\begin{align*}
\|\cR_\kappa(\s )-\mathrm{I}\|_p &= \left\|\begin{pmatrix}
\frac{1-\alpha_\kappa\s ^2}{1+\alpha_\kappa\s ^2}-1 & -\frac{2\alpha_\kappa\s }{1+\alpha_\kappa\s ^2}\\ \frac{2\s }{1+\alpha_\kappa\s ^2} & \frac{1-\alpha_\kappa\s ^2}{1+\alpha_\kappa\s ^2}-1
\end{pmatrix}\right\|_p \\
&=  \max\left\{\left|\frac{-2\alpha_\kappa\s ^2}{1+\alpha_\kappa\s ^2}\right|_p,\,\left|\frac{-2\alpha_\kappa\s }{1+\alpha_\kappa\s ^2}\right|_p,\,\left|\frac{2\s }{1+\alpha_\kappa\s ^2}\right|_p \right\}\\
& = \max\{\lvert\alpha_\kappa\rvert\lvert\s \rvert^2_p,\,\lvert\alpha_\kappa\rvert\lvert\s |_p,\,|\s |_p\}\\& = \begin{cases}
\max\{p^{-1}|\s |_p^2,p^{-1}|\s |_p,|\s |_p\}=|\s |_p,&\textup{if } \kappa\in\left\{p,up\right\},\\
\max\{|\s |_p^2,|\s |_p\}=|\s |_p,&\textup{if } \kappa=-v.
\end{cases}
\end{align*}
Hence, $\|\mathcal{R}_\kappa(\s )-\mathrm{I}\|_p\leq p^{-n}\textup{ if and only if } |\s |_p\leq p^{-n}$, that is
\beq \label{eq.91}
    \left\{\s \in\ZZ_p\midd \cR_\kappa(\s )\in B_{-n}(\mathrm{I})\right\}=\{\s \in\mathbb{Q}_p\midd|\s |_p\leq p^{-n}\}=D_{-n}(0).
\eeq 
One repeats the procedure for the set $\left\{\s \in\QQ_p\setminus\ZZ_p\midd \cR_\kappa(\s )\in B_{-n}(\mathrm{I})\right\}$, with the change of parameter $\s =-\frac{1}{\alpha_\kappa\tau}$ as in Remark~\ref{rem:paraminter}:
\begin{align}
\|\cR_\kappa(\s )-\mathrm{I}\|_p&=\|-\cR_\kappa(\tau)-\mathrm{I}\|_p = \left\|\begin{pmatrix}
    \frac{-2}{1+\alpha_\kappa\tau^2} & \frac{2\alpha_\kappa\tau}{1+\alpha_\kappa\tau^2}\\
    \frac{-2\tau}{1+\alpha_\kappa\tau^2} & \frac{-2}{1+\alpha_\kappa\tau^2}
    \end{pmatrix}\right\|_p\nonumber\\
    & =\max\left\{|-2|_p,|-2\tau|_p,|2\alpha_\kappa\tau|_p\right\}=1,
\end{align}
however $1\leq p^{-n}$ is impossible for $n\in\NN$, therefore $\left\{\s \in\QQ_p\setminus\ZZ_p\midd \cR_\kappa(\s )\in B_{-n}(\mathrm{I})\right\} =\emptyset$.
\end{proof}

Using Lemma~\ref{prop.boccoord}, normalisation~\eqref{eq:normlafacts}, and  Lemma~\ref{prop.5.1} for $k=-n<0$, the integral~\eqref{eq.78} is easily computed, eventually getting to the following result.

\begin{proposition}\label{prop:equivmeash}For every prime $p>2$ and $\kappa\in\{-v,p,up\}$, the Haar measure $\overline{\mu}_{p,\kappa}$ (cf.~\eqref{eq:defolmu}) coincides with $\widetilde{\mu}_2^{(\kappa)}$ (cf.~\eqref{normhaarmes2})  on $\so(2)_{p,\kappa}$.
\end{proposition}


\bibliographystyle{unsrt}

\begin{thebibliography}{10}

\bibitem{aniello2023}
P. Aniello, S. Mancini, V. Parisi, A $p$-adic model of quantum states and the $p$-adic qubit, Entropy 25 (1) 86 (2023).

\bibitem{our3rd} P. Aniello, S. L'Innocente, S. Mancini, V. Parisi, I. Svampa, A. Winter, Invariant measures on $p$-adic Lie groups: the $p$-adic quaternion algebra and the Haar integral on the $p$-adic rotation groups, 
Lett. Math. Phys. 114(78) (2024). \url{https://doi.org/10.1007/s11005-024-01826-8}

 
\bibitem{metricompladic}
M. F. Atiyah, I. G. MacDonald, Introduction to commutative algebra, Addison-Wesley Series in Mathematics, Addison-Wesley Publishing company, 1969.

\bibitem{Bochner}
S. Bochner, Harmonic analysis and the theory of probability, University of California Press, 1955.

\bibitem{BourInt}
N. Bourbaki, S. K. Berberian, Integration II, Chapters 7-9, Elements of Mathematics, Springer, 2004.

\bibitem{BourSet}
N. Bourbaki, Theory of sets, Elements of Mathematics, Springer, 2004.

\bibitem{BourTop}
N. Bourbaki, General topology, Chapters 1-4, Elements of Mathematics, Springer, 1995.

\bibitem{broughton}
A. Broughton, B. W. Huff, A comment on unions of sigma-fields, Am. Math. Mon. 84 (7) (1977) 553-554.

\bibitem{Cassels}
J. W. S. Cassels, Rational Quadratic Forms, London Mathematical Society Monographs Vol.\ 13, Courier Dover Publications, 2008.

\bibitem{Choksi}
J. R. Choksi, Inverse limits of measure spaces, Proc. Lond. Math. Soc. s3-8 (3) (1958) 321-342.

\bibitem{conradMULvar}
K. Conrad, A multivariable Hensel's lemma, available online at \\\url{https://kconrad.math.uconn.edu/blurbs/gradnumthy/multivarhensel.pdf}.

\bibitem{our1st}
S. Di Martino, S. Mancini, M. Pigliapochi, I. Svampa, A. Winter, Geometry of the $p$-adic special orthogonal group $\so(3)_p$, Lobachevskii J.\ Math. 44 (6)  (2023) 2135-2159.

\bibitem{Dirac}
P. A. M. Dirac, Quantised singularities in the electromagnetic field, Proc. Roy. Soc. London A 133 (1931) 60-72. 

\bibitem{fisher}
B. Fisher, A note on Hensel's lemma in several variables, Proc. Am. Math. Soc. 125 (11) (1997) 3185-3189.

\bibitem{Folland}
G. B. Folland, A course in abstract harmonic analysis, Studies in advanced mathematics Vol.\ 29, CRC Taylor and Francis, 2016.

\bibitem{Folland99}
G.B. Folland, Real analysis, Pure and Applied Mathematics: A Wiley Series of Texts, Monographs, and Tracts, John Wiley, 1999.

\bibitem{profinitem}
M. D. Fried, M. Jarden, Field arithmetic, A series of Modern Surveys in Mathematics 3 Vol.\ 11, Springer, 2005.

\bibitem{Fuchs}
L. Fuchs, Abelian groups, Springer Monographs in Mathematics, Springer, 2015.

\bibitem{Gouvea}
F. Q. Gouv\^ea, {$p$-Adic numbers: an introduction}, Universitext, Springer, 2020.

\bibitem{seminalHaar}
A. Haar, Der Massbegriff in der Theorie der Kontinuierlichen Gruppen, Annals of Mathematics (Second Series) 34 (1) (1933) 147–169.

\bibitem{Halmos}
P. R. Halmos, Measure theory, Graduate Texts in Mathematics Vol.\ 18, Springer, 1974.

\bibitem{hewitt1979}
E. Hewitt, K. A. Ross, Abstract harmonic analysis I, Grundlehren der mathematischen Wissenschaften Vol.\ 115, Springer, 1979.

\bibitem{igusa2000}
J. Igusa, An introduction to the theory of local zeta functions, Studies in Advanced Mathematics Vol.\ 14, American Mathematical Society International Press, 2000.

\bibitem{Lam}
T.Y. Lam, Introduction to quadratic forms over fields, Graduate Studies in Mathematics Vol.\ 67, American Mathematical Society, 2005.

\bibitem{Niklas}
H. Reiter, J. D. Stegeman, Classical harmonic analysis and locally compact groups, London Mathematical Society Monographs New Series Vol.\ 22, Oxford University Press, 2001.

\bibitem{profinite}
L. Ribes, P. Zalesskii, Profinite groups, A Series of Modern Surveys in Mathematics, Springer, 2010.

\bibitem{Rotman}
J. J. Rotman, An introduction to homological algebra, Universitext, Springer, 2008.

\bibitem{serre}
J.-P. Serre, A course in arithmetic, Graduate Texts in Mathematics, Vol.\ 7, Springer, 1973.

\bibitem{serre2}
J.-P. Serre, Galois cohomology, Springer Monographs in Mathematics, Springer, 1994.

\bibitem{our2nd} 
I. Svampa, S. Mancini, A. Winter, An approach to $p$-adic qubits from irreducible representations of $\so(3)_p$, J.\ Math.\ Phys. 63 (7) (2022), 072202.

\bibitem{rooij78}
A.C.M. van Rooij, Non-Archimedean functional analysis, Monographs and textbooks in pure and applied mathematics Vol.\ 51, Marcel Dekker, 1978.

\bibitem{VaradarajanCit}
V. S. Varadarajan, Supersymmetry for mathematicians: an introduction, Courant Lecture Notes 11, American Mathematical Society, 2004.

\bibitem{Varadarajan1}
V. S. Varadarajan, J. T. Virtanen, Structure, classification, and conformal symmetry of elementary particles over non-Archimedean space-time, $P$-Adic Num. Ultrametr. Anal. Appl. 2 (2) (2010) 157-174.

\bibitem{Varadarajan2}
V. S. Varadarajan, Multipliers for the symmetry groups of $p$-adic spacetime, $P$-Adic Numbers, Ultrametr. Anal. Appl. 1 (2009) 69-78.

\bibitem{Volovich2}
V. S. Vladimirov, I. V. Volovich, $p$-Adic quantum mechanics, Commun. Math. Phys. 123 (1989) 659–676.

\bibitem{VVZ} 
V. S. Vladimirov, I. V. Volovich, E. I. Zelenov, $p$-Adic analysis and mathematical physics, Series on Soviet and East European Mathematics Vol.\ 1, World Scientific, 1994.

\bibitem{Volovich1}
I. V. Volovich, $p$-Adic space-time and string theory, Theor. Math. Phys. 71(3) (1987) 574–576.

\bibitem{vonNeumann}
J. von Neumann, Functional operators, vol.\ I: measures and integrals, Annals of Mathematics Studies, Vol.\ 21, Princeton University Press, 1950.




\end{thebibliography}


\end{document}